\documentclass[11pt,a4paper,german]{report}

\usepackage[left=2cm, right=2cm, bottom=2.5cm, top=2.5cm]{geometry}
\usepackage[utf8x]{inputenc}
\usepackage{ae,aecompl}
\usepackage{graphicx}
\usepackage{ucs}
\usepackage{dsfont}
\usepackage{subfigure}
\usepackage{comment}
\usepackage{amsmath}
\usepackage{amsthm}
\usepackage{float}
\usepackage{wallpaper}
\usepackage{setspace}
\usepackage{caption}
\usepackage{youngtab}
\usepackage{mathbbol}
\usepackage{color}
\usepackage[bookmarks, colorlinks, linkcolor=dgreen, urlcolor=navy, citecolor=navy]{hyperref}
\usepackage{amsfonts}
\usepackage{amssymb}

\newtheorem{defn}{Definition}
\newtheorem{thm}[defn]{Theorem}
\newtheorem{prop}[defn]{Proposition}
\newtheorem{cor}[defn]{Corollary}

\newtheorem{lem}[defn]{Lemma}
\newtheorem{qu}[defn]{Question}
\theoremstyle{definition}
\newtheorem{rem}[defn]{Remark}

\definecolor{navy}{rgb}{.05,.15,.5}
\definecolor{dgreen}{rgb}{0,.4,0}

\numberwithin{equation}{chapter}
\numberwithin{defn}{chapter}

\newcommand{\tr}{\mathrm{tr}}
\newcommand{\sgn}{\mathrm{sgn}}
\newcommand{\diag}{\mathrm{diag}}
\newcommand{\mor}{\mathrm{mor}}
\renewcommand{\Re}{\mathrm{Re}}

\newcommand{\Gl}{\mathrm{Gl}}
\newcommand{\Om}{\Omega}
\newcommand{\om}{\omega}
\newcommand{\C}{\mathbb{C}}
\newcommand{\R}{\mathbb{R}}
\newcommand{\N}{\mathbb{N}}
\newcommand{\Z}{\mathbb{Z}}
\newcommand{\Q}{\mathbb{Q}}
\newcommand{\D}{\mathrm{d}\,}
\newcommand{\ket}[1]{\left|#1\right\rangle}
\newcommand{\bra}[1]{\left\langle #1\right|}
\newcommand{\ketbra}[2]{\left|#1\right\rangle\!\!\left\langle #2\right|}

\newcommand{\bracket}[2]{\left\langle #1| #2\right\rangle}
\newcommand{\norm}[1]{\left\|#1\right\|}
\newcommand{\conv}{\mathrm{conv}}
\newcommand{\cone}{\mathrm{cone}}

\newcommand{\im}{\mathrm{im}}
\newcommand{\spa}{\mathrm{span}}
\newcommand{\spec}{\mathrm{spec}}
\newcommand{\Hom}{\mathrm{Hom}}
\newcommand{\End}{\mathrm{End}}
\newcommand{\sh}{\mathrm{sh}}

\author{Christian Majenz}

\title{Constraints on Multipartite Quantum Entropies}

\begin{document}
\begin{titlepage}
\begin{flushright}
\end{flushright}
 \begin{center}
 {\huge\bfseries Constraints on Multipartite Quantum Entropies \\}
 \vspace{1.5cm}
 {\Large\bfseries Christian Majenz}\\[5pt]
 christian.majenz@pluto.uni-freiburg.de\\[14pt]
 \vspace{.6cm}
 Supervisor: Prof. David Gross\\
 \vspace{1.4cm}
 \includegraphics[width=0.38\textwidth]{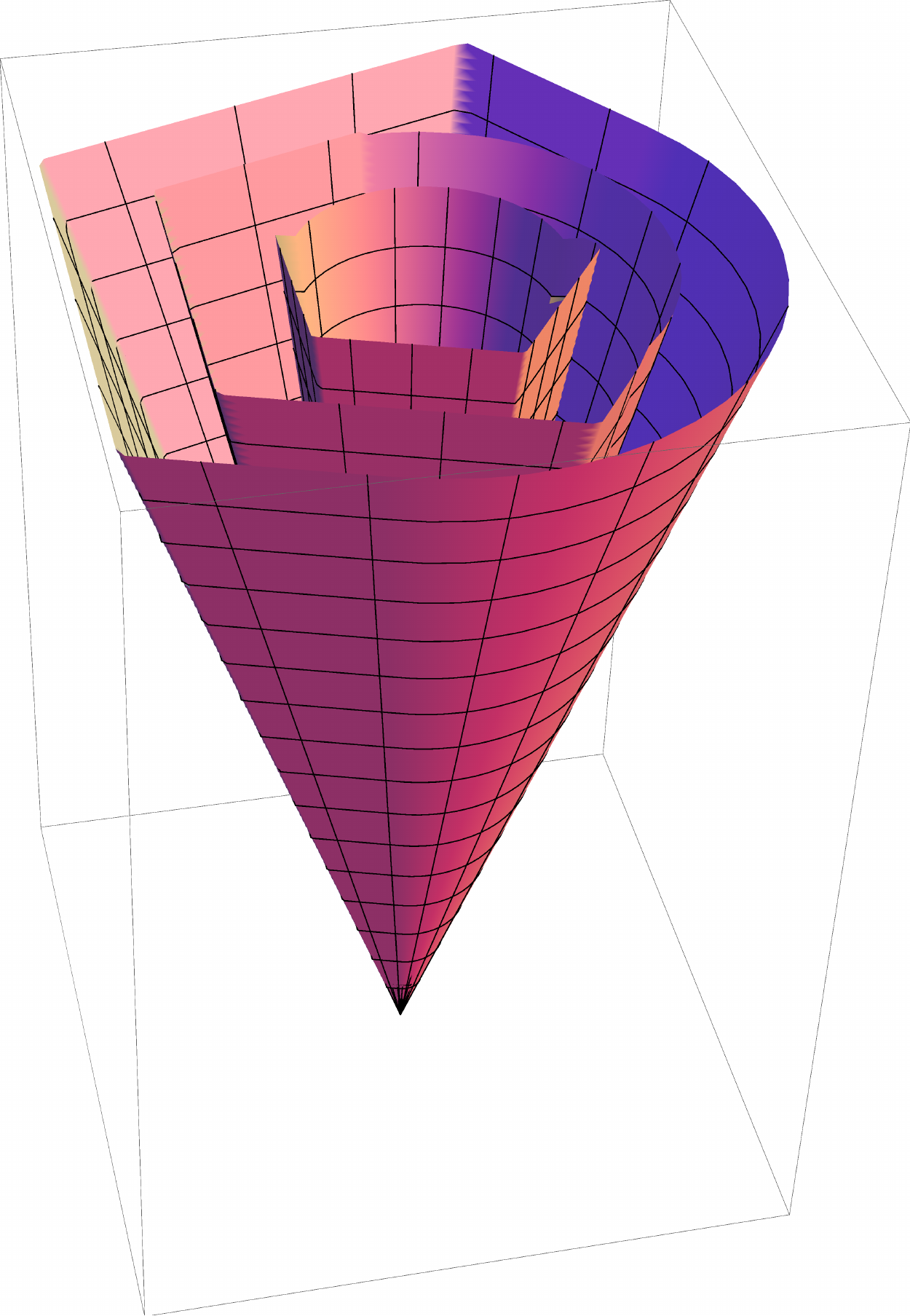}\\[5pt]
\vspace{.6cm}
{Master's Thesis  submitted to} \\[5pt]
\emph{{Albert-Ludwigs-Universität Freiburg}}\\
\vspace{.6cm}

 \includegraphics[width=0.19\textwidth]{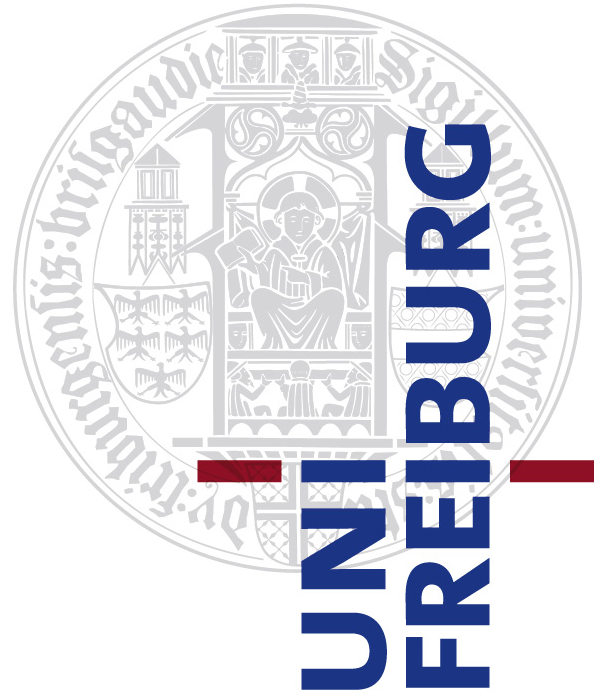}\\[5pt]

{Febuary 2014}
 \end{center}
 
\end{titlepage}

\newpage

\tableofcontents
\newpage


\section*{Abstract}

The von Neumann entropy plays a vital role in quantum information theory. As the Shannon entropy does in classical information theory, the von Neumann entropy determines the capacities of quantum channels. Quantum entropies  of composite quantum systems are important for future quantum network communication their characterization is related to the so called \emph{quantum marginal problem}. Furthermore, they play a role in quantum thermodynamics. In this thesis the set of quantum entropies of multipartite quantum systems is the main object of interest. The problem of characterizing this set is not new -- however, progress has been sparse, indicating that the problem may be considered hard and that new methods might be needed. Here, a variety of different and complementary aprroaches are taken.

First, I look at global properties. It is known that the von Neumann entropy region -- just like its classical counterpart -- forms a \emph{convex cone}. I describe the symmetries of this cone and highlight geometric similarities and differences to the classical entropy cone. 

In a different approach, I utilize the \emph{local} geometric properties of \emph{extremal rays} of a cone.  I show that quantum states whose entropy lies on such an extremal ray of the quantum entropy cone have a very simple structure. 

As the set of all quantum states is very complicated, I look at a simple subset called \emph{stabilizer states}. I improve on previously known results by showing that under a technical condition on the local dimension, entropies of stabilizer states respect an additional class of information inequalities that is valid for random variables from linear codes.

In a last approach I find a representation-theoretic formulation of the classical marginal problem simplifying the comparison with its quantum mechanical counterpart. This novel correspondence yields a simplified formulation of the group characterization of classical entropies (IEEE Trans. Inf. Theory, 48(7):1992–1995, 2002) in purely combinatorial terms.

\newpage
\section*{Zusammenfassung}
Die Von-Neumann-Entropie spielt eine zentrale Rolle in der Quanteninformationstheorie. Wie die Shannonentropie in der klassischen Informationstheorie charakterisiert die Von-Neumann-Entropie die Kapazität von Quantenkanälen. Quantenentropien von Quantenvielteilchensystemen bestimmen die Kommunikationsrate über ein Quantennetzwerk und das problem ihrer Charakterisierung ist verwandt mit dem sogenannten Quantenmarginalproblem. Außerdem spielen sie in der Quantenthermodynamik eine Rolle. In dieser Arbeit liegt das Hauptaugenmerk auf der Menge der Quantenentropien von Quantenzuständen einer bestimmten Teilchenzahl. Das Characterisierungsproblem für diese Menge ist nicht neu -- Fortschritte wurden bisher jedoch nur wenige erzielt, was darauf hinweist, dass das Problem als schwierig bewertet werden kann und dass wahrscheinlich neue Methoden benötigt werden, um einer Lösung näher zu kommen. Hier werden verschieden Herangehensweisen erprobt.

Zuerst gehe ich das Problem aus einer ``globalen Perspektive'' an. Es ist bekannt, dass die Region aller Von-Neumann-Entropien einen konvexen Kegel bildet, genau wie ihr klassisches Gegenstück. Ich beschreibe Symmetrien dieses Kegels und untersuche Gemeinsamkeiten und Unterschiede zum klassischen Entropiekegel. 

Ein komplementärer Ansatz ist die Untersuchung von \emph{lokalen} geoemetrischen Eigenschaften -- Ich zeige, dass Quantenzustände, deren Entropien auf einem Extremstrahl des Quantenentropiekegel liegen, eine sehr einfache Struktur besitzen.

Da die Menge aller Quantenzustände sehr kompliziert ist, schaue ich mir eine einfache Untermenge an: die Menge \emph{Stabilisatorzustände}. Ich verbessere bisher bekannte Ergebnisse, indem ich zeige, dass die Entropien von Stabilisatorzuständen eine zusätzliche Klasse von Ungleichungen erfüllen, die für Zufallsvariablen aus linearen Codes gelten.

Ein vierter Ansatz, den ich betrachte, ist der darstellungstheoretische. Ich formuliere das klassische Marginalproblem in der Sprache der Darstellungstheorie, was den Vergleich mit dem Quantenmarginalproblem vereinfacht. Diese neuartige Verknüpfung ergibt eine vereinfachte kombinatorische Formulierung der Gruppencharakterisierung von klassischen Entropien (IEEE Trans. Inf. Theory, 48(7):1992–1995, 2002).

\newpage 

\section*{Acknowledgements}

First of all I want to thank my supervisor David Gross. Only his support and encouragement as well as our discussions made this thesis possible, our collaboration was a great pleasure. I want to thank Michael Walter for great discussions. Special thanks go to all of the quantum correlations research group at University of Freiburg, which has been a splendid environment for the last year. In particular I want to thank Lukas Luft and Rafael Chaves for sharing their perspective on convex geometry and Shannon entropic inequalities. I want to thank Joe Tresadern for proofreading part of this thesis.

I want to thank my parents, Jaqueline and Klaus Majenz, for their support. Furthermore I want to thank Laura König as well as my housemates for their leniency when I missed some elements of reality every now end then due to their low dimensionality.

I Acknowledge financial support by the German National Academic Foundation.

\newpage
\chapter{Introduction}
\section{Motivation}

The main goal of this thesis is a better understanding of the entropies of multi-particle quantum states. This is an important task from a number of perspectives.

\paragraph*{}First, there is the information theoretic perspective. In both classical and quantum information theory, entropy is a key concept which determines the capacity of a comunication channel \cite{shannon1948mathematical, schumacher1995quantum}. In simple communication scenarios with one sender and one receiver, it suffices to study the entropies of bipartite systems, i.e.\ of two random variables or a bipartite quantum state. Bipartite entropies are well understood in both classical and quantum information theory.

\paragraph{}In a network scenario, however, where data has to be sent from multiple senders to multiple receivers, relations between joint and marginal entropies of multiple random variables determine the constraints on achievable communication rates \cite{yeung2008information}. Although little progress has been made for almost fifty years, in the past fifteen years finally there have been results towards understanding the laws governing the entropies of more than two random variables. In the quantum setting virtually nothing is known. In particular, as the bipartite case shows very strong similarities between quantum and classical entropies, it is promising to search for analogues of the aforementioned recent multivariate classical results.

\paragraph{}In this regard the problem of characterizing the region of possible entropy vectors of multipartite quantum states naturally appears as part of one of the overarching programs in quantum information theoretic research: If possible, find quantum analogues to the results and concepts from classical information theory, otherwise shed light on the differences between the two theories.

Sometimes insights from quantum information theory also have an impact on classical information theory \cite{ozols2013bound}, providing another motivation to study quantum information problems that might be still far from possible applications.

\paragraph*{} Another perspective is that of the \emph{quantum marginal problem}. This is defined more formally in Section \ref{quantmarg}, and can be stated as follows:
Given a multipartite quantum system and some reduced states, is there a global state of that system that is compatible with the given reductions?

A general solution to this problem would have vast implications for quantum physics and quantum information theory. It would, for example, render the task of finding ground states of lattice systems with nearest neighbor interaction \cite{eisert2008gaussian} and the calculation of binding energies and other properties of matter \cite{klyachko2006quantum} computationally tractable. This is unfortunately too optimistic an assumption as the quantum marginal problem turns out to be QMA-complete \cite{liu2006consistency}, as are several specialized variants of practical relevance \cite{liu2007quantum, wei2010interacting}. This is believed to imply that these problems are intractable even for a quantum computer, as QMA is the quantum analogue of the complexity class NP.

\paragraph{}Due to the difficulty of the quantum marginal problem there is little hope for a general solution. But this is not the end of the research program, it is natural to study a ``coarse-grained ''variant: Quantum entropies are functions of the marginals and seem to be amenable to analytic insight.

\paragraph*{}A third motivation to study multi-particle entropies comes from the very field where researchers defined the first entropies, that is from thermodynamics. The strong subadditivity inequality \cite{lieb1973proof} of the von Neumann entropy, for example, has applications in quantum thermodynamics. In one of these applications it is used to prove that the mean entropy of the equilibrium state of an arbitrary quantum system exists in the thermodynamic limit \cite{lanford2003mean, wehrl1978general}, underpinning the correctness of the mathematical formalism used to explicitly take the latter. The application of information theoretic tools in thermodynamics is possible because the respective notions of entropy are mathematically identical and also physically closely related \cite{landauer1961irreversibility, del2011thermodynamic}.

The applications of strong subadditivity suggest that further results in the direction of understanding quantum entropies of multi-particle systems could lead to thermodynamic insights as well.

\section{Goals and Results}

The general program pursued in this thesis -- i.e. understanding multiparticle quantum entropies -- is not new. Several experienced researchers have worked on it before \cite{lieb1973proof, pippenger2003inequalities, linden2005new, christandl2006spectra, christandl2012recoupling, cadney2012infinitely, ibinson2007all, linden2013structure, gross2013stabilizer, linden2013quantum}. Progress, however, has been scarce. In that sense, the problem of finding constraints on quantum entropies can be considered "hard" and it would be too much to ask for anything approaching a complete solution. As a result, we have pursued a variety of very different approaches to the problem in order to gain partial insights. The overall goal of this thesis is to show which approaches could be promising. We are therefore not solving the problem completely, but instead determining which methods may prove useful.
As a consequence, The results obtained in this thesis therefore comprise of a collection of relatively independent insights, rather than being one 'final theorem'. For the benefit of the reader, a list of these individual results are given below.

\paragraph{Global perspective.} A quantum state on an $n$-fold tensor product Hilbert space gives rise to $2^n$ entropies, one for each subset of subsystems. Collecting them in a real vector yields a point in the high-dimensional vector space $\R^{2^n}$. It turns out, that the set of all such \emph{entropy vectors} forms a \emph{convex cone} \cite{pippenger2003inequalities}. The same is known to be true for the classical entropy region defined analogously \cite{zhang1997non}. In Chapter \ref{sec:mor} some global properties of this geometric object are investigated:
\begin{itemize}
 \item Proposition \ref{Gammasym} shows that the quantum entropy cone has a symmetry group that is strictly larger than the known symmetry group of its classical analogue.
 \item Corollary \ref{wmfacets} uses this symmetry to show that some known quantum information inequalities define \emph{facets} of the quantum entropy cone, i.e. they are independent from all other (known and unknown) quantum information inequalities.
 \item The classical entropy cone is known to have  the property that all interesting information inequalities satisfy a number of linear relations \cite{chan2003balanced}. Such information inequalities are called \emph{balanced}. Corollary \ref{sigmastardirsum} and the preceding discussion clarify the geometric property underlying this result: The \emph{dual} of the quantum entropy cone has a certain direct sum structure. Theorem \ref{nodirsum} proves a characterization of cones whose \emph{duals} have this structure. Corollary \ref{gammanodirsum} uses this theorem and the facets identified in Corollary \ref{wmfacets} to show that the quantum entropy cone does not have this simpler structure and that therefore the result from \cite{chan2003balanced} does not have a straightforward quantum analogue.
\end{itemize}

\paragraph{Local perspective.} The most important points of a convex set are its \emph{extremal points}. Here we study the local geometry of \emph{extremal rays}, which are the cone analogues of extremal points. In particular, we characterize quantum states that have an entropy vector that lies on such an extremal ray.
\begin{itemize}
 \item Theorem \ref{poprays} proves that all non-trivial states whose entropy vectors lie on an edge of the quantum entropy cone have the property that all marginal spectra are \emph{flat}, i.e. that the reduced states have only one distinct nonzero eigenvalue. This is a very simple structure and narrows down the search for states on extremal rays tremendously.
 \item Theorem \ref{classpoprays} provides an analogous result for the classical entropy cone. 
\end{itemize}

\paragraph{Variational perspective.} As the characterization of the whole quantum entropy cone for $n\ge4$ parties has so far proved elusive, and even the classical entropy cone is far from characterized in this case, it seems reasonable to start by finding simpler inner approximations. This can also be done by looking at a subset of states that has additional structure such as to allow for a direct algebraic characterization of the possible entropy vectors.

One subset that allows for such an algebraic approach is the set of \emph{stabilizer states} \cite{linden2013quantum, gross2013stabilizer}. In Chapter \ref{stabs} the results from \cite{gross2013stabilizer} and \cite{linden2013quantum} are improved: 
\begin{itemize}
\item Corollary \ref{stabsqfree} states that under a technical assumption on the local Hilbert space dimension, entropies from stabilizer states satisfy an additional class of linear inequalities that governs the behavior of linear network codes, the \emph{linear rank inequalities}. The result includes the important qubit case. This partially answers a question raised in \cite{linden2013quantum} and shows that stabilizer codes behave similar to classical linear codes from an entropic point of view. 
\end{itemize}

\paragraph{Structural perspective.} For the classical entropy cone \cite{chan2002relation} provides a remarkable characterization result: for a given entropy vector there exists a group and subgroups thereof such that the entropies are determined by the relative sizes of the subgroups. Considerable research effort has been directed at finding an analogous relation for quantum states \cite{christandl2006spectra, christandl2012recoupling}. Chapter \ref{CYR} is concerned with the result from \cite{chan2002relation} and possible quantum analogues:
\begin{itemize}
 \item Theorem \ref{grouplesscy} recasts the main result from \cite{chan2002relation} purely in terms of certain combinatorial objects known from mathematical statistics called \emph{type classes}. This characterization is simpler in the sense that type classes are much simpler objects than finite groups.
 \item Section \ref{strings} provides a connection between strings and certain representations of the symmetric group called \emph{permutation modules}.
 \item This formalism allows for representation-theoretic proofs of the Shannon-type information inequalities such as the strong subadditivity for the Shannon entropy (Proposition \ref{repssa}).
 \item Section \ref{repcy} gives a novel decomposition of $\left(\C^d\right)$ into a direct sum of permutation representations of the direct product of symmetric groups $S_n\times S_d$. This allows for a simple argument why Theorem \ref{grouplesscy} does not have a direct quantum analogue.
 \item Theorem \ref{UtoS} gives a formula for the decomposition of \emph{Weyl modules} restricted from the unitary to the symmetric group into irreducible representation of the latter as a byproduct.
\end{itemize}

\section{Overview}

After this introduction, there are two chapters devoted to introducing the mathematical and the physical and information theoretical fundamentals respectively. In Chapter \ref{math} the relevant mathematical background is discussed, that is convex geometry, Lie groups and Lie algebras, the group algebra and representation theory. The section about representation theory is somewhat longer, as deeper results from that field are used in Chapters \ref{stabs} and \ref{CYR}, where as the other sections are mostly dedicated to introducing the concepts and fixing a notation. Chapter \ref{info} contains an introduction to the information theoretical and some physical concepts, i.e.\ classical information theory, quantum information theory and some concepts from quantum mechanics. In this chapter the classical and quantum entropy cones are introduced that are the main objects of study in this thesis.

Chapter \ref{sec:mor} is concerned with the convex geometry of entropy cones. Section \ref{sym} clarifies the symmetries of the quantum entropy cone. 
Section \ref{bal} is concerned with investigating the possibility of generalizing a result from classical information theory \cite{chan2003balanced}.
The last short section in this chapter, Section \ref{sub}, reviews a class of maps between entropy cones of different dimensions introduced in \cite{ibinson2008quantum} and presents them in a more accessible way using the cone morphism formalism.

Chapter \ref{differential} makes a complementary approach to characterizing the quantum entropy cone: While Chapter \ref{sec:mor} investigates global properties by looking at symmetry operations, this Chapter is concerned with the local geometry of extremal rays.
In Section \ref{classdiff} the classical case is investigated with the techniques developed for the quantum case.

Chapter \ref{stabs} introduces the set of stabiliser states and their description by a finite phase space. The independently obtained result in \cite{linden2013quantum} and \cite{gross2013stabilizer} is also strengthened by partly answering a question posed in \cite{gross2013stabilizer}.

Chapter \ref{CYR}, is concerned with the representation theoretic point of view on the quantum marginal problem and quantum information inequalities introduced in \cite{christandl2006spectra}. A The classical result that inspired the research in this direction, \cite{chan2002relation}, is reviewed and reformulated in a more information theoretic way using type classes. 

A quantum analogue of the construction is attempted, but only succeeds for the trivial case $n=1$.

\section{Conventions}
The following conventions and notations are used in this thesis:
\begin{itemize}
	\item $\N=\left\{0,1,2,3,...\right\}$ is the set of integers including zero
	\item $\log(x)$ is the logarithm with basis two, otherwise the basis is specified as in $\log_{10}(x)$, $\ln$ is the natural logarithm
	\item $\R_+=\{x\in\R|x>0\},\ \R_{\ge0}=\{x\in\R|x\ge 0\}$
	\item In a topological space, given a set $A$ I denote its closure by $\overline{A}$.
	\item For a subset $A\subset M$ the complement of a is denoted by $A^c=M\setminus A$. If $A=\{a\}$ is a singleton, I write $a^c$ instead of $A^c$.
	\item $A:=B$ or $B=:A$ means ``define $A$ to be equal to $B$'', $A=B$ means ``$A$ is equal to $B$''
\end{itemize}


\chapter{Mathematical Background}\label{math}
\section{Convex Geometry}\label{conv}

As already mentioned in the introduction, convex geometry plays an important role in the characterization efforts for classical and quantum joint entropies. In particular, the notion of a convex cone is important when investigating joint entropies, as both the set of Shannon entropy vectors of all n-partite probability distributions and the set of von Neumann entropy vectors of all n-partite quantum states, which I will define in Sections \ref{classical} and \ref{quantum} respectively, can be proven to form convex cones up to topological closure. In this chapter I introduce some basic notions of convex geometry. A more careful introduction can be found for example in \cite{barvinok2002course}.

Roughly speaking, a body is \emph{convex}, if it has neither dents nor holes. Mathematically, let us make the following
\begin{defn}[Convex Set]
	Let $V$ be a real vector space. A subset $B\subset V$ is called convex, if
	\begin{equation}
		\, \forall x,y\in B:\, \forall \lambda \in [0,1]: \lambda x+(1-\lambda)y\in B.
	\end{equation}
\end{defn}
The concept that is most important for this thesis among the ones introduced in this section is the (convex) cone. We define a cone to be a convex set that invariant under positive scaling, i.e.
\begin{defn}[Cone]
	Let $V$ be a real Vector space. A convex subset $C\subset V$ is a cone, if
	\begin{equation}
		\, \forall x\in C, \, \forall \lambda \in \R_{+}: \lambda x\in C
	\end{equation}

\end{defn}

For an arbitrary subset $A\subset V$ we define the \emph{convex hull} 
\begin{equation}
	\conv (A)=\left\{\lambda x+(1-\lambda)y|x,y\in A, \lambda \in [0,1]\right\}
\end{equation}
 as the smallest convex set that contains the original one and analogously the conic hull $\cone(A)=\R_{\ge0}\conv(A)$.
Simple examples of cones are the open and the closed quadrants in $\R^2$, the open and the closed octants in $\R^3$ or the eponymous one, $C_{\bigcirc}=\left\{(x,y,z)\in\R^3\Big|x^2+y^2-z^2\le0\right\}$ shown in Figure \ref{cone}.

A \emph{face} of a convex set is, roughly speaking, a flat part of its boundary, or, mathematically precisely put,
\begin{defn}[Face]
	Let $V$ be a real vector space and $A\subset V$ convex. A face of $A$ is a subset $F\subset \bar A$ of its \emph{closure} such that there exists a linear functional $f:V\to\R$ and a number $\alpha\in\R$ with
	\begin{equation}
		F=\bar A\cap \left\{v\in V|f(v)= \alpha\right\}\text{ and } A\cap \left\{v\in V|f(v)< \alpha\right\}=\emptyset.
	\end{equation}
If the face $F=\{v_0\}$ is a singleton, $v_0$ is called an \emph{exposed point}. A face $F$ is called \emph{proper} if $\emptyset\not=F\not=A$. If there is no proper face that contains a face $F$ except for $F$ itself, we call $F$ a \emph{facet}.
\end{defn}
Note that for a proper face of a cone one always has $\alpha=0$.

A \emph{base} of a cone $K$ is a minimal convex set $B\subset K$ that generates $K$ upon multiplication with $\R_{\ge 0}$, i.e.
\begin{defn}[Base]
	Let $K$ be a cone. A base of $K$ is a convex set $B\subset K$ such that for each $v\in K$ there exist unique $b\in B$ and $\lambda \in \R_{\ge 0}$ with $v=\lambda b$.
\end{defn}
A \emph{ray} is a set of the form $\R_{\ge 0}v$ for a vector $v\in\R^n$. A cone contains each ray that is generated by one of it's elements, and there is a natural bijection between the set of rays and any base. That motivates the definition of \emph{extremal rays}, which correspond to extremal points of any base: 
\begin{defn}[Extremal Ray]
	Let $K\subset\R^n$ be a convex cone. A ray $R\subset K$ is called extremal, if for each $v\in R$ and each sum decomposition $v=x+y$ with $x,y\in K$ we have $x,y\in R$. We denote the set of extremal rays of $K$ by $\mathrm{ext}(K)$.
\end{defn}
In convex geometry \emph{duality} is an important concept. Instead of describing which points are in a convex set, one can give the set of affine inequalities that are fulfilled by all points in the convex set. By inequality we always mean statements involving the non-strict relations $\le$ and $\ge$. If a certain affine inequality is valid for all elements of a cone, then also its homogeneous, i.e.\ linear, version holds. By fixing the exclusive usage of either $\ge$ or $\le$, any linear inequality on a vector space $V$ can be described by an element of the \emph{dual space} $V^*$. We adopt the convention to use $\ge$ and define the
\begin{defn}[Dual Cone]
	Let $V$ be a real vector space and $V^*$ its dual space. Let $K\subset V$ be a convex cone. The dual cone is defined by
	\begin{equation}
		K^*=\left\{x\in V^*|x(y)\ge 0\, \forall y\in K\right\}.
	\end{equation}
\end{defn}
If $\dim V<\infty$ we have $V\cong \R^n$ and hence $V^*\cong V$ via the standard inner product in $\R^n$. The extremal rays of the dual cone are exactly the ones corresponding to facets.

Convex sets come n different shapes, e.g.\ a circle is convex as well as a triangle. An important difference between the two is that the latter is described by finitely many faces or finitely many extremal points.
\begin{defn}[Polyhedron, Polytope]
	Let $V$ be a real vector space. A convex set $B\subset V$ is called a \emph{polyhedron}, if it is the intersection of finitely many halfspaces, i.e.\ there exist a finite number $k$ of functionals $f_1, ..., f_k \subset V^*$ and a real number $\alpha_j,\ j=1,...,k$ for each functional such that
	\begin{equation}
		B=\left\{v\in V|f_j(v)\ge \alpha_j,\ j=1,...,k\right\}.
	\end{equation}
	$B$ is, in addition, compact, it is called a \emph{polytope}.
\end{defn}
We call a cone \emph{polyhedral}, if it is a polyhedron.

As the classical and quantum entropy cones are not exactly cones but only after topological closure, I reproduce a characterization result here for such sets. Adopting the notions in \cite{pippenger2003inequalities}, we say a subset $A\subset V$ of a vector space $V$ with $\dim V<\infty$ is \emph{additive}, if $\, \forall x,y\in A: x+y\in A$, and a set is said to be \emph{approximately diluable} if for all $\epsilon>0$ there exists a $\delta >0$, such that for all $x\in A,  0\le \lambda \le \delta$ there exists a $ y\in A$ such that $\left\|\lambda x-y\right\|<\epsilon$. Note that, as we are talking about finite dimensional vector spaces, all norms are equivalent, so we do not have to specify. It turns out that a set that is additive and approximately diluable turns into a cone after taking the closure:
\begin{prop}[\cite{pippenger2003inequalities}]
	Let $V$ be a real vector space and $A\subset V$ additive and approximately diluable. Then $\overline{A}$ is a convex cone.
\end{prop}

\begin{figure}
\centering
 \includegraphics[width=7cm]{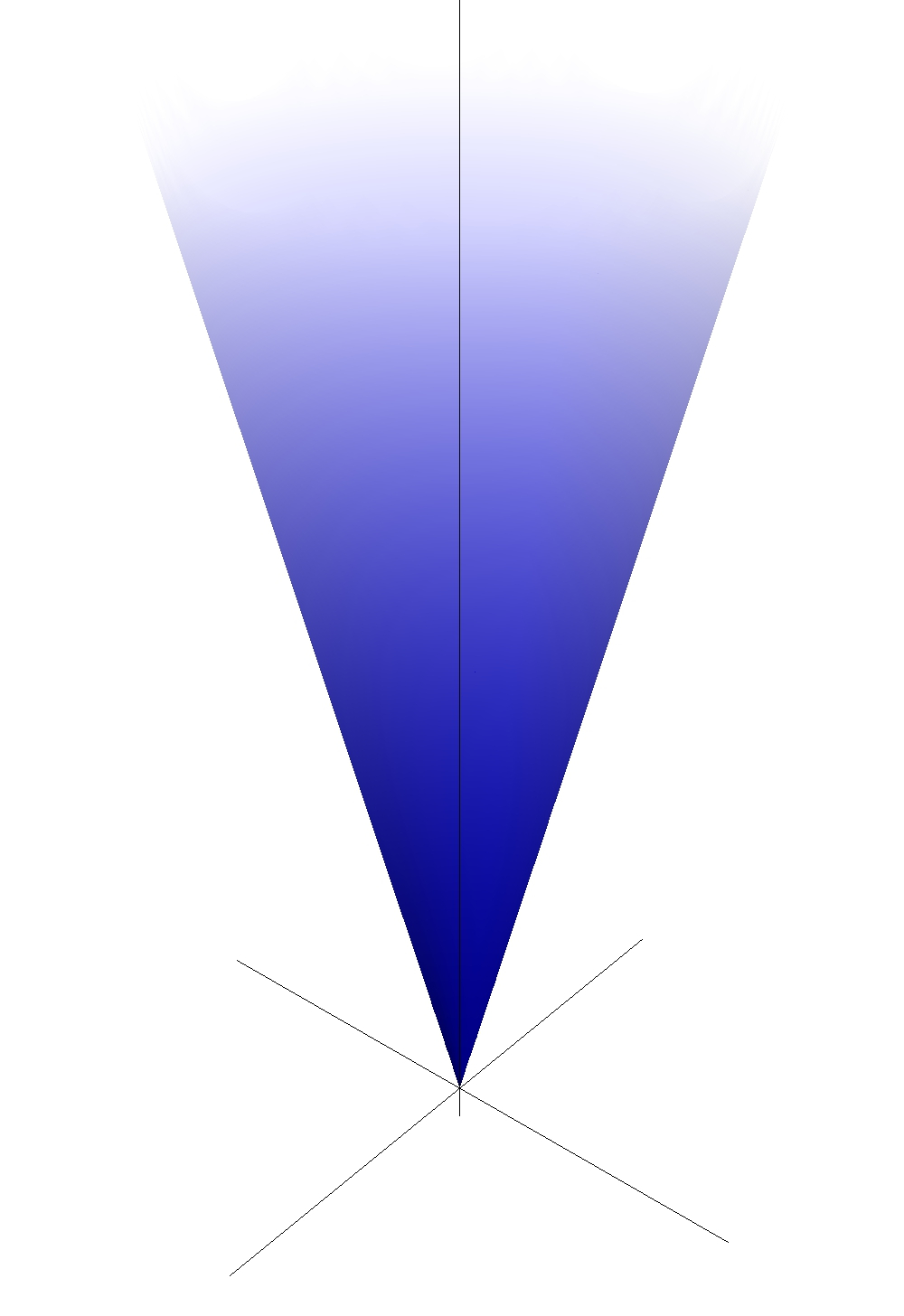}
 \caption{A cone.}\label{cone}
\end{figure}

To investigate relations between different cones and to find their symmetries, we would like to introduce a class of maps between vector spaces containing cones that preserves their structure. The set of maps will be a subset of the homomorphisms of the ambient vector spaces that map cone points to cone points, i.e.\ 
\begin{defn}[Cone Morphism]
	Let $V_1$, $V_2$ be real vector spaces, $K_1\subset V_1$, $K_2\subset V_2$ cones. A map $\phi \in \hom(V_1,V_2)$ is called cone morphism, if
	\begin{equation}
		\phi(K_1)\subset K_2.
	\end{equation}
	If $\phi|_{K_1}$ is injective and $\phi(K_1)=K_2$, $\phi$ is called a \emph{cone isomorphism}, in this case $K_1$ and $K_2$ are called \emph{isomorphic}. The set of all such cone morphisms is denoted by $\mor(K_1,K_2)$.
\end{defn}
\begin{rem}
	The notion of a cone isomorphism introduced here coincides with the notion of an \emph{order isomorphism} in the theory of ordered vector spaces.
\end{rem}

Note that a cone isomorphism is not always a vector space isomorphism. However, if $\dim K_i=\dim V_i$, then a cone isomorphism $\phi$ is also a vector space isomorphism. A cone homomorphism $\phi\in \mor(K_1,K_2)$ naturally induces a cone homomorphism $\phi^\dagger\in\mor\left(K_2^*,K_1^*\right)$ by pulling back functionals via $\phi$, i.e.\ $\phi^\dagger(f)=f\circ\phi$. If we look at an arbitrary linear map $L: V_1\to V_2$ we get a new cone in $V_2$ from a cone $C\subset V_1$, that is $C'=L(C)$. Can we express $C'^*$ by $C^*$ and $L$? We calculate
\begin{eqnarray}
	L(C)^*&=&\left\{f\in V_2^*|f(L(x))\ge 0\, \forall x\in C\right\}=\left\{f\in V_2^*|(L^\dagger f)(x)\ge 0\, \forall x\in C\right\}\nonumber\\
	&=&\left\{f\in V_2^*|L^\dagger f\in C^*\right\}=(L^\dagger)^{-1}C^*
\end{eqnarray}
where $(L^\dagger)^{-1}$ is the set-valued inverse of the adjoint of $L$.


\section{Groups and Group Algebras}

The Following chapter is dedicated to a concise introduction to Lie groups, and group algebra, as they may be not familiar to all readers and also to fixing a notation for the subsequent chapters. A reference for a more extensive introduction that is still focused on representation theory is \cite{fulton1991representation}.


\subsection{Lie Groups}\label{lie}
A Lie group is, roughly speaking, a Group that also is a $\mathcal{C}^\infty$-manifold and in which the group structure is smooth with respect to differentiation on the manifold. Recall the definition of a
\begin{defn}[Manifold]
	A $\mathcal{C}^\infty$-manifold is a topological space $M$ (Hausdorff, paracompact) with the following properties:
	\begin{enumerate}
		\item[(i)] There exists a \emph{dimension} $n\in\N$ such that for all $x\in M$ there is an open neighborhood $U\subset M$ of $x$ and a homeomorphism $\phi: U\to \R^n$ called \emph{chart}.
		\item[(ii)] For two such maps, $\phi_1: U_1\to \R^n$ and $\phi_2: U_2\to \R^n$ with $U_1\cap U_2\not=\emptyset$, the map $\phi_2\circ\phi_1^{-1}: \phi_1(U_1\cap U_2)\to \R^n$ is $\mathcal{C}^\infty$ or \emph{smooth}.
	\end{enumerate}
\end{defn}
In this text all manifolds are $\mathcal{C}^\infty$. A map $f: M\to N$ between manifolds is called smooth, if $\phi_N\circ f \circ \phi_M^{-1}$ is smooth, where $\phi_M (\phi_N)$ are charts on $M$($N$) respectively.  With this in mind we can go forward and define a
\begin{defn}[Lie Group]
	A Lie Group is a group $(G,\cdot)$ with the additional property that $G$ is a manifold and the maps $\cdot\, :G\times G\to G$ and $\mathrm{inv}: G\to G, g\mapsto g^{-1}$ are smooth. Note that the $G\times G$ is equipped with the obvious manifold structure.
\end{defn}
Lie groups can be characterized by manifold properties such as connected, simply connected or compact, and by group properties such as simple or Abelian. A manifold can be a complicated object, but we can always map its local properties to its \textit{tangent space}, the same can be done with the group structure of a Lie group. This motivates the definition of a
\begin{defn}[Lie Algebra]
	A Lie algebra $\mathfrak{g}$ is a vector space over a field with characteristic $0$, with a bilinear map $[\cdot,\cdot]:\mathfrak{g}\times \mathfrak{g}\to \mathfrak{g}$ called \emph{Lie bracket}, which fulfills the following properties:
	\begin{enumerate}
		\item[(i)] it is alternating, that is $[X,Y]=-[Y,X]\, \forall X,Y\in\mathfrak{g}$
		\item[(ii)] it fulfills the Jacobi identity $[X,[Y,Z]]+[Y,[Z,X]]+[Z,[X,Y]]=0$
	\end{enumerate}
\end{defn}
A representation of a Lie algebra $\mathfrak{g}$ is vector space homomorphism $\phi: \mathfrak{g}\to \C^{n\times n}$ which maps the Lie bracket to the commutator:
\begin{equation}
	\phi([X,Y])=[\phi(X),\phi(Y)]
\end{equation}

The tangent space of a Lie group $G$ has a natural bilinear map of this form. To construct it we write down the conjugation map
\begin{equation}
	\Psi(g,h)=ghg^{-1}
\end{equation}
and differentiate in both arguments at $g=h=\mathds{1}$. The resulting bilinear map makes the tangent space a Lie algebra, as can easily be checked for the special case of $G$ being a subgroup of $Gl(n,\C)$, the only case we will be dealing with. In that case the Lie bracket is the commutator.

The important fact about the Lie algebra of a Lie group is that it contains the essential part of the group structure in the sense that each representation of a Lie group $G$ defines a representation of its Lie algebra, and each representation of its Lie algebra defines a representation of the \emph{universal cover} of the connected component of the Identity. Usually the Lie groups are real manifolds, and therefore they have real Lie algebras, but often it is simpler to have a complex algebra, especially because $\C$ is algebraically closed. A helpful fact is that, given a Lie algebra $\mathfrak{g}$, the representations of the complexified Lie algebra $\mathfrak{g}_\C$ are irreducible if and only if the 
corresponding 
representation of the real Lie algebra is irreducible. 

The connection between Lie group and Lie algebra is even more explicit. An element $X$ of the Lie Algebra $\mathfrak{g}$ generates a one parameter subgroup of $G$: we just find a smooth curve $\gamma:[0,1]\to G$ with $\gamma(0)=\mathds{1}$ and $\dot{\gamma}(0)=X$ and define
\begin{equation}
	e^{sX}=\lim_{n\to\infty}\gamma\left(s\frac{1}{n}\right)^n.
\end{equation}
In matrix Lie groups/algebras this coincides with the matrix exponential.

An Important representation of a Lie algebra is the adjoint representation. A Lie algebra $\mathfrak{g}$ acts on itself by means of the bracket, i.e.\
\begin{equation}\label{ad}
	ad: \mathfrak{g}\to \mathfrak{gl}(\mathfrak{g}),\ \ ad(X)Y=[X,Y].
\end{equation}
The Jacobi identity ensures that the bracket is preserved under this vector space homomorphism, it thus really is a representation of $\mathfrak{g}$


\subsection{Group Algebras}
Let $G$ be a finite group and $\mathcal{A}(G)=\C G$ the free complex vector space over $G$. Then $\mathcal{A}(G)$ inherits the multiplication law from $G$ which makes it an associative unital algebra:
\begin{equation}
	\alpha\cdot\beta=\left(\sum_{g\in G}\alpha_g g\right)\left(\sum_{g\in G}\beta_g g\right)=\sum_{g, h\in G}\alpha_g \beta_h gh=\sum_{g\in G}\left(\sum_{h\in G}\alpha_{gh^-1}\beta_{h}\right)g
\end{equation}
$\mathcal{A}(G)$ is usually equipped with the standard inner product of $\C^{|G|}$ rescaled by the size of $G$:
\begin{equation}\label{galgip}
	(\alpha,\beta)=\frac{1}{|G|}\sum\alpha_g^* \beta_g
\end{equation}
The free complex vector space $\C M$ over any set $M$ is nothing else but the vector space of complex functions on $M$, so we can also view elements of the group algebra $\mathcal{A}(G)$ as complex functions on $G$. A \emph{projection} in $\mathcal{A}(G)$ is an element $0\not=p\in\mathcal{A}(G)$ with $p^2=p$. A projection is called \emph{minimal}, if it can not be decomposed into a sum of two projections.
The concept of a group algebra generalizes in a straightforward way to compact Lie groups, where the sums over $G$ have to replaced by integrals with respect to the invariant Haar measure that assigns the volume 1 to $G$.



\section{Representation Theory}
The basic results stated in this section can be found in textbooks like \cite{goodman1998representations} and \cite{fulton1991representation}. Given a group $G$ we can investigate homomorphisms $\phi: G\to \mathrm{Gl}(n,\mathds{k})$ to the general linear group of the $n$-dimensional Vector space over some field $\mathds{k}$ which are called \emph{representations} of $G$. More generally we write $\phi: G\to \mathrm{Gl}(V)$ for a representation on an arbitrary vector space $V$. The vector space which the group acts on is called representation space. The representation is called complex (real) representation if $\mathds{k}=\mathbb{C}$ ($\mathds{k}=\mathbb{R}$). In the context of quantum information theory we are almost exclusively concerned with complex representation, as quantum mechanics take place in a complex Hilbert space (Although Asher Peres once said that ``...quantum phenomena do not occur in a Hilbert space, they occur in a laboratory.'' \cite{nielsen2010quantum}, page 112). Any representation of a 
group $G$ can be extended by linearity to a representation of the group algebra $\mathcal{A}(G)$. Two representations $\phi_
1$ and $\phi_2$ are considered equivalent if there exists a vector space isomorphism $\psi: \C^n\to\C^n$ which acts as an intertwiner for the two representations:
\begin{equation}
	\phi_1\circ \psi=\phi_2
\end{equation}
A representation $\phi$ on $\C^n$ is called irreducible if it has no non-trivial proper invariant subspaces, otherwise it is called reducible. A representation $\phi$ on $\C^n$ is called completely reducible if $\C^n=\bigoplus_iV_i$, $V_i$ are invariant subspaces and $\phi|_{\Gl (V_i)}$ is irreducible. All representations of finite Groups are completely reducible. Also this result, which is built on the possibility of averaging over the group, generalizes to compact Lie groups. In the sequel we do not always distinguish between a representation and its representation space. Given a group $G$ and a representation $\phi: G\to \mathrm{Gl}(V)$ we say $V$ is a representation of $G$ and write $gu=v$ if $\phi(g)u=v$. An important tool in representation theory is Schur's lemma which characterizes the homomorphisms between two representations that commute with the action of th group:
\begin{lem}[Schur's Lemma]\label{schur}
	Let $V$ and $W$ be irreducible representations of a Group $G$, and let $\phi: V\to W$ be a vector space homomorphism that commutes with the action of $G$, i.e.
	\begin{equation}
		\phi(g v)=g\phi(v) \, \forall g\in G, v\in V.
	\end{equation}
	Then either $\phi=0$ or $\phi$ is an isomorphism. In particular, if $V=W$ then $\phi=\lambda \mathds{1}_V$ for some $\lambda\in\C$.
\end{lem}
\begin{proof}
	Observe that if $\phi(v)=0$, then $\phi(gv)=g 0=0$, i.e.\ $\ker\phi$ is an invariant subspace of $V$ which can, by the irreducibility of $V$, only be zero or $V$. This proves that $\phi$ is either $0$ or injective. Also $\im\phi$ is invariant, as $g\phi(v)=\phi(g v)\in\im\phi$. This shows that $\phi$ is surjective, unless it is $0$. We conclude that $\phi$ is either $0$ or an isomorphism. For $V=W$, $\phi$ is an endomorphism of a vector space over the algebraically closed field $\C$, so it has an eigenvalue $\lambda$. Hence $\ker\left(\phi-\lambda \mathds{1}\right)\not=0$ and $\phi-\lambda \mathds{1}$ commutes with the action of $G$, so by the first part of this proof $\phi-\lambda \mathds{1}=0$
\end{proof}

As an important corollary of this lemma, we find the multiplicity of an irreducible representation $V$ of a group $G$ in some representation $W$ being equal to the dimension of the space of $G$-invariant homomorphisms i.e.\ of the space
\begin{equation}
	\mathrm{Hom}^G(W,V)=\left\{\phi\in\mathrm{Hom}(W,V)\big|\phi(gv)=g\phi(v)\, \forall g\in G, v\in W\right\}
\end{equation}

\begin{cor}\label{dimhom}
	Let $W=\bigoplus_\alpha V_\alpha^{\oplus m_\alpha}$ be a representation of a finite group $G$. Then
	\begin{equation}
		m_\alpha=\dim\mathrm{Hom}^G(W,V_\alpha)
	\end{equation}
\end{cor}
\begin{proof}
	Let $\phi$ be any element from $\mathrm{Hom}^G(W,V_\alpha)$. Then $\phi$ has the form
	\begin{equation}
		\phi=\bigoplus_\beta\bigoplus_{i=1}^{m_\beta}\phi_{\beta i}, \text{ with }\phi_{\beta i}: V_\beta\to V_\alpha.
	\end{equation}
	According to Schur's lemma (Lemma \ref{schur})
	\begin{equation}
	\phi_{\beta i}=\begin{cases}
	               	\lambda_i\mathds{1}_{V_\alpha}& \beta=\alpha\\ 0& \text{else}
	               \end{cases},
	\end{equation}
	with $\lambda_i\in\C$. Thus we have an obvious isomorphism
	\begin{equation}
		\mathrm{Hom}^G(W,V_\alpha)\stackrel{\sim}{\longrightarrow}\C^{m_\alpha}
	\end{equation}
	and the statement follows.
\end{proof}
The unitary representations of a finite group $G$, somewhat surprisingly, provide us with an orthonormal basis of the group algebra. Here we prove a first part of this fact:
\begin{thm}[Schur Orthogonality Relations, Part I]\label{schurortho1}
	Let $G$ be a finite Group. Let $\alpha$ label the equivalence classes of irreducible representations of $G$ and pick a unitary representative $U_\alpha: G\to U(V_\alpha)$ from each class. Then
	\begin{equation}\label{ureporth}
		(U_{\alpha i j},U_{\beta k l})=\delta_{\alpha\beta}\delta_{ik}\delta_{jl}\frac{1}{d_\alpha},
	\end{equation}
	where $(\cdot,\cdot)$ is the inner product of the group algebra \eqref{galgip} and $d_\alpha=\dim(V_\alpha)$. 
\end{thm}
\begin{proof}
	for two fixed unitary irreducible representations $U_\alpha, U_\beta$ we define for each map $A\in\Hom(V_\alpha,V_\beta)$ an associated element $A^\sharp\in \Hom^G(V_\alpha,V_\beta)$ by
	\begin{equation}
		A^\sharp=\frac{1}{|G|}\sum_{g\in G}U_\beta(g^{-1})AU_{\alpha}(g).
	\end{equation}
	If now $E^{(ij)}$ is the standard basis of the space of $d_\alpha\times d_\beta$-matrices $\C^{d_\alpha\times d_\beta}\cong \Hom(V_\alpha,V_\beta)$, i.e.\ $E^{(ij)}_kl=\delta_{ik}\delta_{lj}$, then
	\begin{equation}
		\left(E^{(ij)}\right)^\sharp_{kl}=\left(U_{\alpha ik},U_{\beta jl}\right)
	\end{equation}
	According to Schur's lemma (Lemma \ref{schur}) we have $\Hom^G(V_\alpha, V_\beta)=\C\delta_{\alpha\beta}\mathds{1}_{V_\alpha}$. So with the above equation we already get $\left(U_{\alpha ik},U_{\beta jl}\right)=0 $ if $\alpha\not=\beta$.
	
	For $\alpha=\beta$ we have $\delta_{ij}=\tr E^{(ij)}=\tr \left(E^{(ij)}\right)^\sharp=\tr(\lambda \mathds{1}_{V_\alpha})=d_\alpha\lambda$, hence $\lambda=\delta_{ij}\frac{1}{d_\alpha}$ and thus
	\begin{equation}
		\left(U_{\alpha ik},U_{\alpha jl}\right)=\delta_{ij}\delta_{kl}\frac{1}{d_\alpha}.
	\end{equation}
	This proves \eqref{ureporth}. 
\end{proof}

\subsection{Restriction and Induction}
Given a Group $G$ with a representation $V$ and a subgroup $H\subset G$ it is straightforward to define a representation of $H$ on $V$ by \emph{restriction}. for this representation we write $V\!\downarrow^G_H$. A little less obvious is the construction of a representation $W$ of $G$ from a representation $V$ of $H$. To define this recipe called induction we need the definition of a
\begin{defn}[Transversal]
	Let $G$ be a group and $H\subset G$ a subgroup. A subset $T\subset G$ is called \emph{(left) transversal for $H$}, if
	\begin{enumerate}
		\item[(i)]$TH=G$
		\item[(ii)] $xH\cap yH=\emptyset\,\, \forall x,y\in T, x\not=y$.
	\end{enumerate}
\end{defn}
The above definition is equivalent to saying that a transversal for $H$ in $G$ contains exactly one element from each (left) coset. Let us now define the
\begin{defn}[Induced Representation]
	Let $G$ be a group, $H\subset G$ a subgroup and $(\phi,V)$ a representation of $H$. Furthermore set $\phi(g)=0$ for $g\in G\setminus H$ and fix and order a transversal $T=(t_1,...,t_k)$. Then we define the \emph{induced representation} $\phi\!\uparrow_H^G$ on $W=V^{\oplus k}$ by
	\begin{equation}
		\phi\!\uparrow_H^G(g)=\left(\begin{array}{ccc}
						\phi(t_1gt_1^{-1})&\dots&\phi(t_1gt_k^{-1})\\
						\vdots&\ddots&\vdots\\
						\phi(t_k g t_1^{-1})&\dots&\phi(t_k g t_k^{-1})
		                            \end{array}\right).
	\end{equation}
\end{defn}
It is straightforward to verify that the induced representation is a representation and that induced representations corresponding to different transversals of the same subgroup are isomorphic. While being easily explained in simple terms, the above construction is somewhat dissatisfactory because it first uses a transversal and it has to be proven afterwards that the construction doesn't depend on it. This can be circumvented by giving the definition in terms of a generalized notion of tensor products.
\begin{defn}[Tensor Product]
 Let $R$ be a ring, $M$ a right $R$-module and $N$ a left $R$-module. Let $F=\Z(M\times N)$ be the free Abelian group over the Cartesian product of $M$ and $N$ and define the subgroup $I$ generated by the set $S=S_1\cup S_2$,
 \begin{eqnarray}
  S_1&=&\left\{(x+y)\times z-x\times z-y\times z\Big| x,y\in M, z\in N\right\}\\
  S_2&=&\left\{(x \alpha)\times y-x\times(\alpha y)\Big| x\in M, y\in N, \alpha\in R\right\}.
 \end{eqnarray}
Then
\begin{equation}
 M\otimes_RN=F/I
\end{equation}
Is the $R$-tensor-product of $M$ and $N$. Whenever $M$ is also a left $R'$-module for another ring $R'$, $M\otimes_RN$ is a left $R'$ module as well, and when $N$ is also a right $R''$-module for yet another ring $R''$, $M\otimes_RN$ is a right $R''$-module as well.
\end{defn}
Note that this definition specializes to the usual definition of a tensor product between vector spaces if $R=\mathbb{F}$ is a field.

We are now in the position to give a transversal independent definition of the induced representation:

\begin{defn}[Induced Representation, 2nd Definition]
	Let $\mathbb{F}$ be a field, $G$ be a group, $H\subset G$ a subgroup and $(\phi,V)$ a $\mathbb{F}$-representation of $H$. This is equivalent to stating that $V$ is a $\mathbb{F}H$-left-module. Then then induced representation is the $\mathbb{F}G$-left-module
	\begin{equation}
	  V\uparrow_H^G=(\mathbb{F}G)\otimes_{{}_{\mathbb FH}}V.
	\end{equation}
\end{defn}
Note that $\mathbb FG$ is a $\mathbb F H$-bimodule for any subgroup $H\subset G$. To recover the transversal dependent construction, we choose a left-transversal $T=\left\{t_1,...,t_{|G|/|H|}\right\}$ and observe that the set $\left\{t\otimes_{{}_{\mathbb FH}}v\Big| t\in T,v\in B\right\} $ generates  $V\uparrow_H^G$ for any basis $B$ of $V$.


\subsection{Character Theory}
Character theory is a powerful means of analyzing group representations. Given a representation $(\phi_\alpha,V_\alpha)$ of a finite group $G$, we define its \emph{character} as the map (group algebra element)
\begin{equation}
	\chi_\alpha: G\to \C,\ \ g\mapsto \tr \phi_\alpha(g).
\end{equation}
Note that, for the purpose of a clear definition of the character, we have temporarily reintroduced the distinction between the representation (-map) $\phi$ and the representation space $V$. The characters are in the center of $\mathcal{A}(G)$ denoted by $\mathcal{Z}(G)$, which follows from the fact that they are constant on conjugacy classes: 
\begin{equation}
	\chi(h^{-1}gh)=\tr\phi(h^{-1}gh)=\tr(\phi(h)^{-1}\phi(g)\phi(h)=\tr\phi(g)=\chi(g)
\end{equation}
The characters of equivalent representations are identical, as the trace is basis independent and the transition to an equivalent representation can be viewed as a basis change. It follows directly from the Schur orthogonality relations, Theorem \ref{schurortho1} that the characters are orthonormal in $\mathcal{A}(G)$, i.e.
\begin{equation}\label{chiorth}
	(\chi_\alpha,\chi_\beta)=\delta_{\alpha\beta}.
\end{equation}
This provides us with a way finding the multiplicity of an irreducible representation in a given representation far simpler that Corollary \ref{dimhom}. If an arbitrary representation $W$ has a decomposition into irreducible representations
\begin{equation}
	W=\bigoplus_\alpha V_\alpha^{\oplus m_\alpha},
\end{equation}
then its character is easily determined to be
\begin{equation}
	\chi_W=\sum_\alpha m_\alpha\chi_\alpha,
\end{equation}
and hence, using \eqref{chiorth},
\begin{equation}
	(\chi_\alpha,\chi_W)=m_\alpha.
\end{equation}
All the above can be summarized by the statement that an equivalence class of representations is uniquely determined by its character and that the irreducible characters are orthonormal.


\subsection{The Regular Representation}\label{regrep}
Let $G$ be a group. Consider the action
\begin{equation}\label{regac}
	R: G\times \C G\to \C G,\ \ (h,\sum \alpha_g g)\mapsto \sum \alpha_{h^{-1}g}g
\end{equation}
on the group algebra as a vector space. Let $G$ for now be finite. Using the theory of characters introduced in the last subsection we can analyze the regular representation. It's character is
\begin{equation}
	\chi_R(g)=\begin{cases}
	          	|G|& g=e\\
	          	0&\text{else}
	          \end{cases},
\end{equation}
as $gh\not=h$ if $g\not=e$. Explicitly calculating the inner product of $\chi_R$ with the irreducible representations yields
\begin{equation}
	(\chi_R, \chi_V)=\dim V
\end{equation}
which implies that the decomposition of the Regular representation into a sum of irreducible representations is
\begin{eqnarray}\label{regdec}
	\C G&=&\bigoplus_{V\text{ irrep of }G} V^{\otimes \dim V}\nonumber\\ &\cong&\bigoplus_{V\text{ irrep of }G} V\otimes V
\end{eqnarray}
The last expression reflects the fact that there is, in addition to the left action \eqref{regac}, a right action
\begin{equation}\label{regright}
	R': G\times \C G\to \C G,\ \ (h,\sum \alpha_g g)\mapsto \sum \alpha_{gh}g
\end{equation}
which commutes with the former.
From the decomposition \eqref{regdec} we also get an explicit formula for the cardinality of the group in terms of the dimensions of its irreducible representations,
\begin{equation}\label{dimsquarestoG}
	|G|=\sum_{V\text{ irrep of }G} (\dim V)^2.
\end{equation}
We are now ready to prove part two of Theorem \ref{schurortho1}.
\begin{thm}[Schur orthogonality relations, Part II]\label{schurortho2}
	Let $G$ be a finite Group. Let $\alpha$ label the equivalence classes of irreducible representations of $G$ and pick a unitary representative $U_\alpha: G\to U(V_\alpha)$. Then $\{U_{\alpha ij}\}$ is a basis of $\mathcal{A}(G)$, and the $\alpha$ components don't mix in the sense that
	\begin{equation}\label{algmult}
		U_{\alpha ij}U_{\beta kl}=\delta_{\alpha\beta}\delta_{jk}\frac{|G|}{d_\alpha}U_{\alpha il}
	\end{equation}
\end{thm}
\begin{proof}
	In Theorem \ref{schurortho1} we already saw that $\{U_{\alpha ij}\}$ is an orthogonal set and in particular linearly independent. But Equation \ref{dimsquarestoG} directly implies $\left|\{U_{\alpha ij}\}\right|=\dim\mathcal{A}(G)$, hence $\{U_{\alpha ij}\}$ is indeed a basis.
	For the last part of the theorem, let us calculate
	\begin{eqnarray}
		(U_{\alpha ij}U_{\beta kl}, U_{\gamma mn})&=&\frac{1}{|G|}\sum_{x,y\in G}U^*_{\alpha ij}(xy^{-1})U^*_{\beta kl}(y) U_{\gamma mn}(x)\nonumber\\
		&=&\frac{1}{|G|}\sum_{x,y\in G}\sum_{s=1}^{d_\alpha}U^*_{\alpha is}(x)U_{\alpha js}(y)U^*_{\beta kl}(y) U_{\gamma mn}(x)\nonumber\\
		&=&|G|\sum_{s}\left(U_{\alpha is},U_{\gamma mn}\right)\left(U_{\beta kl},U_{\alpha js}\right)\nonumber\\
		&=&\frac{|G|}{d_\alpha^2}\delta_{\alpha\beta}\delta_{\alpha\gamma}\delta_{im}\delta_{ln}\delta_{jk},
	\end{eqnarray}
	where for the second equality we used the properties of a unitary representation and for the third one we used Theorem \ref{schurortho1}.
	Using the orthonormal basis property proven above this implies the multiplication law \eqref{algmult}.
\end{proof}
The last result implies that the group algebra of a finite group $G$ is isomorphic to the direct sum of the matrix algebras over the irreducible representation spaces,
\begin{equation}
	\mathcal{A}(G)\cong\bigoplus_{\alpha}\End(\C^{d_\alpha}),
\end{equation}
for example via the isomorphism
\begin{eqnarray}\label{groupalgdirsum}
	\mathcal{A}(G)&\stackrel{\sim}{\longrightarrow}&\bigoplus_{\alpha}\End(\C^{d_\alpha})\nonumber\\
	\sum_{\alpha}\sum_{i,j=1}^{d_\alpha}a_{\alpha ij}U_{\alpha ij}&\mapsto& \bigoplus_{\alpha}A_\alpha \text{ with }(A_{\alpha})_{ij}=a_{\alpha ij},
\end{eqnarray}
where in the second line we fixed a set of unitary irreducible representations, or, equivalently, a basis for each $\C^{d_\alpha}$ to choose a definite isomorphism. 
Can we explicitly find minimal projections of $\mathcal{A}(G)$ as well as its center? According to Theorem \ref{schurortho2} the diagonal elements of any irreducible unitary representation are proportional to projections, and in view of \eqref{groupalgdirsum} they are also minimal. In view of the decomposition \eqref{regdec} they project onto a single copy of the corresponding irreducible representation with respect to the right action \eqref{regright}. The isomorphism \eqref{groupalgdirsum} also implies, together with \eqref{chiorth} that the set of irreducible characters forms an orthonormal basis of $\mathcal{Z}(G)$. The multiplication rule \eqref{algmult} from Theorem \ref{schurortho2} implies furthermore that the irreducible characters must square to multiples of themselves, in fact, explicitly exploiting \eqref{algmult},
\begin{equation}
	\chi_\alpha\chi_\alpha=\frac{|G|}{d_\alpha}\chi_\alpha,
\end{equation}
and therefore
\begin{equation}
	\pi_\alpha=\frac{d_\alpha}{|G|}\chi_\alpha
\end{equation}
are the minimal central projections. Now consider an arbitrary representation $W$ of $G$ with decomposition into irreducible representations
\begin{equation}
	W=\bigoplus_{\alpha}V_{\alpha}^{\oplus m_\alpha}.
\end{equation}
As in the group algebra $\chi_\alpha$ projects onto the $V_\alpha$ irreducible component, $\chi_\alpha$ acts on $W$ by projecting onto $V_\alpha^{\oplus m_\alpha}$.


\subsection{Irreducible representations of $S_n$}\label{symirreps}
We want to identify the irreducible representations of the symmetric group $S_n$, i.e.\! the permutation group of $n$ elements. To this end, we make use of the regular representation as it contains all irreducible representations of a finite group. Let us first introduce the important tool called young diagrams.

Given a partition $\lambda=\left(\lambda_1,\lambda_2,...,\lambda_d\right)$ of $n\in\N$ into a sum of non increasing numbers $\lambda_i\in\N$ we can define the corresponding
\begin{defn}[Young Diagram]
	A \emph{Young diagram} is a subset $T\subset \N_+\times\N_+$ for which the following holds:
	\begin{equation}
		(n,m)\in T\implies (k,m)\in T\, \forall k<n \text{ and } (n,l)\in T\, \forall l<m
	\end{equation}
	Elements of a Young diagram are called \emph{boxes}, subsets with constant first component are called \emph{columns}, such with constant second component \emph{rows}. For a Young diagram $\lambda$ of $n$ boxes we write $\lambda\vdash n$, for a Young diagram of $n$ Boxes and at most $d$ rows $\lambda\vdash (n,d)$. 
\end{defn}
The picture one should have in mind reading this definition is the one obtained by taking an empty box for each element of $T$ and arranging them in a diagram such that the ``origin'' of $\N_+\times\N_+$ is in the upper left corner:
\begin{equation*}
	\lambda=\Yvcentermath1\yng(5,3,2,2,1)
\end{equation*}
This example corresponds to a partition of $n=13$, namely $13=5+3+2+2+1$. We also write $\lambda=(532^21)$. A Young diagram filled with numbers is called a \textit{Young tableau}:
\begin{equation*}
	\young({124},3)\ ,\ \young({112},2)
\end{equation*}
The first kind is called standard, the second semistandard:
\begin{defn}
	A Young tableau $T$ of shape $\lambda$ is a Young diagram $\lambda$ with a number in each box. We also write $\sh(T)=\lambda$. A standard Young tableau is a Young diagram of $n$ boxes that is filled with the numbers from 1 to $n$ such that numbers increase from left to right along each row and down each column. A semistandard Young tableau is a Young diagram filled with numbers which are nondecreasing along each row and increasing down each column. We write $T(\lambda)$, $S(\lambda)$ and $s(\lambda,d)$ for the sets of Young tableaux filled with the numbers $1$ to $n$, the standard Young tableaux and the semistandard Young Tableaux with entries smaller or equal to $d$, respectively.
\end{defn}
A standard Young tableau $T$ of shape $\lambda\vdash n$ defines two subgroups of $S_n$, one that permutes the entries of the rows, $R_T$, and one that permutes the entries of the columns, $C_T$. Define the group algebra elements
\begin{eqnarray}
	r_T&=&\sum_{r\in R_T}r\nonumber\\
	c_T&=&\sum_{c\in C_T}\sgn(c)c\nonumber\\
	e_T&=&r_Tc_T\label{youngsym}
\end{eqnarray}
The last one, $e_T$, is proportional to a minimal projection
\begin{equation}\label{youngproj}
	e_T^2=\frac{k!}{\dim V_\lambda}e_T
\end{equation}
and is called the \emph{young symmetrizer}. As a minimal projection, according to the discussion in section \ref{regrep}, it projects onto a single copy of an irreducible representation in the decomposition of the group algebra $\mathcal{A}(S_n)$ with respect to the right action of $S_n$.

Another way of constructing the representations of $S_n$ is to stop after symmetrizing the rows of a Young tableau and looking at the corresponding representation, that is
\begin{equation}
	S_n\looparrowright M^\lambda=\C\left\{T\in T(\lambda)\big|T(i,j)\le T(i,j+1)\right\}
\end{equation}
where the action is defined by permuting the entries and then resorting the rows. Tableaux whose rows are ordered but their columns are not are called \emph{row-standard}. This representation is called \emph{permutation module}. It can also be constructed in a different way. Each Young diagram $\lambda\vdash n$ defines a subgroup of $S_n$, the so called \emph{Young subgroup} $S_\lambda:\cong S_{\lambda_1}\times...\times S_{\lambda_n}$ with $S_{\lambda_1}$ permutes the first $\lambda_1$ elements, $S_{\lambda_2}$ permutes $\lambda_1+1, ...,\lambda_1+\lambda_2$ etc. Then it is easy to verify that the permutation module $M^\lambda$ is the representation of $S_n$ induced by the trivial representation of the corresponding Young subgroup $S_\lambda$, as a formula $M^\lambda=1\!\uparrow_{S_\lambda}^{S_n}$.

\paragraph{}What is the relation between the permutation modules and the irreducible representations of $S_n$, which are also called Specht modules? The permutation module $M_\lambda$ contains the irreducible representation $[\lambda]$ exactly once, and otherwise only contains irreducible representations $[\mu]$ with $\mu\vdash|\lambda|$ and $\mu\succ \lambda$, i.e.
\begin{equation}\label{kostka}
	M^\lambda\cong\bigoplus_{\mu\succ\lambda}K_{\mu\lambda}[\mu],
\end{equation}
the multiplicities $K_{\mu\lambda}$ are called \emph{Kostka numbers}. These numbers have a simple combinatorial description: Define the set $s(\lambda,\mu)$ of semistandard Young tableau with shape $\lambda$ and content $\mu$, i.e.\ the numbers in the tableau have frequency $\mu$ as a string. Then the Kostka number $K_{\mu\lambda}$ is the number of such tableaux,
\begin{equation}
	K_{\mu\lambda}=|s(\lambda,\mu)|
\end{equation}


\subsection{Irreducible Representations of the Unitary Group}
The unitary group $\mathrm{U}(n)$ is defined as the group of endomorphisms of $\C^n$ that leaves the standard inner product invariant. It is a Lie group and we can easily find its Lie algebra. For $X\in \C^{n\times n}$ with  $X^\dagger =-X$, $e^X\in \mathrm{U}(n)$. Also $e^X\not\in \mathrm{U}(n)$ if $X$ is not antihermitian, so the Lie algebra $\mathfrak{u}(n)$ is the set of antihermitian $n\times n$-matrices. If we look at an arbitrary (finite dimensional, unitary) irreducible representation $U$ of $\mathrm{U}(n)$, we know that for the restriction to the Abelian subgroup 
\begin{equation}
	\mathrm{H}(n)=\left\{U\in \mathrm{U}(n)|U \text{ diagonal}\right\}\cong U(1)^{\times n}
\end{equation}
 of diagonal unitaries (in some fixed basis) -- a \emph{Cartan} subgroup --  there is a basis $\left\{\left|v_i\right\rangle\right \}$ of $U$ where it also acts diagonal. The holomorphic irreducible representations of $\mathrm{H}(n)$ are just 
 \begin{equation}
  \diag(u_1,...,u_n)\mapsto \prod_{i=1}^nu_i^{\lambda_i}
 \end{equation}
 for some $\lambda\in\Z^n$, so to each basis vector $\left|v_i\right\rangle$ there is a $\lambda \in\Z^n$ such that for all $u=\diag(u_1,u_2,...,u_n)\in\mathrm{H}(n)$ we have 
 \begin{equation}
 u\left|v_i\right\rangle=\prod_{j=1}^nu_j^{\lambda_j}\left|v_i\right\rangle.
 \end{equation}
Such a vector is called a \textit{weight vector}, and $\lambda$ is called its \textit{weight}. This translates to a similar property in the Lie algebra. The restriction to Lie subalgebra $\mathfrak{h}(n)$ corresponding to $\mathrm{H}(n)$ of the Lie algebra representation defined on $U$ acts diagonally in the same basis, for $h\in \mathfrak{h}(n)$ we get
\begin{equation}
	h\left|v_i\right\rangle=\sum_{j=1}^nh_j\lambda_j\ket{v_i}
\end{equation}
This procedure of diagonalizing the action of a Cartan subgroup and the corresponding Cartan subalgebra can also be done for the adjoint representation (see Section \ref{lie}). The weights that are encountered there are called \emph{roots}, the vector space they live in is called \emph{root space}. The structure of the \emph{root lattice} generated by the roots captures the properties of the underlying Lie group.

Let us now look at the complexified Lie algebra $\mathfrak{u}(d)_\C=\mathfrak{gl}(n)=\C^{n\times n}$. The representations of a real Lie algebra and its complexification are in a one to one correspondence, see e.g.\ \cite{carter32lectures}. Define the standard basis of $\C^{n\times n}$ to be the set of matrices $E_{ij}$ which have a one at position $(i,j)$ and are zero elsewhere. This basis is an eigenbasis of the adjoint action defined in \eqref{ad}, because
 \begin{equation}
 	[\diag(h_1,...,h_n), E_{ij}]=(h_i-h_j)E_{ij}.
 \end{equation}
 Using the matrices $E_{ij}$, which are the multidimensional analogues of the well known ladder operators of $\mathrm{SU}(2)$ we can reconstruct the whole irreducible representation. Let us consider an ordering on the set of weights, for example the lexicographical order, that is
 \begin{equation}
 	\lambda <\mu:\Leftrightarrow \lambda\not=\mu\text{ and }\lambda_s<\mu_s \text{ for } s=\min\left\{s|\lambda_s\not=\mu_s\right\}.
 \end{equation}
 Note that this ordering is arbitrarily chosen. This is equivalent to choosing an irrational functional in the dual of the root space and thus totally ordering the roots.
 Looking at an arbitrary representation $U$ again, because $\dim U<\infty$ we can find a highest weight $\lambda$ and at least one corresponding weight vector $\left|v_\lambda \right\rangle$. It turns out that this highest weight vector is unique. But first take $h\in \mathfrak{h}(n)$ and calculate
 \begin{equation}
 h E_{ij}\left|v_\lambda\right\rangle=[h,E_{ij}]\ket{v_\lambda}+E_{ij}\sum_{j=1}^nh_j\lambda_j\ket{v_\lambda}=\left(h_i-h_j+\sum_{j=1}^nh_j\lambda_j\right)E_{ij}\ket{v_\lambda}
 \end{equation}
so either $E_{ij}\ket{v_\lambda}$ is zero, or it is a weight vector for the weight $\lambda +\epsilon_{ij}$ where $\epsilon^{(ij)}_k=\delta_{ik}-\delta_{jk}$. Because $\lambda$ is the highest weight, $E_{ij}\ket{v_\lambda}=0$ for $i>j$. For fixed $i<j$ look at the three Lie algebra Elements
\begin{eqnarray}
H&=&E_{ii}-E_{jj},\nonumber\\
X&=&E_{ij}\text{ and}\nonumber\\
Y&=&E_{ji}.
\end{eqnarray}
Then $\left\{H,X,Y\right\}$ is a Lie subalgebra isomorphic to $\mathfrak{su}(2)_\C=\mathfrak{sl}(2)$, as $[H,X]=2X, \ [H,Y]=-2Y$ and $[X,Y]=H$, so $\ket{v_\lambda}$ generates a $2\left(\lambda_i-\lambda_j\right)$+1-dimensional representation of $\mathfrak{sl}(2)$. In this fashion repeated application of the $E_{ij}$, $i<j$, yields a basis for the irreducible representation $U$.


\subsection{Schur-Weyl Duality}

In this section I shortly explain the Schur-Weyl duality theorem. A good introduction to this topic can, for example, be found in \cite{christandl2006structure}. This will be important when I consider similar constructions in Chapter \ref{CYR}.

Consider the tensor product space $\left(\C^d\right)^{\otimes n}$. The symmetric group $S_n$ has a natural unitary action on that space by permuting the tensor factors, i.e.
\begin{equation}\label{Sntens}
	\pi \ket{v_1}\otimes\ket{v_2}\otimes...\otimes\ket{v_n}=\ket{v_{\pi^{-1}(1)}}\otimes\ket{v_{\pi^{-1}(2)}}\otimes...\otimes\ket{v_{\pi^{-1}(n)}},
\end{equation}
And $\mathrm{U}(d)$ acts via its tensor representation, i.e.\ for $U\in\mathrm{U}(d)$,
\begin{equation}
	U\cdot\ket{e_{i_1}}\otimes\ket{e_{i_2}}\otimes...\otimes\ket{e_{i_n}}=\left(U\ket{e_{i_1}}\right)\otimes\left(U\ket{e_{i_2}}\right)\otimes...\otimes\left(U\ket{e_{i_n}}\right)
\end{equation}
Obviously the two actions commute. But even more is true, that is, the subalgebras of End$\left(\left(\C^d\right)^{\otimes n}\right)$ generated by the two representations are each others commutants. The Schur-Weyl duality theorem then states that
\begin{equation}\label{schurweyl}
	\left(\C^d\right)^{\otimes n}\cong\bigoplus_{\lambda\vdash(n,d)}[\lambda] \otimes V_\lambda
\end{equation}
where $[\lambda]$ is the the representation projected out by the central projection corresponding to the frame $\lambda$ and $V_\lambda$ is the representation of U$(d)$ with highest weight $\lambda$.

Of course the representation of the group extends by linearity to a representation of the group algebra. Recall the definition of the Young symmetrizer $e_T$ corresponding to a standard young tableau $T$ of shape $\lambda\vdash(n,d)$. It projects onto a single vector in the representation $[\lambda]$ of $S_n$ so it projects onto a space of dimension equal to the multiplicity of $[\lambda]$ in the representation \eqref{Sntens}. More precisely, because $[\lambda]$ is paired with $U_\lambda$ ind \eqref{schurweyl}, $e_T$ actually projects onto a copy of $U_\lambda$.


\chapter{Physical and Information Theoretical Background}\label{info}
\section{Classical Information Theory}\label{classinfo}
In the following chapter I first give a short introduction into the mathematical formalism of classical information theory. In the subsequent sections I introduce the Shannon entropy, investigate its basic properties, and describe the convex geometry framework used to describe joint and marginal entropies of a multipartite random variable. Finally I give a short example how characterization results for the entropy cone are useful in applications by elaborating the connection to network coding.

In classical information theory states are modeled as measurable functions $X: \Om\to\mathcal{X}$ called \textit{random variables}, where $(\Om,\mathsf{X},\mathbf{P})$ is a probability space and $\mathcal{X}$ is a discrete set called \textit{alphabet}. Explaining the concept of a probability space at length lies beyond the scope of this thesis, an introduction can be found in \cite{kallenberg2002foundations}. In simple words, $\Omega$ is just a set, $\mathsf{X}\subset 2^\Om$ is a sigma algebra of \emph{measurable sets}, and $\mathbf{P}:\mathsf{X}\to [0,1]$ is a measure with the additional requirement that $\mathbf{P}(\Om)=1$ called \emph{probability measure}. As most of the following is concerned with random variables on finite alphabets, it is enough to know that $X$ is an object that outputs elements from this alphabet with fixed probabilities given by the corresponding \emph{probability distribution}
\begin{equation}
	p_X: \mathcal{X}\to [0,1],\ x\mapsto \mathbf{P}(X^{-1}(\{x\}))
\end{equation}
A realization of a random variable is called a \emph{variate}. The set of all probability distributions for $n$ outcomes is an $n-1$-simplex
\begin{equation}
 \mathcal{P}^{n}=\left\{p\in\R_{\ge 0}^n\Big|\sum_{i=1}^np_i=1\right\}
\end{equation}
called \emph{probability simplex}. We also call the elements of the probability simplex \emph{probability vectors}, especially if no corresponding alphabet is specified.

To quantify the information content of a random variable, or, in other words, the information that is gained by learning its outcome, information theory uses certain functions called \textit{entropies}.


\subsection{The Shannon Information Measures}\label{shannent}
Entropies play a key role in information theory, classical and quantum. In fact, they did so right from the beginning, Shannon introduced the entropy that was later named after him in the very same paper that is said to constitute the birth of modern information theory. Entropies are functionals on the state space of the respective information theory that quantify the average information content of a state. In classical information theory, the Shannon entropy of a random variable $X$ on some alphabet $\mathcal{X}$ is defined as
\begin{equation}
 H(X)=-\sum_{x\in\mathcal{X}}p_X(x)\log p_X(x).
\end{equation}
Note that we omit the subscript $X$ if it is clear from context to which random variable the probability distribution belongs. Observe moreover that the Shannon entropy only depends on the probability distribution $p$ of $X$, therefore we also write $H(p)$ sometimes.
Originally, Shannon derived this entropy, up to a constant factor, from the following simple axioms \cite{shannon1948mathematical}:
\begin{enumerate}
 \item $H$ should be continuous in the $p_i$
 \item for a uniform distribution $H$ should monotonically increase with the number of possible outcomes
 \item If $p_i=q_{j,i}\bar p_j$ describing a two step random process, then $H$ should be the weighted sum of the entropies of the steps:
 \begin{equation}
  H(p)=H(\bar p)+\sum_j p_j H(q_j)
 \end{equation}
\end{enumerate}
However, the clearest justification for the claim that this of all functionals quantifies the information content of $X$ is due to the famous
\begin{thm}[Noiseless Coding Theorem (\cite{shannon1948mathematical}, Theorem 9)]
 Let a source produce independent copies of a random variable $X$ and let $\mathsf{C}$ be a channel with capacity $C$. Then each rate $R<\frac{C}{H(X)}$ of perfect transmission of letters can be achieved, and each rate $R>\frac{C}{H(X)}$ is impossible to achieve.
\end{thm}
This provides us with a connection to our intuitive understanding of information: Suppose the channel is just a device that perfectly transmits bits. Then the capacity is 1 and the coding theorem implies that we need on average at least $H(X)$ bits to encode $X$.

\paragraph{}Let us review some basic properties of the Shannon entropy.
The Shannon entropy is nonnegative,
\begin{equation}
	H(X)\ge 0,
\end{equation}
as $0\le p(x) \le 1\, \forall x\in \mathcal{X}$. Looking at more than one random variable, there are several quantities commonly used in information theory that are defined in terms of the entropies of their joint and marginal distributions. Considering two random variables $X$ and $Y$ we define the \textit{conditional entropy}
\begin{eqnarray}
	H(X|Y)&=&\sum_{x,y\in\mathcal{X}}p(x,y)\log(p(x|y))\nonumber\\
	&=&\sum_{x,y\in\mathcal{X}}p(x,y)\log(\frac{p(x,y)}{p(y)})\nonumber\\
	&=&\sum_{x,y\in\mathcal{X}}p(x,y)\left(\log(p(x,y))-\log(p(y))\right)\nonumber\\
	&=&H(XY)-H(Y),
\end{eqnarray}
 where $H(XY)$ is the natural extension of the Shannon entropy to many random variables, that is the Shannon entropy of the random variable $Z: \Om\to \mathcal{X}\times\mathcal{X},\ \om \mapsto (X(\om),Y(\om))$. Note that we assumed for notational convenience that $X$ and $Y$ are defined on the same alphabet. The conditional entropy has a very intuitive operational interpretation as well. Suppose we want to encode a string of $X$-variates but we know that the receiver has The corresponding string of $Y$-variates as side information to help him decode the message. Then we can achieve a rate no better than $H(X|Y)$.
 
Another information measure also defined by Shannon is the \textit{mutual information} of two random variables $X$ and $Y$,
\begin{eqnarray}\label{mutual}
	I(X:Y)&=&H(XY)-H(X|Y)-H(Y|X)\nonumber\\
	&=&H(X)+H(Y)-H(XY)=H(X)-H(X|Y)=H(Y)-H(Y|X).
\end{eqnarray}
Like the other quantities the mutual information has a precise operational meaning, that is, it is the average information about $X$ that is gained by looking at $Y$, or vice versa. As the Shannon entropy itself, also for the mutual information one can define a conditional version, the \textit{conditional mutual information} that is obtained by taking one of the expressions \eqref{mutual} for the mutual information and conditioning on another random variable $Z$,
\begin{eqnarray}
	I(X:Y|Z)&=&H(XY|Z)-H(X|YZ)-H(Y|XZ)\nonumber\\
	&=&H(XYZ)-H(Z)-H(XYZ)+H(YZ)-H(XYZ)+H(XZ)\\
	&=&H(XZ)+H(YZ)-H(XYZ)-H(Z)
\end{eqnarray}
The conditional mutual information is nonnegative, i.e.
\begin{equation}\label{condmutpos}
	I(X:Y|Z)\ge 0.
\end{equation}
A proof for that fact can be found in \cite{yeung2008information}, an alternative proof using type classes is given in Section \ref{strings}. It follows that all Shannon information measures introduced above are nonnegative, as, with a trivial random variable $\delta$ that has probability 1 for a certain outcome and zero elsewhere, $H(X)=I(X:X|\delta)$, $H(X|Y)=I(X:X|Y)$ and $I(X:Y)=I(X:Y|\delta)$. The inequalities ensuring positivities of the Shannon information measures and linear combinations of these are said to be of \textit{Shannon type}.

Another information theoretic quantity that is related to the Shannon entropy is the \emph{relative entropy}. Given two random variables $X$ and $Y$ on the same alphabet $\mathcal{X}$, it is defined as
\begin{equation}\label{relent}
 H(X\| Y)=\begin{cases}
           \sum_{x\in \mathcal{X}}p(x)\log\left(\frac{p(x)}{q(x)}\right)& p(x)=0\,\forall x\in \mathcal{X} \text{ with } q(x)=0\\
           \infty&\text{else}
          \end{cases}.
\end{equation}
It is used in the literature under a variety of other names such as information distance, information divergence or Kullback-Leibler-distance. In the next section it will appear in the context of quantum information theory as well, where it serves as a distance measure between spectra. Note that $H(X\| Y)$ is not a metric, although it is called a ``distance''.

An important result from classical information theory is the asymptotic equipartition property. It can be understood as a strengthening of the law of large numbers for random variables on finite alphabets, as it implies, among other statements, convergence of the \emph{empirical distribution} of a sample.
\begin{thm}[Asymptotic Equipartition Property]\label{AEP}
 Let $X_i$, $i\in \N$ be independent and identically distributed on a finite alphabet $\mathcal{X}$ according to a probability distribution $p$. Then
 \begin{equation}
  -\frac{1}{n}\log p(X_1,...,X_n)\to H(p)
 \end{equation}
in probability\footnote{A sequence of random variables $(Z_i)_{i\in\N}$ with range in a metric space converges to a random variable $Z_\infty$ \emph{in probability}, if $\mathbf{P}\left(\left|Z_i-Z_{\infty}\right|>\epsilon\right)\to0\,\forall \epsilon>0$.} for $n\to\infty$.
\end{thm}
The proof using the law of large numbers (which is probably more widely known among physicist) is very simple and therefore I include it here for completeness.
\begin{proof}\cite{cover2012elements}
 As the random variables $X_i$ are independent, so are $Y_i=-\log(p(x_i))$. Hence by the law of large numbers the mean value $\overline Y_i:=\frac{1}{n}\sum_{j=1}^i Y_j$ converges to the expectation value in probability, which is equal to the Shannon entropy $H(p)$.
\end{proof}

This theorem implies, that there is a subset $\mathcal{A}\subset\mathcal{X}^n$ of size approximately $2^{n H(p)}$ such that $\mathbf{P}( \mathcal{A})\stackrel{n\to\infty}{\longrightarrow}$1.


\subsection{The Classical Entropy Cone}\label{classical}

As seen in the last section, interesting non-trivial constraints govern the Shannon entropies of a number of random variables. In this section I introduce the formalism to treat the characterization of the entropies of a collection of random variables.

Consider a collection of random variables $X_1,X_2,...,X_n$. Then for each $I\subset[n]=\{1,2,...,n\}$ there is an entropy $H\left(X_I\right)=H\left((X_i)_{i\in I}\right)$, $2^n$ entropies in total, where we adopt the convention $H(X_\emptyset)=0$, i.e.\ the Shannon entropy of zero random variables is zero. Each such collection defines an \textit{entropy vector} $(H(X_I))_{I\subset [n]}\in \R^{2^n}:= V_n$. We call $V_n$ the \emph{entropy space} and we denote its standard basis vectors by $\{e^{(I)}|I\subset[n]\}$. For fixed $n$, define the set of all such vectors
\begin{equation}
	\Sigma_n=\left\{h(X)\in  V_n|X=(X_i)_{i\in[n]},\ X_i\text{ random variables}\right\},
\end{equation}
where $h(X)=(H(X_I))_{I\subset [n]}$ is the \emph{entropy function}. It turns out that its topological closure $\overline{\Sigma}_n$ is a convex cone \cite{zhang1997non}. One could think that $\Sigma_n$ were closed itself because of the continuity of the Shannon entropy and thus a cone. This is not the case for the following reason: The Shannon entropy is continuous only for a fixed \emph{finite} alphabet. It may happen, however, that some points in the closure of $\Sigma_n$ can only be approached by sequences where the alphabet sizes of the corresponding random variables are not bounded. Matu\v s later proved that $\text{ri}(\overline{\Sigma_n})\subset\Sigma_n$ \cite[Theorem 1]{matus2007two}, where ri$(A)$ denotes the relative interior of $A$, i.e.\ the interior of $A$ in the topology of $\spa A$. This means, that the only points in $\overline{\Sigma}_n$ that cannot be realized by random variables on a finite alphabet are located on the boundary of $A$.
The proof of that fact also produces the dimension of $\Sigma_n$ as a byproduct:
\begin{prop}[\cite{matus2007two}]\label{dimSig}
	$\dim \Sigma_n=2^n-1$
\end{prop}
\begin{proof}
	Let $Y$ be a fair coin and $T$ a trivial random variable. For each $I\subset [n]$ define a collection of random variables $\left(X^{(I)}_j\right)_{j\in[n]}$ such that
	\begin{equation}
		X^{(I)}_j=\begin{cases}
		    	Y&j\in I\\
		    	T& \mathrm{else}
		    \end{cases}.
	\end{equation}
	Then we have 
	\begin{equation}
	v^{(I)}_J:=h\left(X^{(I)}\right)_J=\begin{cases}
	                        	0&I\cap J=\emptyset\\
	                        	1&\mathrm{else}
	                        \end{cases}.
	\end{equation}
	Consider the linear map
	\begin{eqnarray}
		A: V_n&\to&V_n\nonumber\\
		v&\mapsto& w \text{ with }w_I=\sum_{K\subset I}(-1)^{|I\setminus K|}\left(v_{[n]}-v_{[n]\setminus K}\right).\label{matustrafo}
	\end{eqnarray}
	It follows from an elementary calculation that
	\begin{equation}
		Av^{(I)}=e^{(I)}, \emptyset\not=I\subset [n]
	\end{equation}
	Therefore, as the family $\left(e^{(I)}\right)_{\emptyset\not= I\subset[n ]}$ is linearly independent, so is the family $\left(v^{(I)}\right)_{\emptyset\not= I\subset[n ]}$, and we have
	\begin{eqnarray}
	2^{n}-1&=&\dim\spa\{e^{(I)}|\emptyset\not= I\subset [n]\}\nonumber\\&=& \dim\spa\{v^{(I)}|I\subset [n]\}\nonumber\\ &\le&\dim \spa \Sigma_n\nonumber\\&=&\dim\Sigma_n.
	\end{eqnarray}
	But $v_{\emptyset}=0$ for all $v\in\Sigma_n$ and dim$V_n=2^n$, so $\dim \Sigma_n\le2^n-1$, ergo $\dim \Sigma_n=2^n-1$.
\end{proof}

It is an important problem to characterize this cone. One way to do this is to find linear inequalities for $\Sigma_n$, e.g.\ $H(X_I)\ge 0\,\forall I\subset[n]$. For a long time it was not known whether there are more inequalities in addition to the Shannon type inequalities, i.e.\ the positivity conditions for the Shannon information measures introduced in the last section. It was not until almost fifty years after Claude Shannon's seminal work \cite{shannon1948mathematical} until Yeung and Zhang discovered a new, non Shannon type inequality \cite{zhang1997non}.
This new inequality is quite complicated, in its compact form using conditional mutual informations it reads
\begin{equation}
	I(X_1:X_2)+I(X_1:X_{34})+3I(X_3:X_4|X_1)+I(X_3:X_4|X_2)-2I(X_3: X_4)\ge 0.
\end{equation}
Expanding it into Shannon entropies we get the lengthy expression
\begin{equation}
	-H(1)-2H(3)-2H(4)-2H({12})+3H({34})+3H({31})+3H({41})+H({32})+H({42})-4H({134})-H({234})\ge0.
\end{equation}
None of the two expressions has an operational meaning easily accessible to understanding, nevertheless it can be shown that non Shannon-type information inequalities play a role, for example in entropic marginal problems \cite{fritz2011entropic} or network coding \cite{dougherty2007networks}.

Later Matu\v{s} found an infinite family of independent inequalities \cite{matus2007infinitely} for four or more random variables. In addition he proved that infinitely many of them define facets, proving that the cone is not \emph{polyhedral}, i.e.\ its base is not a polytope.

Let us formalize the notion of an information inequality. Looking at the vector space $ V_n$ that contains $\Sigma_n$ we see that we can identify a linear information inequality with an element of its dual space. An element $f\in  V_n^*\cong  V_n$ corresponds to a valid information inequality if
\begin{equation}
	\sum_{I\subset[n]}f_I H(X_I)\ge 0
\end{equation}
for all sets of random variables $(X_1,...,X_n)$. Or, without explicit reference to random variables, ergo in purely geometric terms, $f$ corresponds to a valid information inequality if and only if
\begin{equation}
	f(x)\ge 0\, \forall x\in \Sigma_n,
\end{equation}
which means that the dual cone $\Sigma_n^*$ is exactly the set of valid information inequalities. Let us call an information inequality \emph{essential for }$\Sigma_n$, if it is an extremal ray of $\Sigma_n^*$. That is equivalent to the fact that that it defines a facet of $\Sigma_n$.

There is an important subcone of $\Sigma_n^*$, that is the set of all balanced information inequalities. A functional $f\in  V_n^*$ is defined to be \emph{balanced}, if
\begin{equation}
	\sum_{I\ni i}f_I=0\, \forall \ i\in[n].
\end{equation}
The $n$ equations above are independent, i.e.\ they define a subspace $B_n\subset V_n^*$ with $\dim(B_n)=2^n-n$. The subset of valid information inequalities in $B_n$ is the cone of balanced information inequalities,
\begin{equation}
	\Sigma_{n,b}^*=\Sigma_n^*\cap B_n.
\end{equation}
Closely related is the notion of \emph{residual weights} introduced by Chan \cite{chan2003balanced}. Given a functional $f\in V_n^*$, its $i$th residual weight is defined by
\begin{equation}\label{resids}
	r_i(f)=\sum_{I\ni i}f_I.
\end{equation}
The definition is equivalent to saying that the $i$th residual weight of $f$ is defined by $r_i(f)=f(v^{(i)})$ with $v^{(i)}\in\Sigma_n$, $v^{(i)}_I=|\{i\}\cap I|$.

In 2002 Chan and Yeung proved a theorem that provides an algebraic characterization of the classical entropy cone by connecting entropies and subgroup sizes:
\begin{thm}[\cite{chan2002relation}]\label{chanyeung1}
	Let $X=(X_i)_{i\in [n]}$ be an $n$-partite random variable. Then there exits a sequence of tuples of finite groups $(G,G_1, G_2,...,G_n)_k, k\in\N$ with $G_i\subset G$ subgroups, such that
	\begin{equation}
		H(X_I)=\lim_{k\to\infty}\frac{1}{k}\log\frac{|G|}{\left|G_I\right|}\, \forall I\subset[n]
	\end{equation}
	where $G_I=\bigcap_{i\in I}G_i$.
	Conversely, for any group tuple $(G,G_1, G_2,...,G_n)$ there is a random variable $Y=(Y_I)_{I\subset [n]}$ such that
	\begin{equation}
		H(Y_I)=\log\frac{|G|}{\left|G_I\right|}
	\end{equation}
\end{thm}
On the one hand this is a very nice result as it provides us with an additional toolbox for attacking the entropy cone problem. On the other hand, finite groups are completely characterized indeed \cite{aschbacher2004status}, but this characterization is hugely complicated and suggests that one should not expect too much of a simplification switching from entropies to finite groups.

\subsubsection{Application: Network Coding}

Entropy inequalities are extremely useful in practice. For example they are the laws constraining \emph{network codes}. Although not widely used as of today, the current research effort indicates that network coding will be commercially applied in the future (see for example \cite{medard2011network}, Chapter 4.2 or \cite{heide2009network, pedersen2008implementation}). 

\paragraph{}The following short introduction to network coding is similar to the one in \cite{yeung1997framework}.
To get an idea how network coding can be useful let us first understand how almost the entire network infrastructure of today's world works. We can describe a network as a directed graph, where each vertex represents a node and each edge represents a channel. Each edge also has a number assigned to it which is the \emph{capacity}. The common \emph{store-and-forward} network architecture amounts to mere routing: A message is encoded by the sender node, then it is routed through the network to the receiver node, where it is decoded. This protocol is optimal for exactly one sender and one receiver being active in the network. But already when to nodes want to \emph{exchange} a pair of messages, there are conceivable network scenarios where a store-and-forward protocol cannot reach the maximum possible capacity.

\paragraph{}Network coding means that not only sender and receiver may perform coding operations, but also intermediate nodes. This provides an advantage in a variety of scenarios, one of which is described in the following paragraph.
\paragraph{Example}
To see how network coding protocols can outperform store-and-forward protocols \cite{yeung2008information} consider the following situation. Let Alice and Bob be situated on two different continents. They want to communicate over a satellite that can perform one of two operations per time interval, it can either receive a unit message from one sender or broadcast a unit message. This system can be described by the graph shown in Figure \ref{satellite-graph}.
\begin{figure}
\centering
 \includegraphics[width=7cm]{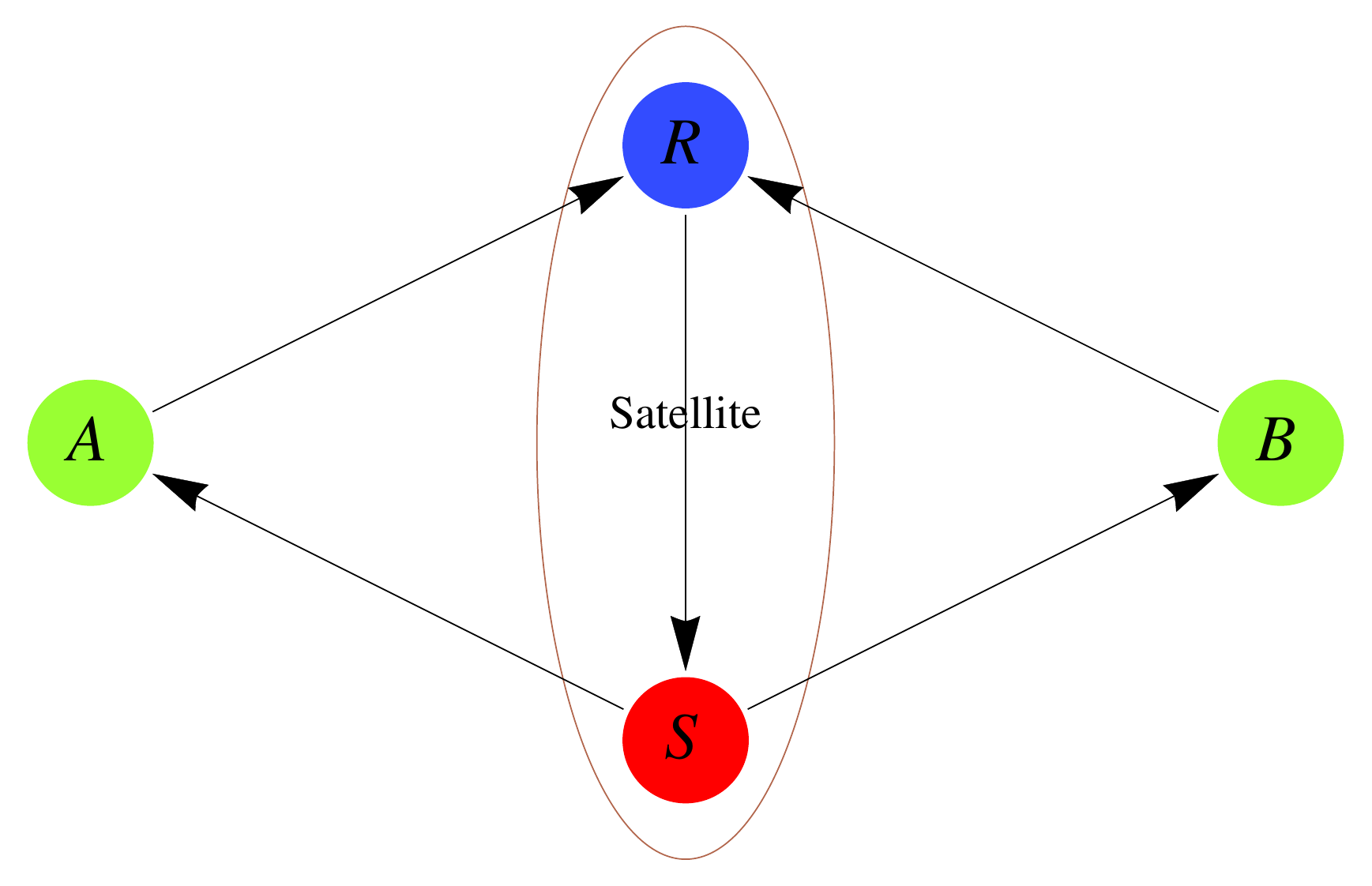}
 \caption{Directed graph representing a satellite communication scenario. 'A' and 'B' stand for Alice and Bob, the two communication partners, 'R' stands for receiver and 'S' for sender, and it is assumed that each channel represented by an arrow has unit capacity, R can receive from the left or from the right channel but not both and that only R or S can be active at the same time.}\label{satellite-graph}
\end{figure}
Now assume Alice and bob want to exchange unit messages $m_A$ and $m_B$. With a store-and-forward protocol this needs 4 time intervals: 
\begin{enumerate}
 \item A sends $m_A$ to the satellite
 \item B sends $m_B$ to the satellite
 \item The satellite broadcasts $m_A$
 \item The satellite broadcasts $m_B$
\end{enumerate}
However, if we allow the satellite to perform a very simple coding computation, the communication task can be completed within 3 time intervals:

\begin{enumerate}
 \item A sends $m_A$ to the satellite
 \item B sends $m_B$ to the satellite
 \item The satellite broadcasts $m_S=m_A\oplus m_B$
\end{enumerate}
Here $\oplus$ denotes modulo two addition. Alice can now decode $m_B=m_S\oplus m_A$ because she already has $m_A$ and Bob can in the same fashion compute $m_A=m_S\oplus m_B$.

\paragraph{}In this simple example network coding was able to outperform routing by 25

\paragraph{}
To mathematically formalize a general network coding scenario, let us recall some notions from graph theory.
\begin{defn}[Graph, Multigraph]
 A \emph{graph} $\mathcal{G}$ is a pair $\mathcal{G}=(V,E)$, where $V$ is a finite set called \emph{vertex set}, and $E\subset \left\{\{x,y\}\Big|x,y\in V\right\}$ is called \emph{edge set}. In a \emph{directed graph} the edge set contains ordered pairs, i.e.\ $E\subset V\times V$, we say an edge $e=(v,w)$ points from $v$ to $w$. In a (directed) \emph{multigraph} the edges from a multiset of (ordered) pairs from $V$.
\end{defn}
The pictures to have in mind reading this definition are shown in Figure \ref{graphs}.
\begin{figure}
\begin{center}
\subfigure[Graph]{\includegraphics[width=2.4in]{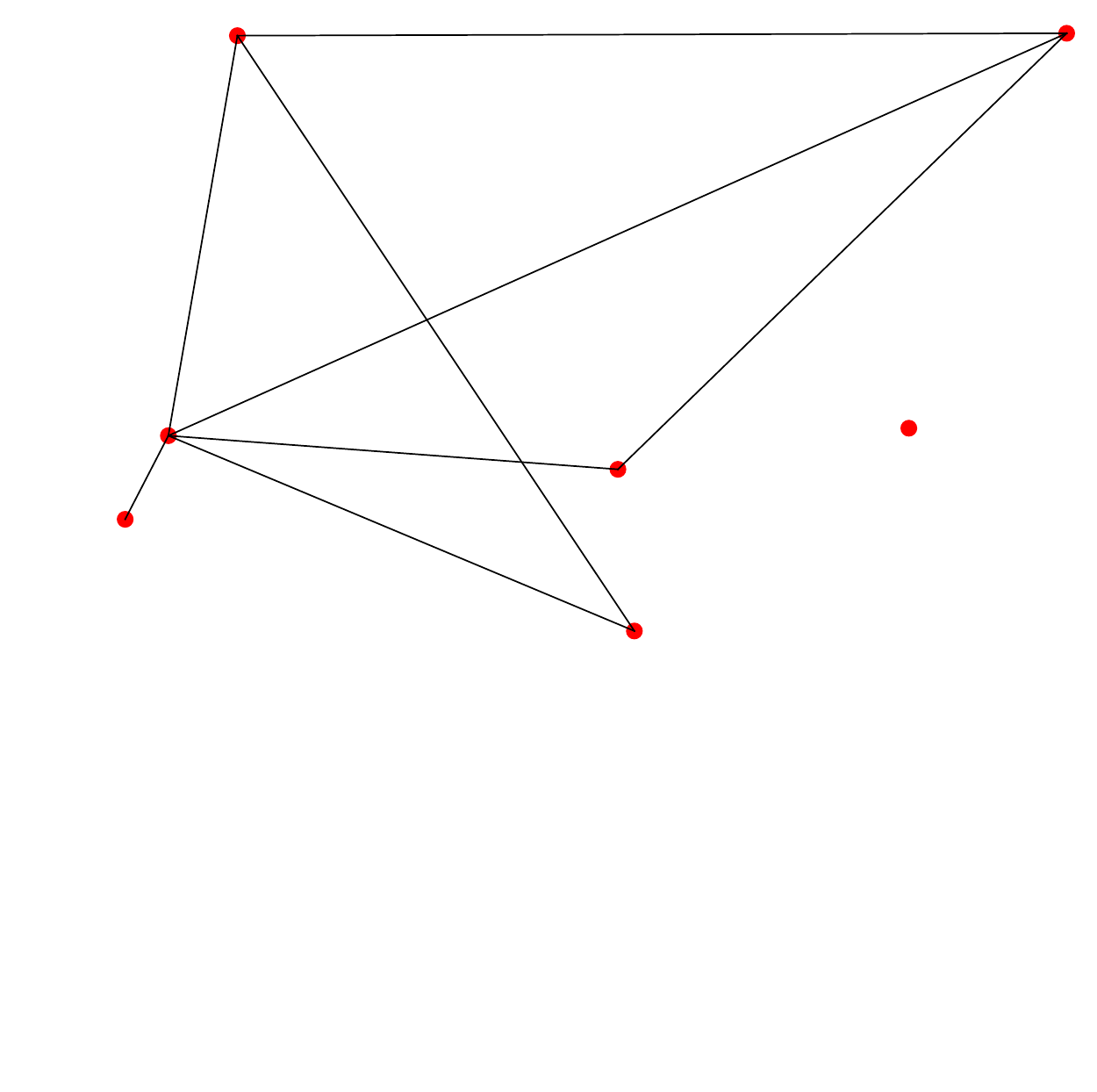}  \hspace{0.5cm}}
\subfigure[Directed graph]{\includegraphics[width=2.4in]{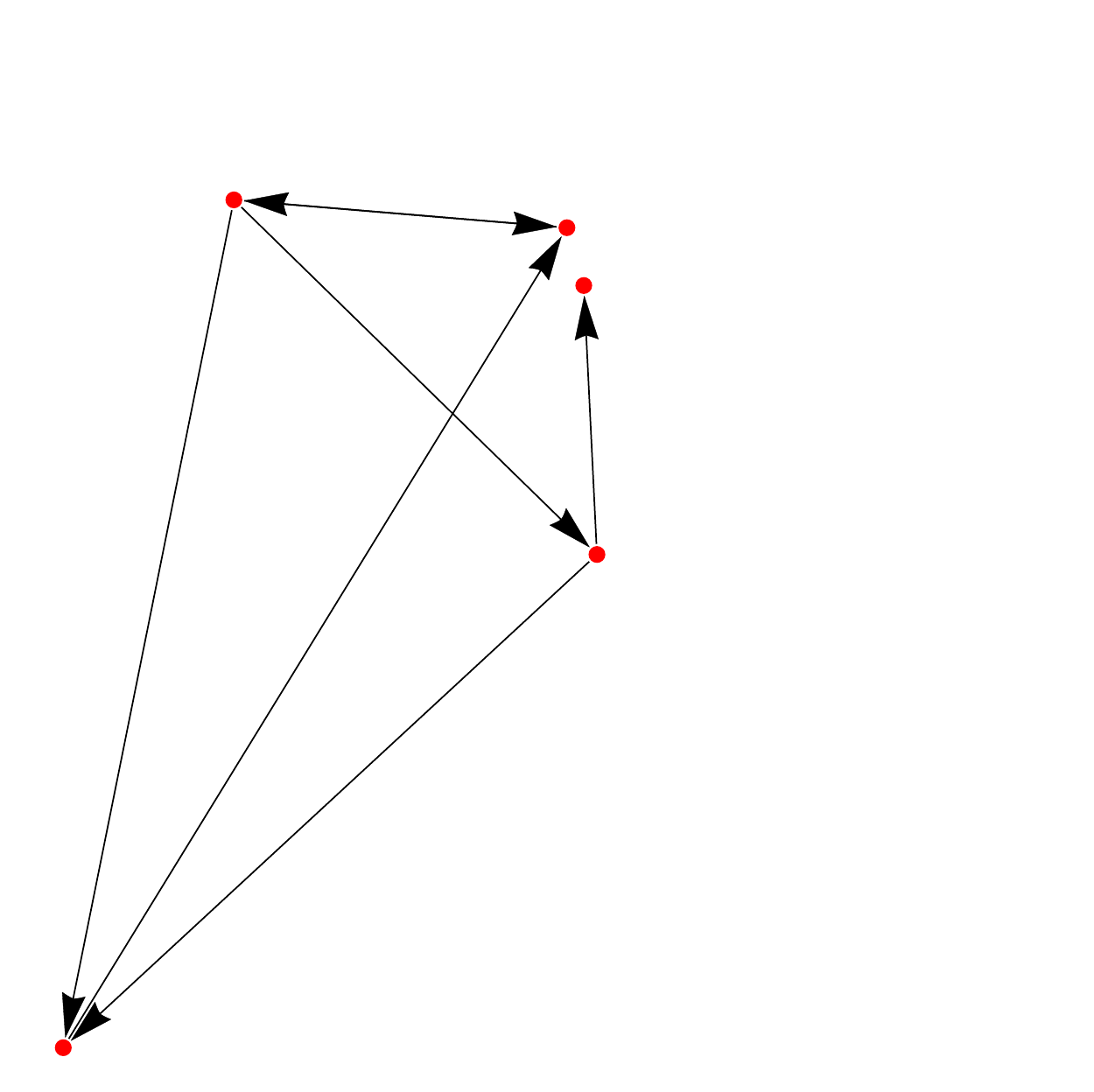} } \\
\subfigure[Multigraph]{\includegraphics[width=2.4in]{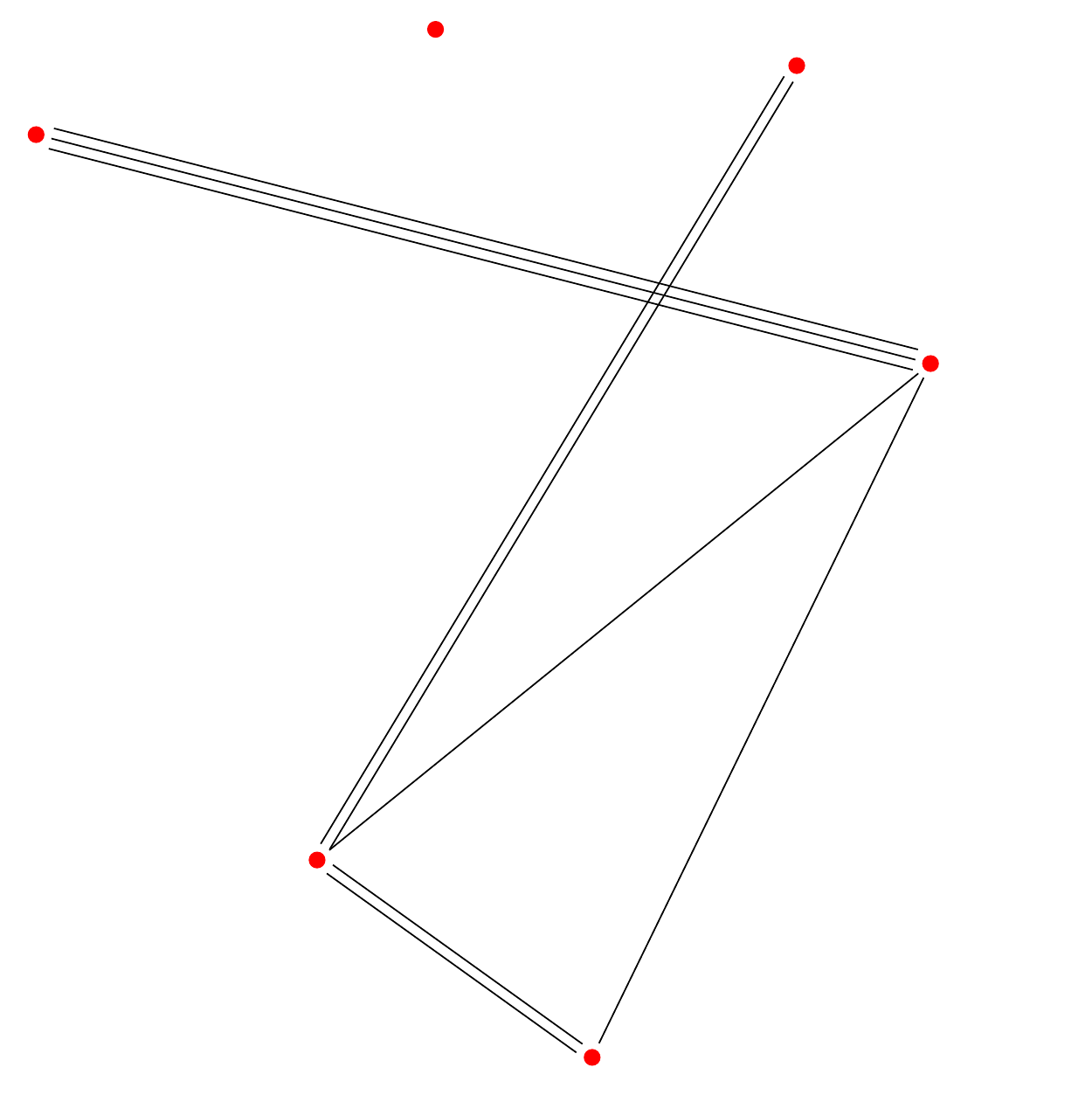}} \hspace{0.5cm}
\subfigure[Directed Multigraph]{\includegraphics[width=2.4in]{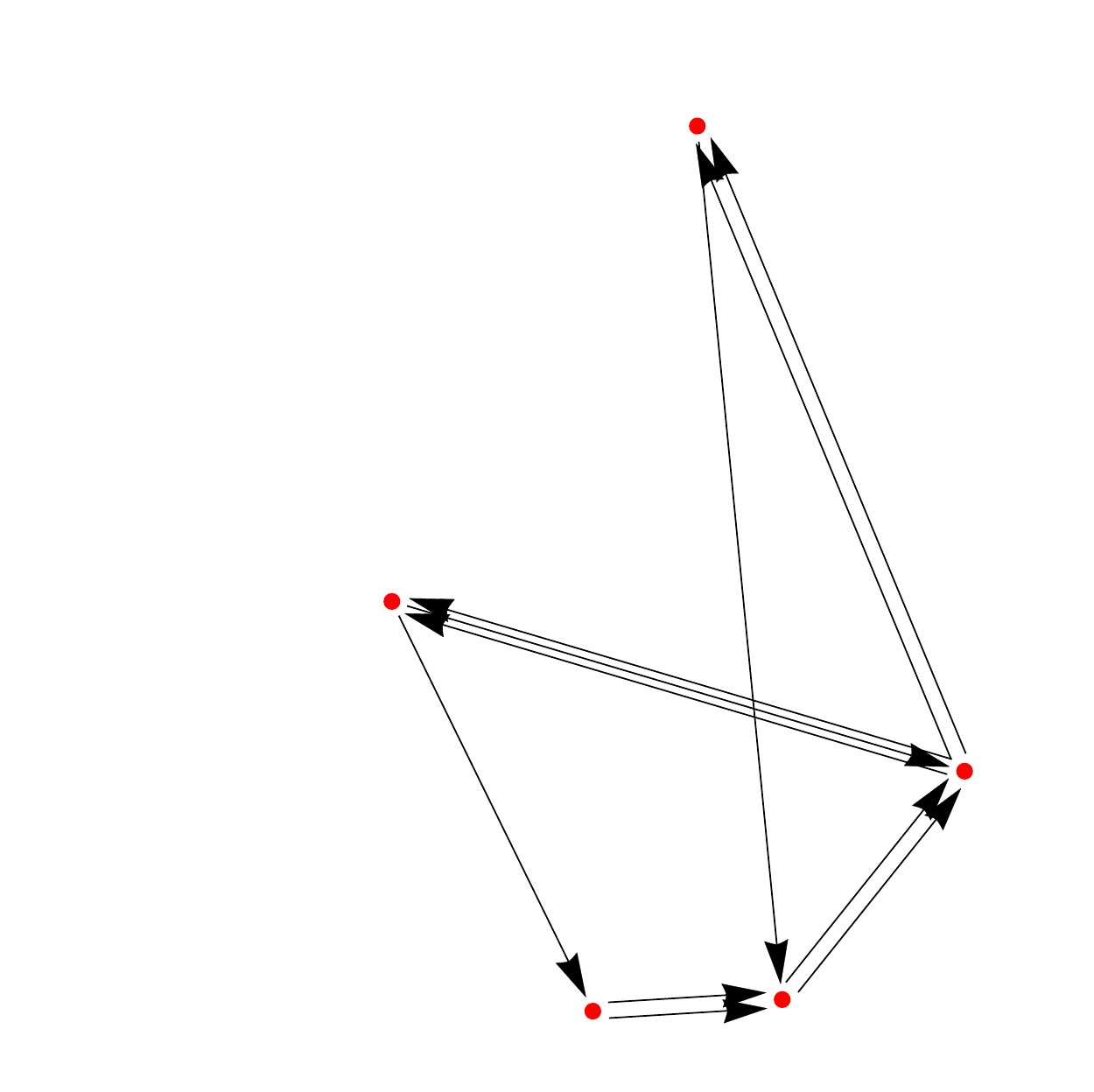}} 
\end{center}
\caption{Pictures of different graph notions.}
   \label{graphs}
\end{figure}
For a vertex $v\in V$ of directed graph $\mathcal{G}=(V,E)$ we denote by $\mathrm{in}(v)\subset E$ the set of edges pointing at $v$, and by $\mathrm{out}(v)\subset E$ the set of vertices pointing away from $v$.

A general network communication scenario can be described by the following data:
\begin{itemize}
 \item A directed graph $G={V,E}$. The vertices represent nodes in the network, the edges represent communication channels.
 \item A map $P: V\to\{\mathbf{s,n,t}\}$ which specifies whether a node is a source node ($\mathbf s$), a regular node ($\mathbf n$) or a target node ($\mathbf t$). define $S=P^{-1}(\{\mathbf s\})$ and $T=P^{-1}(\{\mathbf t\})$.
 \item A map $C: E\to \R_+$ which specifies the \emph{capacity} of each channel
 \item A map $\omega: S\to \R_+$ specifying the \emph{rate} of the sources. We write $\omega_s:=\omega(s)$. 
 \item A map $D: T\to 2^S$ to specify which target needs to receive which sources' information
\end{itemize}
For convenience of notation continue $\omega$ to all vertices by setting $\omega_v=0\,\, \forall v\not\in S$.
A Network code is now an assignment of random variables $(X_v)_{v\in V}$ and $(Y_e)_{e\in E}$ such that the following conditions are satisfied:
\begin{enumerate}
\item $H(X_S)=\sum_{s\in S}H(X_s)$
  \item $H(X_v)\ge \omega_v$
 \item $H\left(Y_e|Y_{\mathrm{in}(v)}, X_v\right)=0$ for all $e=(v,w)\in E$
 \item $H(Y_e)\le C(e)$
 \item $H\left(X_{D(t)}|Y_{\mathrm{in}(t)}\right)=0$
\end{enumerate}
These conditions mean that the source variables are independent (1.), that they can encode the amount of information given by the source rate (2.), that information transmitted through a channel should be a function of the information available at the sender node (3.), that the information send through a channel is bounded by its capacity (4.), and finally that the information intended for the target node $t$ is actually available there (5.). 

\paragraph{}Having found such random variables, we have solved the task to distribute the information as intended in one time step. We say a rate tuple $\omega=(\omega_s)_{s\in S}$ is \emph{asymptotically achievable} for a certain network with capacities $C$ if for all $\epsilon\ge 0$ there exists a $n\in\N$ such that there is a network code for the rate tuple $\left(n(\omega_s-\epsilon)\right)_{s\in S}$ and capacities $n C$. (Note that this notion is simplified compared to the one in \cite{yeung2008information} to concisely present the concepts. For a practically more relevant definition of achievable information rate see e.g.\ the aforementioned introductory text \cite{yeung2008information}.) The connection to the Shannon entropy cone becomes clear now. Let 
\begin{equation}
 \mathcal{L}_1:=\left\{h\in \R^{2^{|V|+|E|}}|h_S=\sum_{s\in S}h_s\right\}
\end{equation}
the set of vectors in $\R^{2^{|V|+|E|}}$ that satisfies condition 1, and let $\mathcal{L}_2$ to $\mathcal{L}_5$ be defined in an analogous way. Furthermore define the projection onto the source node entropies
\begin{eqnarray}
 \Pi_S: \R^{2^{|V|+|E|}}&\to&\R^{|S|}\nonumber\\
 h&\mapsto&(h_s)_{s\in S}.
\end{eqnarray}
Then a rate tuple $\omega$ is directly achievable if and only if
\begin{equation}
 \omega\in\Pi_S\left[\Sigma_{|V|+|E|}\cap\mathcal L_1\cap \mathcal L_3\cap \mathcal L_4\cap \mathcal L_5\right]
\end{equation}
A similar characterization result can be proven for asymptotically achievable rates (\cite{yeung2008information}, Theorem 21.5). The implications of this result are far-reaching: As pointed out in the introduction, we can expect network coding to be used in communication infrastructure in the not too far future. But we are far from being able to even determine the maximum achievable rate region for a general network, let alone finding an actual implementation that achieves it. This shows that the study of the Entropy cone is far from being of purely academic interest.

\subsection{Classical Marginal Problem}

Let us first look at a geometric marginal problem to get an idea what makes marginal problems so difficult.
\begin{figure}
\begin{center}
\subfigure{\includegraphics[width=4cm]{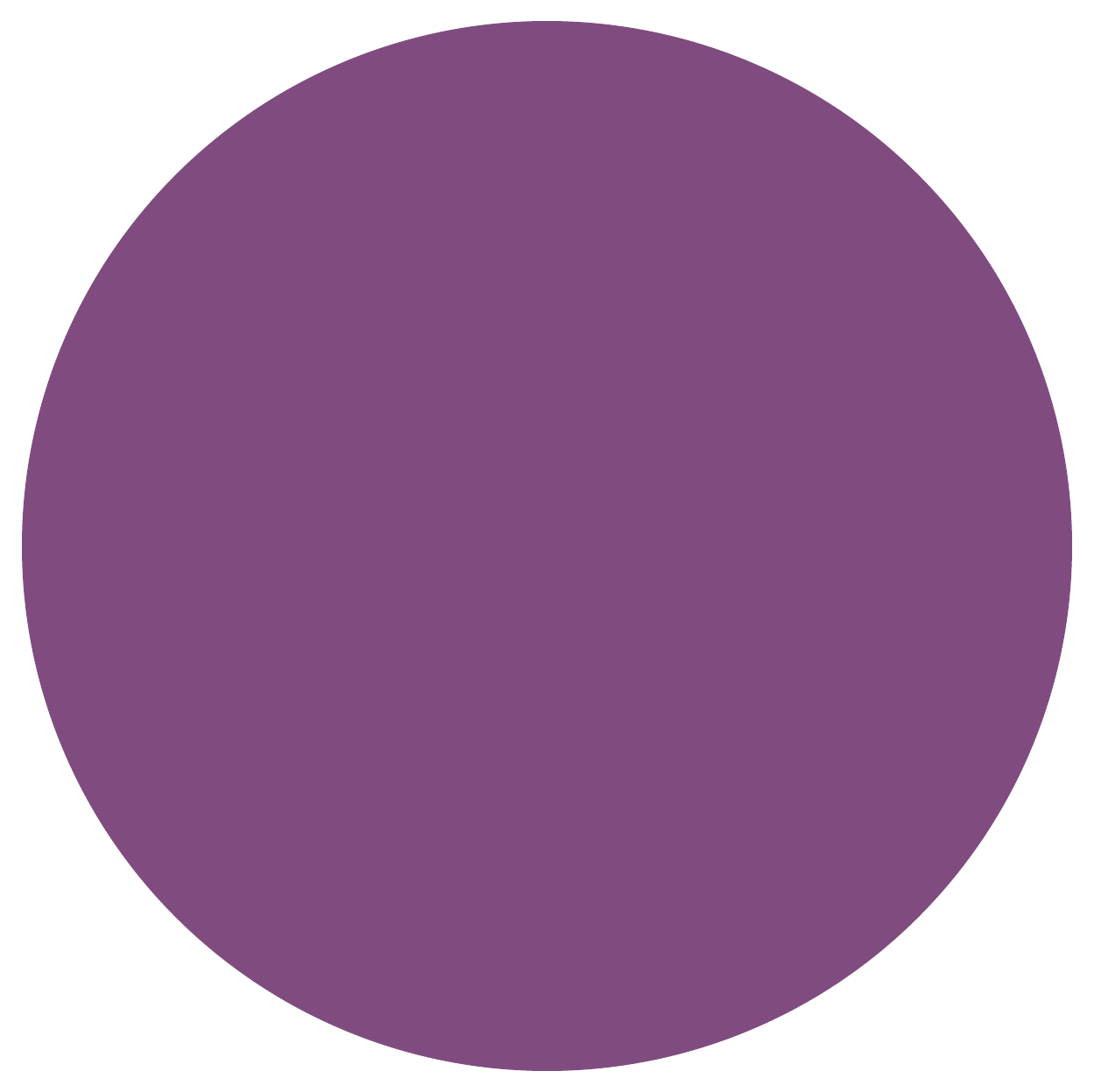} }
\subfigure{\includegraphics[width=4cm]{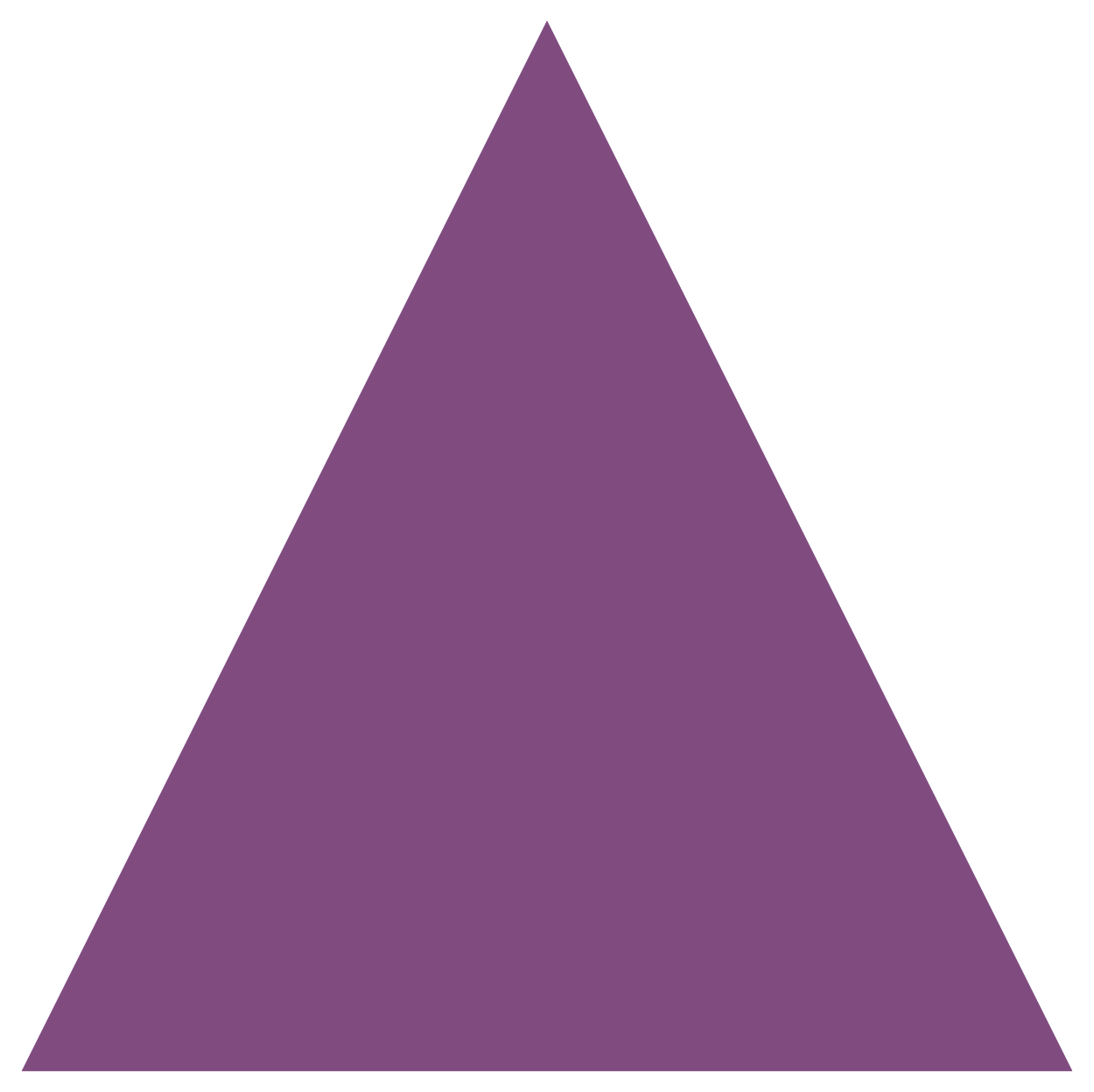}}
\subfigure{\includegraphics[width=4cm]{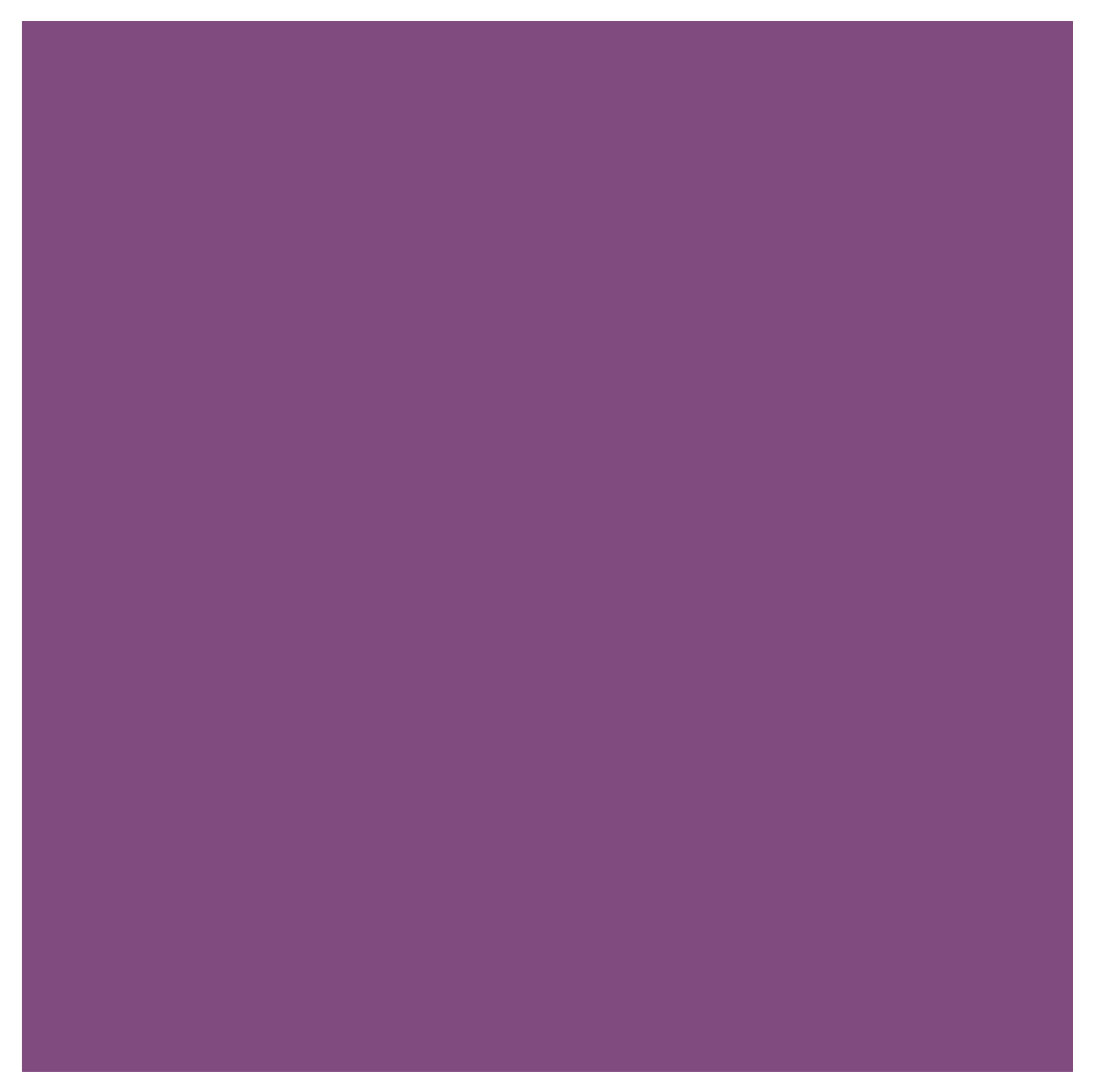} } \\
\subfigure{\includegraphics[width=4cm]{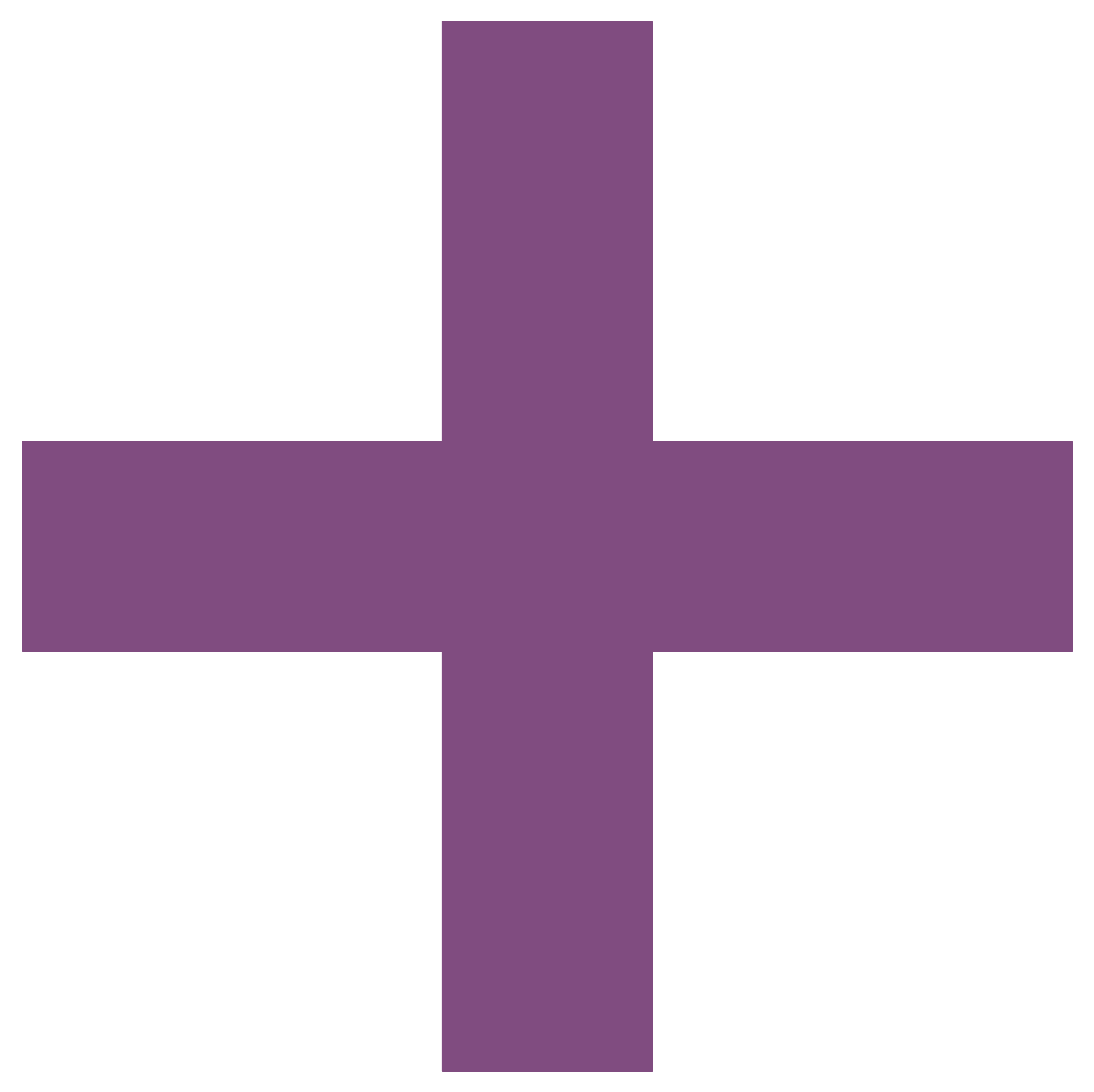}}
\subfigure{\includegraphics[width=4cm]{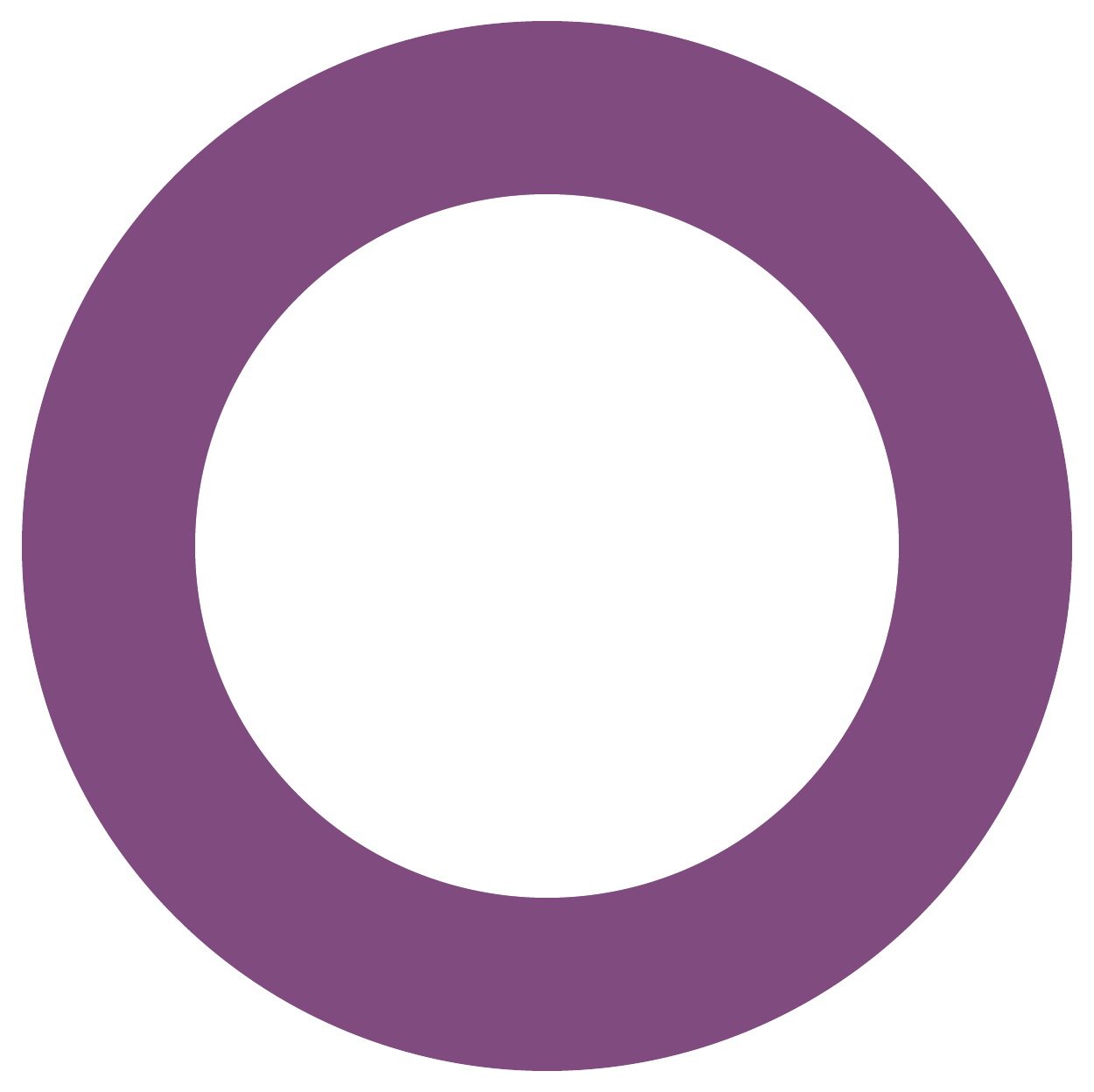}}
\subfigure{\includegraphics[width=4cm]{square}} 
\end{center}
\caption{Geometric marginal problem: Are there genuine three dimensional bodies that have a) the three shapes in the first row as coordinate plane projections, or b) the three shapes in the second row?}
   \label{projs}
\end{figure}
In figure \ref{projs} two triples of two dimensional geometric shapes are shown. Is there a three dimensional body such that the three shapes arise as the three projections onto the coordinate planes? We want an actual three dimensional body with no ``thin'' parts, i.e.\ the closure of the interior should contain the body itself. For the first triple that is certainly possible, the three-dimensional Body is shown in Figure \ref{MargBody}. For the second triple there seems to be no obvious solution.

\begin{figure}
\centering
 \includegraphics[width=7cm]{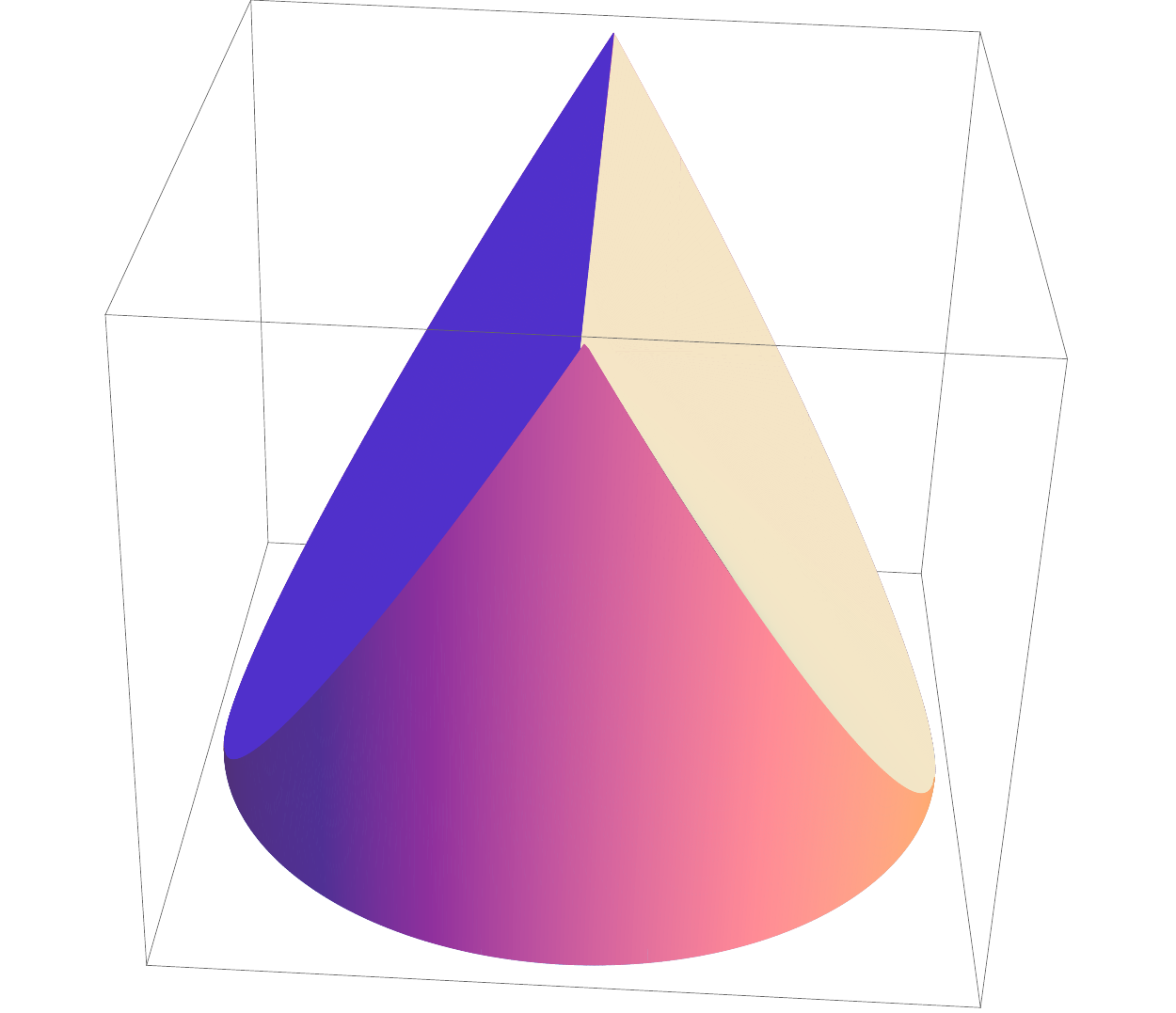}
 \caption{Solution to problem a) from Figure \ref{projs}}\label{MargBody}
\end{figure}

This simple-to-state geometric problem already captures the difficulty of marginal problems: The projections are not independent as overlapping dimensions survive. Finding a four dimensional body that has two given two dimensional projections is fairly easy, the Cartesian product of the two does the trick, which is possible because the two projections can be chosen orthogonal and thus independently controllable.

\paragraph{}The classical marginal problem is that of random variables, which can be stated in the following way: Given some probability distributions claimed to be the marginals of a global distribution, check whether a compatible global distribution exists. In other words, are the given distributions compatible with each other? \cite{fritz2011entropic} describes a couple of examples in which situations marginal problems arise, e.g.\ when investigating privacy issues when anonymizing data from databases, in artificial intelligence or when studying quantum non-locality. In the following paragraph I will describe a classic scenario from the latter field of research as an example.

\paragraph{Example: Bell Inequalities}

One of the counterintuitive features of quantum mechanics is that generically measuring an observable of a system also changes the state of the system. This implies in particular that the results of different measurements on the same system cannot be obtained unless many copies of the state are available. The outcome of a quantum measurement constitutes a random variable. Now consider an $n$-partite quantum state. Measurements on different subsystems commute and can therefore be performed simultaneously. Let us assume that each subsystem admits a number $k$ of different non-commuting measurements, So we get $nk$ random variables $X_{i,j},\ 1\le i\le n,\ 1\le j\le k$, one for each system and each measurement. But we can only obtain the probability distributions $p_{X_{1,j_1}, ..., X_{n,j_n}}$ where for each subsystem only one of the $k$ measurements is considered. The question whether these distributions can arise from a joint probability distribution of all $nk$ random variables is a marginal problem. 

The famous Bell inequalities \cite{bell1964einstein} are the affine inequalities that define the boundary of the image of the set of all possible distributions of $nk$ random variables (defined on a fixed Alphabet), the \emph{probability simplex}, under the linear marginalization map that maps the global distribution to the set of jointly observable marginals.

\paragraph{}This example also shows the connection to entropy inequalities. The Shannon entropy is, as described above, a function of the probability distribution, so it is not surprising that Bell inequalities have non-trivial corollaries in terms of entropies \cite{braunstein1990wringing}. In \cite{chaves2012entropic} and \cite{fritz2011entropic} it is described how in principle a complete set of entropic Bell inequalities can be obtained from a complete description of the Shannon entropy cone.

\paragraph{}Let us now give a more formal definition of the classical marginal problem:

\begin{qu}[Classical Marginal Problem]\label{classmarg}
Let $\mathcal{A}\subset 2^{[n]}$ be a subset of the power set of the $n$ element set, $\mathcal{X}$ an alphabet and let for each $I\in\mathcal{A}$ $p_I: \mathcal{X}^{|I|}\to[0,1]$ be a probability distribution. Do there exist random variables $X_1,...,X_n$ on $\mathcal X$such that for all $I\in \mathcal{A}$ $X_I=(X_i)_{i\in I}$ is distributed according to $p_I$?
\end{qu}

Solving this problem is equivalent to characterizing the image of the probability simplex under the marginal map

\begin{eqnarray}
 m_\mathcal{A}: \mathcal{P}^{|\mathcal X|^n}&\to &\bigoplus_{I\in\mathcal{A}} \mathcal{P}^{|\mathcal X|^{|I|}}\nonumber\\
 \left(p:\mathcal X^n\to[0,1]\right)&\mapsto&\left(p_I: \mathcal X ^{|I|}\to[0,1]\right)_{I\in \mathcal A},\\
 \text{with }p_{I}(x)&=&\sum_{\substack{y\in\mathcal X^{n}\\y_I=x}}p(y)\nonumber
\end{eqnarray}

Calculating the image of a polytope under a linear map is a fairly easy computational task, but anyway problematic in high dimensions. In Chapter \ref{CYR} I will give a connection  of this problem to representation Theory.


\section{Quantum Information Theory}

Quantum information theory is the mathematical framework for utilizing quantum mechanical systems for information processing. In this chapter I want to introduce the mathematical concepts relevant for this thesis. The substructure resembles the one of the last chapter: First, I will establish in brevity the fundamentals of quantum information theory, in the subsequent section I introduce the von Neumann entropy which plays a similar role as the Shannon entropy does in Classical information theory, and eventually I describe the quantum entropy cone and review some of its properties.

A great introduction to quantum information theory can be found in \cite{nielsen2010quantum}, in the following I introduce the basic concepts as they can be found there. In quantum information theory, states are positive semidefinite operators $\rho$ on a finite-dimensional, complex Hilbert space $\mathcal{H}$ that have unit trace, i.e.
\begin{equation}
 \rho\in\left\{\sigma\in \hom(\mathcal{H})\Big|\bra{\psi}\sigma\ket{\psi}\ge 0\,\forall \ket{\psi}\in \mathcal{H},\ \tr\sigma=1\right\}=:\mathcal{B}(\mathcal{H})
\end{equation}
where $\bra{\psi}$ is the dual vector of $\ket\psi$ employing Dirac notation. The operator $\rho$ is called \emph{density operator} of the quantum system. A state is called \emph{pure} if it has rank one, otherwise it is called \emph{mixed}. A pure state $\rho=\ketbra{\psi}{\psi}$ can also be represented as unit vector $\ket{\psi}\in\mathcal{H}$. 

In all quantum theories measurement plays a crucial role. In quantum information theory, in particular, it is important as only classical information is human readable and the measurement is the way the extraction of classical information from a quantum system can be achieved. Mathematically a measurement is specified by a set of measurement operators $M_i, i=1,...,l$ such that
\begin{equation}\label{measnorm}
 \sum_i M_i^\dagger M_i=\mathds{1}.
\end{equation}
The probability, that outcome $i$ occurs when the quantum system that is measured is in state $\rho$ is given by
\begin{equation}
 p_i=\tr(M_i^\dagger M_i\rho).
\end{equation}
the resulting probability distribution is normalized because of the unit trace condition on $\rho$ and Equation \eqref{measnorm}. One of the main differences between classical and quantum theories is, that the measurement process affects the state of the system. The post-measurement state is given by
\begin{equation}
 \rho'=\sum_i M_i\rho M_i^\dagger.
\end{equation}

Quantum information theory is a generalization of classical information theory. A random variable $X$ on an alphabet $\mathcal{X}$ with probability distribution $p$ corresponds to a density operator $\rho_X$ on $\mathcal{H}_\mathcal{X}=\C\mathcal{X}$ that is diagonal in the defining basis, i.e.
\begin{equation}
 \rho_X=\sum_{x\in \mathcal{X}}p(x)\ketbra{x}{x}
\end{equation}
A measurement with measurement operators $M_x=\ketbra{x}{x},\ x\in \mathcal{X}$ recovers the random variable $X$, this is also called ``measuring the basis $\{\ket{x}|x\in \mathcal{X}\}$''.

A composite system that consist of several distinct subsystems is described by a tensor product Hilbert space $\mathcal{H}=\mathcal{H}_1\otimes \mathcal{H}_2\otimes ...\otimes\mathcal{H}_n$. Given a state on that product space, $\rho\in \mathcal{B}(\mathcal{H})$, and a subset $I\subset[n]$, what is the state on $\mathcal{H}_I=\bigotimes_{i\in I}\mathcal{H}_i$, the analogue of the marginal distribution of a probability distribution? It should be an map that sends $\rho$ to the \emph{reduced density operator} $\rho_I$ such that
\begin{equation}
 \tr(A\rho_I)=\tr(A\otimes\mathds{1}_{I^c}\rho)
\end{equation}
for all positive semidefinite operators $A$ on $\mathcal{H}_I$, where $I^c=[n]\setminus I$. That map is the so called \emph{partial trace}. To define it, let us for a moment index every trace operator by the Hilbert space it is defined on, i.e.\ the last equation becomes
\begin{equation}
 \tr_{\mathcal{H}_I}(A\rho_I)=\tr_{\mathcal{H}}(A\otimes\mathds{1}_{I^c}\rho).
\end{equation}
Then the partial trace is defined as the tensor product of the identity on the Hilbert spaces where the reduced density operator is defined on and the trace on the remaining ones, i.e.
\begin{eqnarray}
 \rho_I&=&\tr_{I^c}\rho\nonumber\\
 \tr_{I^c}&=&\mathds{1}_{\mathcal{H}_I}\otimes\tr_{\mathcal{H}_{I^c}}.
\end{eqnarray}

The way random variables are used in classical information theory can be a bit confusing. Most of the time the randomness of a random variable is interpreted as \emph{potential} information. A communication channel, for example, is is not used to transmit random data but its designer treats the data as random variable $X$ such as to build the capabilities to transmit any dataset from the support of $p_X$ etc. Putting it in yet another way, the actual message will be known to the sender, so for him it is in a deterministic state, i.e.\ in a state that is extremal in the convex set of states, the probability simplex. The channel designer assumes a weighted average of the ensemble of messages he expects the user to send, i.e.\ a convex combination of deterministic states.

To generalize this formalism to quantum information theory we observe that the extremal points in the quantum state space are pure states. A mixed state is a convex combination of pure states and can be interpreted as representing an ensemble of pure states analogously to the random variable being interpreted as representing an ensemble of deterministic states.

An important result that turned out to be a powerful proof technique in the quantum marginal problem \cite{christandl2006spectra,christandl2012recoupling} to be introduced in Section \ref{quantmarg} is the so called \emph{spectrum estimation theorem} that was first discovered in many body theory \cite{alicki1988symmetry}. Later it was rediscovered independently in quantum information theory \cite{keyl2001estimating}. It is a quantum version of the asymptotic equipartition property where the role of type classes is played by the \emph{typical subspaces}, the direct summands in the Schur-Weyl decomposition \eqref{schurweyl}.

It states that a high tensor power $\rho^{\otimes n}$  of a density matrix $\rho\in\mathcal{B}\left(\C^d\right)$ is supported mostly on Subspaces $[\lambda]\otimes V_\lambda \subset \left(\C^d\right)^{\otimes n}$ such that $\overline\lambda:=\frac{\lambda}{n}$ is close to the spectrum of $\rho$. In Section \ref{symirreps} the construction of the irreducible representations of the Symmetric group by decomposing the permutation modules $M^\lambda$ into irreducible representations is described, the connection between frequencies, which determine the type class, and partitions, which determine the typical subspace, becomes apparent there. Also the connection between the two concepts is elaborated in Chapter \ref{CYR}.

\begin{thm}[\cite{alicki1988symmetry},\cite{keyl2001estimating}]\label{specest}
	Let $\rho\in\mathcal{B}(\mathcal{H})$ be a density operator on some finite dimensional Hilbert space $\mathcal{H}$ with spectrum $r=\spec(rho)$. Then, for any projector $P_\lambda$ onto a direct summand in $\eqref{schurweyl}$
	\begin{equation}
		\tr P_\lambda\rho^{\otimes k}\le (k+1)^{d(d-1)/2}e^{-k H\left(\overline{\lambda}\|r\right)},
	\end{equation}
	where $\overline{\lambda}=\frac{\lambda}{k}$ and $H(\cdot\|\cdot)$ is the classical relative entropy defined in Equation \eqref{relent}
\end{thm}
A concise proof for this theorem can be found for example in \cite{christandl2006structure}.

\subsection{Von Neumann Information Measures}\label{vNmeasures}

The natural generalization of the Shannon entropy is the von Neumann entropy named after John von Neumann who solidified the mathematical framework of quantum mechanics \cite{neumann1955mathematical}. It is defined as
\begin{equation}
 S(\rho)=-\tr\rho\log\rho
\end{equation}
for a quantum state given by a density operator $\rho$. As easily verified, it is the only quantum generalization possible if we demand the following two reasonable properties:
\begin{itemize}
 \item For classical states, i.e.\ for diagonal density operators, the quantum entropy has to coincide with the Shannon entropy.
 \item The quantum entropy has to be basis independent, i.e.\ invariant under unitary conjugation of $\rho$,
\end{itemize}
or, more formally put,
\begin{eqnarray}\label{vN}
 S\left(\sum_i p_i \ketbra{i}{i}\right)&=&H(p)\, \forall p\in\mathcal{P}^{\dim\mathcal{H}}\ \text{and}\nonumber\\
  S(\rho)&=&S(U\rho U^\dagger)\,\forall\rho\in\mathcal{B}(\mathcal{H}),\, U\in \mathrm U(\mathcal{H}),
\end{eqnarray}
where $\mathrm U(\mathcal H)$ is the group of unitary transformations on $\mathcal{H}$. This fixes $S$ to \eqref{vN}, as any density operator can be diagonalized by a unitary.

\paragraph{}The von Neumann entropy is arguably as important of a concept for quantum information theory as the Shannon entropy is for classical information theory. Analogous to the Shannon entropy the prime justification of the von Neumann entropy as a measure of information is coding, as Shannon's noiseless channel coding theorem can be generalized to coding a source of quantum states:
\begin{thm}[Schumacher’s noiseless channel coding theorem \cite{schumacher1995quantum}]
 Given a source of pure quantum states from a Hilbert space $\mathcal{H}$ distributed i.i.d.\!\!\! according to the density operator $\rho$, then for each $R<S(\rho)$ there exists a protocol to compress the states with rate $R$. If ever $R>S(\rho)$ no such scheme exists.
\end{thm}

\paragraph{}The Shannon information measures have a natural generalization to the quantum theory in terms of von Neumann entropies. For a tripartite state $\rho_{ABC}\in\mathcal{H}^{\otimes3}$ let $\rho_{AB}=\tr_C\rho_{ABC}$ etc. We define the \textit{conditional von Neumann entropy}, the \textit{quantum mutual information} and the \textit{quantum conditional mutual information} by the classical formulas with the Shannon entropy replaced by the Von Neumann entropy:
\begin{eqnarray}
	S(A|B)_\rho&=&S(\rho_{AB})-S(\rho_B)\\
	I(A:B)_\rho&=&S(\rho_A)+S(\rho_B)-S(\rho_{AB})\\
	I(A:B|C)_\rho&=&S(\rho_{AC})+S(\rho_{BC})-S(\rho_{ABC})-S(\rho_C)
\end{eqnarray}
If there is no danger of confusion we denote $S(\rho_A)=S_\rho(A)=S(A)$ and omit the subscript in the von Neumann information measures. Note that, although the mathematical generalization is straightforward, the classical interpretation cannot be generalized to the quantum case in a simple way.
Lieb and Ruskai proved that the quantum conditional mutual information is nonnegative \cite{lieb1973proof}, this result is called \emph{strong subadditivity},
\begin{equation}\label{SSA}
	S(AB)+S(BC)-S(ABC)-S(B)\ge 0.
\end{equation}

\paragraph{}Although the operational meaning of the information measures does not generalize to quantum entropies in a straightforward way, strong subadditivity has many important applications in quantum information theory. For example it turns out that if a bipartite quantum system $AB$ is shared between two parties $A$ and $B$, than the mutual information between the two is equal to the amount of classical information that can be send from $A$ to $B$ securely using \emph{one time pad} encryption \cite{schumacher2007quantum}. Strong subadditivity of the von Neumann entropy implies a result adding plausibility to this interpretation: The quantum mutual information does not increase when a local operation is performed on one of the two systems \cite{nielsen2010quantum}. This result is called data processing inequality and is only one of many inequalities relying on strong subadditivity.

\paragraph{}From strong subadditivity the only known convex independent quantum information inequality can be derived by considering a purification party, that is \emph{weak monotonicity},
\begin{equation}\label{wmo}
	S(AB)+S(BC)-S(A)-S(C)\ge 0
\end{equation}
which replaces the classically valid monotonicity, $H(A|B)\ge 0$. 
Let us shortly recall the possibility of purification and how weak monotonicity follows from strong subadditivity and vice versa. Consider a Hilbert space $\mathcal{H}$ and an arbitrary state $\rho\in\mathcal{B}(\mathcal{H})$. Let $ \rho=\sum_i p_i\ketbra{i}{i}$ be the spectral decomposition of $\rho$. Define the state $\ket{\psi_\rho}=\sum_i\sqrt{p_i}\ket{i}\otimes\ket{i}\in\mathcal{H}^{\otimes 2}$, then $\rho=\tr_2\ketbra{\psi_\rho}{\psi_\rho}$. On the other hand, for any pure state $\ket\Psi\in\mathcal{H}_1\otimes\mathcal{H}_2$ on a bipartite Hilbert space we have the so called \emph{Schmidt decomposition}, that is bases $\{\ket{\phi_i}|i\in[d_1]\}$ of $\mathcal{H}_1$ and $\{\ket{\psi_i}|i\in[d_2]\}$ of $\mathcal{H}_2$ such that 
\begin{equation}
\ket{\Psi}=\sum_{i=1}^{\min(d_1,d_2)}\alpha_i\ket{\phi_i}\otimes\ket{\psi_i}.
\end{equation}
That implies in particular that the spectra and hence the entropies of the reduced states $\rho_i,\ i=1,2$ are the same, i.e.
\begin{equation}\label{bipure}
	S(\rho_1)=S(\rho_2).
\end{equation}
Weak monotonicity follows now from purifying a tripartite state $\rho_{ABC}$ by using another system $D$, i.e.\ $\rho_{ABC}=\tr_D\ketbra{\Psi_{ABCD}}{\Psi_{ABCD}}$, and then eliminating the occurrence of the system $A$ in \eqref{SSA} by means of \eqref{bipure}.

\subsection{The Quantum Entropy Cone}\label{quantum}

In analogy to the classical entropy cone, define the set of entropy vectors of $n$-partite states by
\begin{equation}
	\Gamma_n=\left\{s(\rho)\subset V_n\Big|\rho\in \mathcal{B}\left(\mathcal{H}_1\otimes...\otimes\mathcal{H}_n\right)\right\},
\end{equation}
where $s(\rho)=\left(S(\rho_I)\right)_{I\subset[n]}$ is the entropy vector of $\rho$. If $\rho=\ketbra{\psi}{\psi}$ is a pure state, we write $s(\ket{\psi}):=s(\ketbra{\psi}{\psi})$. 

Pippenger proved that this is a convex cone as well \cite{pippenger2003inequalities}, following Zhang's and Yeung's argument \cite{zhang1997non} for the classical case. Let us review the proof for that fact.
\begin{thm}[\cite{pippenger2003inequalities}]\label{qcone}
	$\Gamma_n$ is a convex cone.
\end{thm}
\begin{proof}
	Let $v,w\in\Gamma_n$ be the entropy vectors of $\rho, \sigma \in\mathcal{H}^{\otimes n}$ respectively. Then $\rho\otimes\sigma\in \left(\mathcal{H}\otimes \mathcal{H}\right)^{\otimes n}$ has the entropy vector $v+w$. This proves additivity. For approximate diluability let $v\in\Gamma_n$ be the entropy vector of $\rho \in \mathcal{H}^{\otimes n}$ and $\epsilon>0$. Now take $0<\delta\le\frac{1}{2}$ such that $h(\delta):=-\delta\log \delta -(1-\delta)\log(1-\delta)\le \epsilon$. For $0<\lambda \le \delta$ take $\mathcal{H}'=\mathcal{H}\oplus\C$ and $\rho'=\lambda \rho\oplus (1-\lambda )\ketbra{0}{0}^{\otimes n}\in\mathcal{H}'^{\otimes n}$. Direct calculation shows that $v'_I=S(\rho'_I)=\lambda S(\rho_I)+h(\lambda)\, \forall I\subset[n]$, which implies $\norm{v'-\lambda v}_\infty=h(\lambda)\le h(\delta)\le \epsilon$.
\end{proof}
Analogously to the classical case, the dual cone $\Gamma_n^*$ is the set of all valid quantum information inequalities, also define the balanced subcone $\Gamma_{n,b}^*=\Gamma_n^*\cap B_n$. For strong subadditivity and weak monotonicity we introduce the notation
\begin{eqnarray}\label{ssa}
	\Delta[I,J]&=&S(I)+S(J)-S(I\cup J)-S(I\cap J)\\
	E[I,J]&=&S(I)+S(J)-S(I\setminus J)+S(J\setminus I)\label{wm}
\end{eqnarray}
Let $\Xi_n$ be the cone defined by the inequalities \eqref{ssa} and \eqref{wm}. We call this cone the \emph{von Neumann cone}. Pippenger identified the extremal rays of $\Xi_n^*$:
\begin{prop}[Pippenger, Corollary 3.6 in \cite{pippenger2003inequalities}]\label{pipdualextrays}
	The set of extremal rays of the dual of the von Neumann cone is $\mathrm{ext}(\Xi_n^*)=\mathcal{E}_\Delta\cup \mathcal{E}_E$ with $\mathcal{E}_\Delta=\{\Delta[I,J]|I,J\subset[n],\ I\setminus J=\{i\},\  J\setminus I=\{j\},\  i<j\}$ and $\mathcal{E}_E=\{E[I,J]|I,J\subset[n],\ I\cap J=\{k\},\ I\cup J=[n],\  k+1\in I\}$.
\end{prop}
The number of essential inequalities is $\left|\mathcal{E}_\Delta\right|=\frac{n(n-1)}{2}2^{n-2}=n(n-1)2^{n-3}$ and $\left|\mathcal{E}_E\right|=n2^{n-2}$ respectively \cite{pippenger2003inequalities}.

Let us define some subcones of $\Gamma_n$. First, we can look at the set of entropy vectors of symmetric states,
\begin{equation}
	\widehat{\Gamma}^s_n=\left\{s(\rho)\Big|\rho\in \mathcal{B}\left(\mathcal{H}^{\otimes n}\right), \phi(\sigma)\rho\phi(\sigma)^\dagger=\rho\right\}
\end{equation}
where $\phi(\sigma)$ is the natural unitary representation of $S_n$ on $\mathcal{H}^{\otimes n}$ permuting the tensor factors. As easily checked, its closure is a convex cone as well.
\begin{cor}
	$\overline{\widehat{\Gamma}^s_n}$ is a convex cone.
\end{cor}
\begin{proof}
Follow the proof of Theorem \ref{qcone} and check that every step conserves the symmetry properties of the involved density matrices.
\end{proof}
The same is true for the set of symmetric entropy vectors, defined by
\begin{equation}
	\widehat{\Gamma}^\sigma_n=\left\{s(\rho)\Big|\rho\in \mathcal{B}\left(\mathcal{H}^{\otimes n}\right), s(\rho)_I=s(\rho)_J \text{ if }|I|=|J|\right\}.
\end{equation}
Obviously $\widehat{\Gamma}^s_n\subset \widehat{\Gamma}^\sigma_n$. Both cones can be mapped bijectively into the lower-dimensional space $\R^n$ by defining
\begin{equation}
	\Gamma^s_n=\left\{\left(S(\rho_{[i]})\right)_{i\in[n]}\Big|\rho\in \mathcal{B}\left(\mathcal{H}^{\otimes n}\right), \phi(\sigma)\rho\phi(\sigma)^\dagger=\rho\right\}
\end{equation}
and analogously $\Gamma^\sigma_n$. Pippenger found the extremal rays of the cone $\Gamma^\sigma_n$ of symmetric entropy vectors \cite{pippenger2003inequalities}, proving that there are no inequalities other than strong subadditivity and weak monotonicity for symmetric entropies.

For the further characterization of the quantum entropy cone there are two main courses of action one can follow: either try to prove that there are more inequalities, of prove that the extremal rays of the cone generated by the known inequalities are extremal rays of the quantum entropy cone, i.e.\ can be approximated by von Neumann entropy vectors. In Chapter \ref{differential} I present some results in the direction of the second path by characterizing states that populate extremal rays using their local geometry.

\subsection{The Quantum Marginal Problem}\label{quantmarg}

The quantum marginal problem is the quantum version of the classical marginal problem described in Section \ref{classmarg} and is closely related to the characterization of the quantum entropy cone. It asks whether a quantum state exists that has certain reduced density matrices. In the following I will introduce the quantum marginal problem in a formal way. In Chapter \ref{CYR}, which is mostly about representation theory and the classical marginal problem, the following definition will be used. An introduction can be found, for example, in \cite{klyachko2004quantum}.

There are many variants in which the problem can be stated, one of which is the following.
\begin{qu}[Quantum Marginal Problem]
 Let $\mathcal{A}\subset 2^{[n]}$ a subset of the power set of the $n$ element set and $s_I$ a spectrum for each $I\in \mathcal{A}$. Is there a quantum state $\rho$ on some $n$ factor tensor product Hilbert space such that for all $I\in\mathcal{A}$ the spectrum of $\rho_I=\tr_{[n]\setminus I}\rho$ is $s_I$, possibly padded with zeros?
\end{qu}

Note that the Hilbert space dimension is not a problem here. As we allow for padding of the spectra with zeros, a larger Hilbert space is no problem and we can just take $\mathcal{H}=\left(\C^d\right)^{\otimes n}$ with $d\ge\max_{I\in \mathcal A}|s_I|^{\frac{1}{|I|}}$.

An equivalence of this problem to representation theoretic problems has been shown for bipartite and tripartite quantum states in \cite{christandl2006spectra} and \cite{christandl2012recoupling}.


\chapter{Entropy Cones and their Morphisms}\label{sec:mor}

In the following chapter I want to illuminate the geometric properties of the entropy cones $\overline\Sigma_n$ and $\overline\Gamma_n$. After briefly discussing relations between the entropy cones of a different number of particles, I will first investigate the symmetries of the cones. I prove, that the quantum entropy cone has a strictly larger symmetry group than its classical analogue, this group is identified and has a clear physical interpretation. As a corollary I can show, that a certain subset of the weak monotonicity facets of the von Neumann cone identified by Pippenger \cite{pippenger2003inequalities} are facets of the quantum entropy cone itself. In the succeeding section I review a result by Chan \cite{chan2003balanced} that reduces the characterization problem to balanced information inequalities. I translate it into purely geometric language which makes it look much simpler: It states that the dual of the Shannon entropy cone is a direct sum of two cones. I show that this simplifying property is missing for the quantum entropy cone, i.e.\ a result like Chan's cannot be achieved. The proof of this fact heavily relies on the results concerning the symmetry group. This shows that the quantum entropy cone is more symmetric than the classical one, but also more complicated. Finally I shortly review a class of morphisms considered by Ibinson \cite{ibinson2008quantum}, thus completing the collection of geometric results concerning the von Neumann cone.

As we discuss general properties of entropy cones in this Chapter, we denote an entropy cone by $\Lambda_n$ if we do not want to specify whether it is classical or quantum. 

\paragraph*{}Given $n\le m$ we have natural morphisms between $\Lambda_n$ and $\Lambda_m$. The most obvious ones are the surjection
\begin{equation}
	\mathrm{sur}_m^n:V_m\to V_n, v=(v_I)_{I\subset[m]}\mapsto w \text{ with } w_J=v_{J} \text{ for } J\subset[n]\subset[m]
\end{equation}
and the injection
\begin{equation}
	\mathrm{inj}_n^m: V_n\to V_m, v=(v_I)_{I\subset[n]}\mapsto w \text{ with } w_J=v_{I\cap J}.
\end{equation}
They correspond to discarding the systems $n+1$ to $m$ and adding trivial systems to a random variable or density matrix generating $v$ respectively. For $\Gamma_n$ there is another natural injection that is given by purification, i.e.
\begin{equation}
	\mathrm{pur}^{n+1}_n: V_n\to V_{n+1}, v=(v_I)_{I\subset[n]}\mapsto w \text{ with } w_J=\begin{cases}
	                                                                               	v_J&J\not \ni n+1\\
	                                                                               	v_{J^c}&\mathrm{else}
	                                                                               \end{cases}.
\end{equation}
Note that purification acts linearly in entropy space while it is a nonlinear map on the state level. The purification map is an isomorphism between $\Gamma_n$ and $\Gamma_{n+1}^p$, the entropy cone of $n\!+\!1$-partite pure states. 

\section{Symmetries}\label{sym}

The following section is dedicated to clarifying symmetry properties of classical and quantum entropy cones using the cone morphism formalism developed in Section \ref{conv}.

The symmetric group $S_n$ acts linearly on the entropy space in a natural way by permuting the subsystems of the state that generates the entropy vector, i.e.
\begin{equation}\label{wGamsym}
	S_n\looparrowright V_n, \sigma\cdot v=\left(w_{\sigma^{-1}I}\right)_{I\subset[n]}
\end{equation}
where $\sigma^{-1}(I)=\left\{\sigma^{-1}(i)|i\in I\right\}$. Of course the entropy cone is invariant under this action, $\sigma\left(\Lambda_n\right)=\Lambda_n$. This implies that the image of $S_n$ under this action is part of the automorphism group of the cone. For $\Gamma_n$, however, this action can be extended to an action of $S_{n+1}$ by considering a purifying system:
\begin{prop}\label{Gammasym}
 The map
 \begin{eqnarray}\label{gamsym}
 S_{n+1}&\looparrowright &V_n\\
  (\pi \cdot v)_I&=&\begin{cases}
             v_{\pi(I)^c}& n+1\in \pi(I)\\
             v_{\pi(I)}& \text{else}
            \end{cases}
 \end{eqnarray}
defines a linear group action of $S_{n+1}$ on $V_n$, and the quantum entropy cone $\overline \Gamma_n$ is invariant under this action.
\end{prop}
Note that the subgroup of permutations that fix $n+1$ generates the permutation action under which also the classical entropy cone is invariant.
\begin{proof}
The action \eqref{gamsym} can be constructed by first applying the purification map, then the usual permutation action in $V_{n+1}$ and then applying the surjection onto the original space, i.e.
\begin{equation}\label{Gamsym}
	 \sigma\cdot v=\mathrm{sur}_{n+1}^n\left(\sigma\cdot\mathrm{pur}^{n+1}_n(v)\right),
\end{equation}
As easily verified, $\left(\mathrm{pur}^{n+1}_n\circ \mathrm{sur}_{n+1}^n\right)\big|_{\overline \Gamma_{n+1}^p}=\mathds 1$, so
\begin{equation}
 (\pi\sigma)\cdot v=\pi\cdot (\sigma\cdot v),
\end{equation}
i.e.\ \eqref{gamsym} indeed defines a group action. The image of the quantum entropy cone under this action is again the whole quantum entropy cone, as by definition $\pi.\Gamma_n\subset\Gamma_n$ and $\pi^{n!}=\mathds 1$, so each $\pi\in S_n$ defines an isomorphism of $\overline \Gamma_n$.
\end{proof}

Explicitly that yields for example
\begin{equation}
	\left((1\ n+1)\cdot v\right)_I=\begin{cases}
	                               	v_{I^c\cup\{1\}}& 1\in I\\
	                               	v_I&\text{else}
	                               \end{cases}
\end{equation}
Note that the corresponding representation is faithful, as there is no permutation $\sigma\in S_{n+1}$ that replaces a number of subsets of $[n+1]$ by their complement and does nothing else, i.e.\ that there exists a subset $\mathfrak{M}\subset 2^{[n+1]}$ such that $\sigma(J)=J^c$ for $J\in \mathfrak{M}$ and $\sigma(J)=J$ else. The action \eqref{Gamsym} defines an adjoint action on $\Gamma_n^*$ as well that is defined naturally as
\begin{equation}\label{adGamsym}
	[\sigma\cdot f](v):=f(\sigma^{-1}\cdot v)
\end{equation}
With the help of this compact notation we have, for example, $\Delta[\{12\},\{23\}]=(14)\cdot E[\{13\},\{23\}]$ as elements of $\Gamma_3^*$, where here $\cdot$ denotes the adjoint action. This demystifies the equivalence proof below Equation \eqref{wmo} as being a consequence of the symmetry:
\begin{cor}
 The orbit of any non-trivial strong subadditivity functional under the adjoint action \eqref{adGamsym} contains a weak monotonicity functional, and vice versa.
\end{cor}
\begin{proof}
 For any $I,J\subset [n]$ such that $I\setminus J\neq\emptyset\neq J\setminus I$ and $I\cap J\neq\emptyset$, we have 
 \begin{equation}
  (i\ n+1).\Delta[I,J]=E[I^c\cup\{1\},J]
 \end{equation}
 and
 \begin{equation}
  (i\ n+1).E[I,J]=\Delta[I^c\cup\{1\},J]
 \end{equation}
for any $i\in I\setminus J$.
\end{proof}

Exploiting the symmetry \eqref{Gamsym} the set of known independent information inequalities can be reduced to, for example, the set
\begin{equation}
	\mathcal{E}'_\Delta=\{\Delta[\{1,...,k\},\{l,...,m\}]|1<l\le k<m\le n,\ k\ge m-l+1\}.
\end{equation}

\paragraph*{}Using this new symmetry we can prove, that some weak monotonicity inequalities define facets of $\overline \Gamma_n$.

\begin{cor}\label{wmfacets}
 The inequalities $E[{ij},i^c],\ i,j\in [n]$ define facets of the quantum entropy cone, i.e.\ they are in particular extremal rays of its dual.
\end{cor}
\begin{proof}
 First observe that
 \begin{equation}\label{ssatowm}
  E[\{ij\},i^c]=(i\ n+1)\cdot \Delta[i^c,j^c]
 \end{equation}
 employing the action from Proposition \ref{Gammasym}. Now look at the entropy vectors
 \begin{equation}
  v^{(I)}_J=\min(1,|I\cap J|),\ J\neq \emptyset
 \end{equation}
which correspond to random variables some subset of which are maximally correlated and the rest are trivial. They are linearly independent according to the proof of \ref{dimSig}, and
\begin{equation}
  \Delta[i^c,j^c](v^{(I)})=\delta_{I\,\{ij\}},
\end{equation}
which proves that the face defined by $\Delta[i^c,j^c]$ has dimension one less then the whole cone, i.e.\ it is a facet. But due to the relation \eqref{ssatowm} this implies that $E[\{ij\},i^c]$ defines a facet as well.
\end{proof}

\section{Balanced Information Inequalities}\label{bal}

In his paper on classical balanced information inequalities \cite{chan2003balanced}, Chan introduces a way of balancing a possibly unbalanced information inequality, i.e.\ the linear projection
\begin{equation}
	\Pi_b: V_n^*\to V_n^*, \ f\mapsto g,\ g_I=\begin{cases}
	                                          	f_I-\sum_{i=1}^n r_i(f)& I=[n]\\
	                                          	f_I+r_i(f)& I=i^c\\
	                                          	f_I& \text{ else}
	                                          \end{cases},
\end{equation}
using the notion of residual weights defined in Equation \eqref{resids}
His main result is the following
\begin{thm}[Chan, \cite{chan2003balanced}]\label{chan}
	The following two statements are equivalent:
	\begin{enumerate}
		\item[(i)] $f\in \Sigma_n^*$
		\item[(ii)] $\Pi_b f\in \Sigma_{n,b}^*$ and $r_i(f)\ge 0$ for all $i\in [n]$.
	\end{enumerate}
\end{thm}
In particular it implies that $\Pi_b$ is a morphism from $\Sigma_n^*$ to $\Sigma_{n,b}^*$. For some information inequality $\sum_{I\subset [n]}f_I H(X_I)\ge 0$ this means that it is valid if and only if $r_i(f)\ge 0$ for all $i\in[n]$ and
\begin{equation}\label{chanspib}
	\sum_{I\subset [n]}f_I H(X_I)-\sum_{i\in[n]}r_i(f) H(X_i|X_{i^c})\ge 0
\end{equation}
is valid. Define the special monotonicity functionals used in the original definition \eqref{chanspib} of $\Pi_b$, i.e.
\begin{equation}
	m(i,i^c)_I=\begin{cases}
	         	1 & I=[n]\\
	         	-1 & I=i^c\\
	         	0 & \text{else}
	         \end{cases}.
\end{equation}
Let $M_{n-1}=\cone\left(\left\{m(i,i^c)|i\in[n]\right\}\right)$ be the cone generated by the $m(i,i^c)$. Theorem \ref{chan} implies the following
\begin{cor}
	$\Sigma_n^*=\Sigma_{n,b}^*+M_{n-1}$, and $\Pi_b$ projects onto $\Sigma_{n,b}$, i.e.\ for any element $f\in\Sigma_n^*$, $f=g+h$ with $\ g\in\Sigma_{n,b}^*$ and $\ h\in M_{n-1}$, $\Pi_b g=g$ and $\Pi_b h=0$.
\end{cor}

\begin{proof}
	Given $f\in\Sigma_n^*$ we have 
	\begin{equation}\label{chandecomp}
	f=\Pi_b f+(\mathds{1}-\Pi_b)f,
	\end{equation}
	$\Pi_b f\in\Sigma_{n,b}^*$ and $(\mathds{1}-\Pi_b)f\in M_{n-1}$. Given $g\in \Sigma_{n,b}^*$ and $h=\sum_{i\in[n]}\alpha_i m(i,i^c)\in M_{n-1}$, $r_i(g+h)=\alpha_i>0$ and $\Pi_b(g+h)=g\in\Sigma_{n,b}^*$ so, according to Theorem \ref{chan}, $g+h\in\Sigma_n^*$.
\end{proof}
The following Lemma gives a geometrical interpretation of the result.
\begin{lem}\label{projequivdirsum}
	Let $V$ be a real vector space and $K_1, K_2\subset V$ convex cones. Then $\spa K_1\cap \, \spa K_2=\{0\}$ if and only if there exists a linear map $P$ with $P a=a\, \forall a\in K_1$ and $P b=0\, \forall b\in K_2$
\end{lem}
\begin{proof}
	Let $\spa K_1\cap \spa K_2=\{0\}$. Take bases $\{e^{(1)}_i\}$ of $\spa K_1$ and $\{e^{(2)}_j\}$ of $\spa K_2$, then any $x \in V$ can be decomposed in a unique way as $x=\sum_i\alpha_i e^{(1)}+\sum_j\beta_j e^{(2)}_j+r$ with $r\in K^\perp$, $K=K_1+K_2$, and $P: x\mapsto \sum_i\alpha_i e^{(1)}$ has the required properties.
	
	Given a map $P$ as specified in the Lemma, we have, due to linearity, $P a=a\, \forall a\in \spa K_1$ and $P b=0\, \forall b\in\spa K_2$, so for any $x\in\spa K_1\cap \spa K_2$, $0=P x=x$.
\end{proof}
If $K=K_1+K_2$ and $\spa K_1\cap\spa K_2=\{0\}$ the sum of the two cones is \emph{direct} meaning that any element $a\in K$ can be written as sum $a=a_1+a_2, a_i\in K_i$ in a unique way. We write $K=K_1\oplus K_2$ then.
\begin{cor}\label{sigmastardirsum}
 \begin{equation}
  \Sigma_n^*=\Sigma_{n,b}^*\oplus M_{n-1}
 \end{equation}

\end{cor}

Let us prove another Lemma that relates the sets of extremal rays of the cones appearing in the previous lemma.
\begin{lem}\label{sumextrays}
	Let $V$ be a real vector space, $K_1, K_2\subset V$ convex cones, $K=K_1+K_2$ with the following properties:
	\begin{enumerate}
		\item[\emph{(i)}]$\spa K_1\cap \, \spa K_2=\{0\}$
		\item[\emph{(ii)}]$\R v\not \subset K\, \forall v\in V$
	\end{enumerate}
Then $\mathrm{ext}(K)=\mathrm{ext}(K_1)\cup \mathrm{ext}(K_2)$.
\end{lem}
Note that assumption (i) implies in particular that $\mathrm{ext}(K_1)\cap \mathrm{ext}(K_2)=\emptyset$.
\begin{proof}We have to prove two inclusions to show equality.
	\begin{itemize}
             	\item ${\mathrm{ext}(K)\subset\mathrm{ext}(K_1)\cup \mathrm{ext}(K_2)}$:
             	\item[]Take any extremal ray $R\subset K$ of $K$ and $0\not=v\in R$. Because of (i) there is a unique sum decomposition $v=a+b,\ a\in K_1,\ b\in K_2$. But $v$ is in an extremal ray of $K$, so $a, b\in R$. Suppose that $a\not=0\not=b$, then $R\subset K_1\cap K_2$, which is a contradiction to (i). Hence either $a=0$ or $b=0$. Let without loss of generality $b=0$, then $a=v\in K_1$. As $R$ is an extremal ray, we have $a,b\in R$ for all $a,b\in K$ such that $a+b=v$. In particular, if we have $a,b \in K_1\subset K$ with $a+b=v$ it follows that $a,b\in R$. so $R$ is an extremal ray of $K_1$.
		\item ${\mathrm{ext}(K)\supset\mathrm{ext}(K_1)\cup \mathrm{ext}(K_2)}$:
		\item[]Let $R\subset K_1$ now be an extremal ray of $K_1$, $0\not=v\in R$. Take any $a,b\in K$ such that $a+b=v$. Let $a=a'+a'',\ b=b'+b''$ with $a',b'\in K_1,\ a'',b''\in K_2$ be the unique sum decompositions of $a$ and $b$. Then $a'+b'=v,\ a''+b''=0$. Suppose that $a''\not=0$, then
		\begin{equation}
		K\supset (\R_{\ge 0}a''\cup \R_{\ge 0}b'')=(\R_{\ge 0}a''\cup \R_{\ge 0}(-a''))=\R a'',\nonumber
		\end{equation}
		which is a contradiction to (ii). Hence $a''=b''=0$ and $a,b\in  K_1$. But $R$ is an extremal ray of $K_1$, which implies $a,b\in R$, so $R$ is an extremal ray of $K$.
             \end{itemize}
\end{proof}
Together with Theorem \ref{chan} the previous lemma implies that an imbalanced information inequality is essential if and only if it is of the from $H(X_i|X_{i^c})\ge 0$. The essential weak monotonicity instances in the von Neumann cone (see Proposition \ref{pipdualextrays}) span a space that has a non-trivial intersection with the balanced subspace of $V_n^*$, therefore a result like Chan's cannot be achieved provided that these instances of monotonicity are essential for $\Gamma_n$.
\begin{thm}\label{nodirsum}
	Let $K\subset V_n$ be a cone with $(e^{(\emptyset)},v)=0$ for all $v\in K$ and $\dim K=2^n-1$ such that there exist $f_i\in \mathrm{ext}(K^*), i=1,...,l$ with $f_i\notin B_n$ but $\spa(f_i|i=1,...,l)\cap B_n\not= 0$. Then there is no cone $C\subset V_n^*$ such that $\spa(C)\cap B_n=\{0\}$ and $K^*=(K^*\cap B_n)+C$.
\end{thm}
\begin{proof}
	Let $V=\spa\{e^{(I)}|\emptyset\not=I\subset [n]\}\subset V_n$. Then $K\subset V$ and we denote by $K'^*=K^*\cap V^*$ the dual of $K$ in $V$. We have $\dim K=\dim V$ and therefore $K'^*$ does not contain any non-trivial subspaces. Suppose there was a cone $C$ as specified in the theorem and denote $C'=C\cap V$. Then, according to Lemma \ref{sumextrays}, 
	\begin{eqnarray}
	  \mathrm{ext}(C')&=&\mathrm{ext}(K'^*)\setminus \mathrm{ext}(K'^*\cap B_n)\nonumber\\
	  &\supset&\left\{f_1,...,f_l\right\}.\nonumber
	\end{eqnarray}
	But then
	\begin{eqnarray}
	 0\not&=&\spa(f_i|i=1,...,l)\cap B_n\nonumber\\
	 &\subset&\spa\ \mathrm{ext}(C')\cap B_n\nonumber\\
	 &=&0,\nonumber
	\end{eqnarray}
which is a contradiction.
\end{proof}

\begin{figure}[h]
	\centering
	\includegraphics[width=6cm]{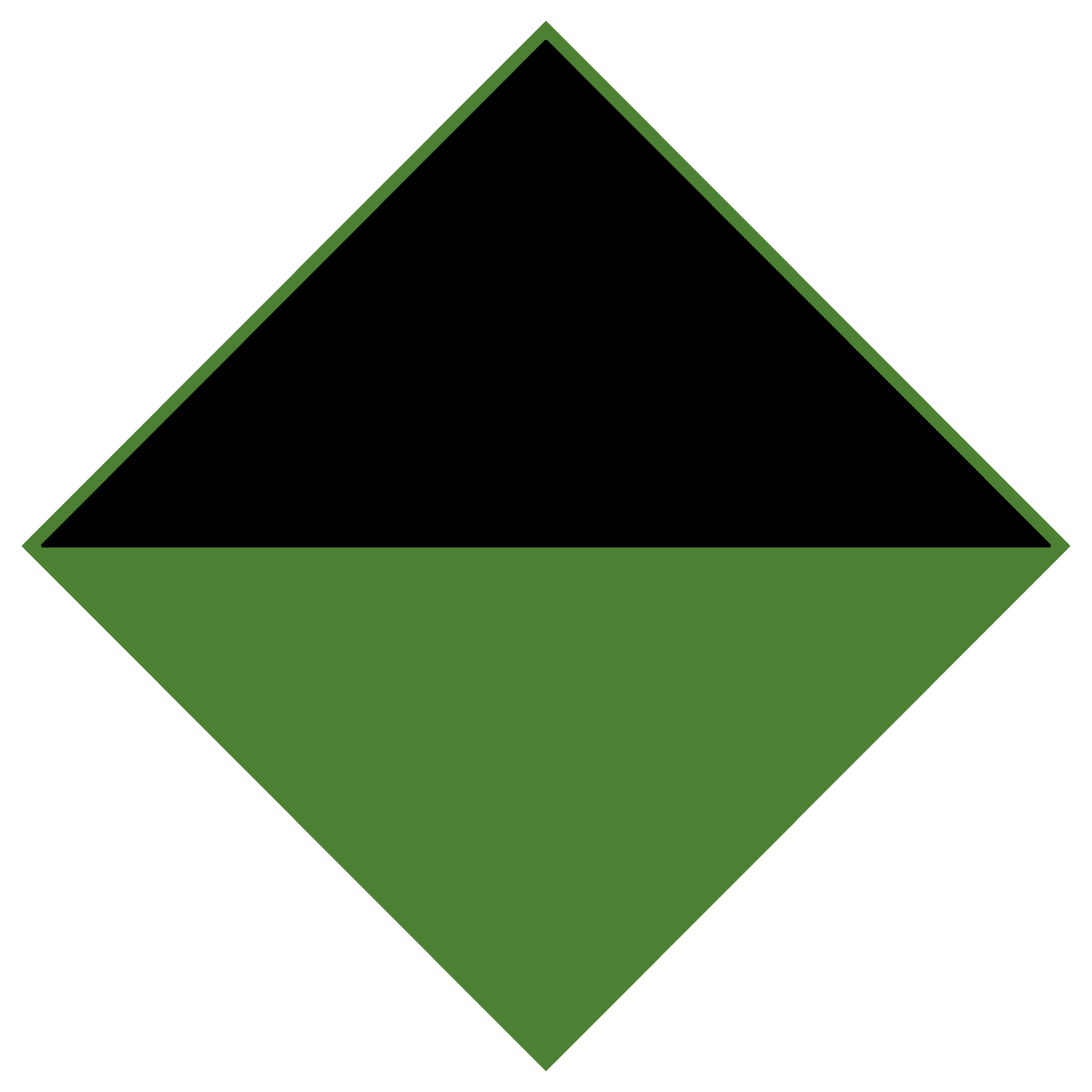}
	\caption{A simplified picture of the difference between the duals of the classical and quantum entropy cone: The 3-dimensional cone generated by the black base is a direct sum of two cones, whereas the one generated by the green base has a larger symmetry group.}\label{conescaricature}
\end{figure}

$\Gamma_n$ satisfies the first two conditions of the theorem, as $\Sigma_n\subset\Gamma_n\subset \spa\{e^{(I)}|\emptyset\not=I\subset [n]\}$, so, with Proposition \ref{dimSig}, $\dim \Gamma_n =2^n-1$. for the remaining condition, the existence of extremal rays of $\Gamma_n^*$ that span a space with non-trivial intersection with $B_n$, we need to show that some weak monotonicity instances that allow for a balanced linear combination are facets.

By choosing a suitable pair of strong subadditivity instances and using the symmetry \eqref{Gamsym}, we can show that $\Gamma_n$ fulfills the conditions of Theorem \ref{nodirsum} by showing that the chosen strong subadditivities define facets and using the cone isomorphism property of the symmetry transformation.
\begin{cor}\label{gammanodirsum}
	$\Gamma^*_n$ is not a direct sum of its balanced part and some other cone, and Chan's theorem \ref{chan} cannot be generalized to the case of quantum information inequalities.
\end{cor}
\begin{proof}
	Consider the weak monotonicity instances $E[\{12\}, 1^c]$ and $E[\{23\}, 3^c]$. Their difference is nonzero and balanced, and they are extremal rays of $\Gamma_n^*$ according to Corollary \ref{wmfacets}, i.e.\ Theorem \ref{nodirsum} applies.
\end{proof}

Figure \ref{conescaricature} shows a three-dimensional analogue of the situation of the dual cones $\Sigma_n^*$ and $\Gamma_n^*$.

\section{Symmetrically Sub-dividable Entropies}\label{sub}
In his thesis \cite{ibinson2008quantum} Ibinson introduced a class of morphisms from $\Lambda_m^\sigma$ to a smaller cone $\Lambda_n$. However, his description is quite complicated and can be significantly simplified using the clear mathematical language developed in this thesis. The goal of the following short section is to do that, completing the picture of known morphisms of the quantum entropy cone.

\paragraph*{}For a fixed partition $\lambda=(\lambda_1,...,\lambda_n)\vdash m$, define the set partition 
\begin{equation}
\mathfrak{M}_\lambda=\left\{M_i=\left\{\lambda_1+...+\lambda_{i-1}+1,...,\lambda_1+...+\lambda_i\right\}| i\in [n]\right\}
\end{equation}
of $[m]$ after setting $\lambda_0=0$. We say $\lambda$ has \emph{length} $n$, or $|\lambda|=n$. For each subset $J\subset [n]$ we have a ``coarse grained'' subset $K(J)=\bigcup_{i\in I}M_i\subset[m]$ and we can define the linear map
\begin{equation}\label{iBlock}
	b_\lambda: V_m\to V_n, v\mapsto w \text{ with } w_J=v_{K(J)}
\end{equation}
We call this operation \emph{blocking}. Let us find $b_\lambda^\dagger$. Take an arbitrary functional $f\in V_n^*$. Then
\begin{equation}
	f(b_\lambda v)=\sum_{J\in [n]} f_J v_{K(J)}=\sum_{I\in [m]}g_I v_I=g(v),
\end{equation}
where
\begin{equation}
	g_I=\begin{cases}
	    	f_J &J=K(I)\\
	    	0 & \text{else}
	    \end{cases},
\end{equation}
so $b_\lambda^\dagger f=g$.
Now, define the cone of all entropy vectors that can be constructed by application $b_\lambda$ for any $\lambda \vdash m$ for arbitrary $m$ from a symmetric entropy vector in $\Gamma_m^\sigma$,
\begin{equation}
	\Gamma_n^{i\sigma}=\bigcup_{m\ge n}\bigcup_{\substack{\lambda\vdash m\\|\lambda|=n}} b_\lambda \Gamma_m^\sigma.
\end{equation}
We call such entropy vectors \emph{symmetrically sub-dividable}, in Ibinson's thesis they are called weakly symmetric. Applying different $b_\lambda$s to non-symmetric entropy vectors would not yield anything new, as $b_\lambda \Gamma_m\subset \Gamma_n$ for all $\lambda\vdash m,\ |\lambda|=n$ and $b_\mu \Gamma_n=\Gamma_n$ for $\mu=(1^n)$. Obviously $\Gamma_n^{i\sigma}$ is a scale invariant, and $\Gamma_n^{i\sigma}\subset \Gamma_n$, so $\conv(\Gamma_n^{i\sigma})\subset\Gamma_n$. Ibinson was able to completely characterize $\Gamma_4^{i\sigma}$ by first finding candidates for inequalities computationally and then proving them, again using a computer program \cite{ibinson2008quantum}.


\chapter{Local Geometry of Extremal Rays}\label{differential}

Let $K$ be a polyhedron and $f: D\to K$ a differentiable function with $\text{ri} K \subset f(D)\subset K$. Assume further that $D$ is open, and that some low-dimensional face $F\subset f(D)$ is contained in the range of $f$. Then the function $f$ has to be highly singular at preimages of elements $v\in F$, as the range of the differential has to be contained in the span of $F$. The described situation precisely occurs for the von Neumann entropy function, I present this result in detail in the first section of this chapter. The second section is dedicated to specializing the result to the classical entropy cone.

\section{Quantum}

Suppose $\rho\in\mathcal{B}\left(\left(\C^d\right)^{\otimes n}\right)$ were a preimage of a point $h\in R\subset\overline{\Gamma}_n$ of an extremal ray $R$ that lies in a \emph{polyhedral} region of the cone $\overline{\Gamma}_n$. We say a point $p\in C$ lies in a polyhedral region of a cone $C\subset V$, if  we can find an open neighborhood $U\subset V$ of $p$ such that 
\begin{equation}
	C\cap U=U \cap \left\{x\in V|f_1(x)\ge 0,...,f_m(x)=0\right\}
\end{equation}
for some finite set of functionals $\left\{f_1,...,f_m\right\}\subset V^*$. If every nonzero point in an extremal ray lies in a polyhedral region, we call the ray \emph{isolated}. Then the differential of the entropy function, $\D s$, has rank one, as the only possible direction is along the ray, $\im(\D s)_\rho=\R s(\rho)$. A component of the entropy function is the entropy of a reduced state,
\begin{equation}
	s(\rho)_I=S(\rho_I)=H(\spec(\tr_{I^c}(\rho))).
\end{equation}
\paragraph{} The only complicated map in the composition on the right hand side is the one mapping a matrix to its \emph{spectrum}, $\spec(\rho)=(\lambda_1,\lambda_2,...,\lambda_r)$, where the $\lambda_i$ are the ordered eigenvalues of $\rho$, $\lambda_i\ge\lambda_j$ for all $j\ge i$. The differential of spec can be found using time independent perturbation theory, which involves diagonalizing the density matrix where the differential is calculated. That is problematic, as we want, in particular, to calculate the above differential for $I$ and $J$ simultaneously for $I\cap J\not=\emptyset$, but we are, in general, not able to choose a basis where both $\rho_I$ and $\rho_J$ are diagonal. However, because of the structure imposed by the strong condition that the rank of the differential has to be one, we get along by looking at disjoint subsets.
\paragraph{} We first restrict ourselves to pure states, in fact, characterizing entropies in of $n$-partite pure states is equivalent to characterizing the entropies of $(n\!-\!1)$-partite mixed states, as the cones are isomorphic (see Section \ref{sec:mor}). We can of course also restrict our attention to pure states $\ket{\psi}$ such that there is no bipartition $[n]=I\cup I^c$ such that $\ket\psi=\ket{\psi_I}\otimes\ket{\psi_{I^c}}$, as otherwise the problem reduces to the characterization of two lower-dimensional entropy cones. It is made precise below how this reduction works. Let us look at a single entropy, i.e.\ we look at the differential of the function
\begin{equation}
	s_I: \mathbb{P}\mathcal{H}\to \R, \pi\mapsto  H(\spec(\tr_{I^c}\pi)),
\end{equation}
where $\mathbb{P}\mathcal{H}=\left\{\ketbra{\psi}{\psi}\big|\ket{\psi}\in\mathcal{H}, \bracket{\psi}{\psi}=1\right\}$ is the projective space corresponding to the Hilbert space $\mathcal{H}=\mathcal{H}_1\otimes ...\otimes\mathcal{H}_n$. Note that this is viewed as a real manifold. Its tangent space at a point $\ketbra{\psi}{\psi}$ is $T_{\ket\psi}\mathbb{P}\mathcal{H}=\left\{\ketbra{\psi}{\phi}+\ketbra{\phi}{\psi}\big|\langle\psi|\phi\rangle=0\right\}$. Where there is no confusion anticipated we denote an element $\ketbra{\psi}{ \psi}$ of the projective space by $\psi$ and also an element $\ketbra{\psi}{\phi}+\ketbra{\phi}{\psi}$ of the tangent space at $\ket\psi$ by $\phi$. 

\paragraph{} The partial trace is linear and hence equal to its own differential. Let us have a look at the spectral map. If $\spec(\tr_{I^c}\ketbra{\psi}{\psi})$ is non degenerate, the differential is just given by the first order eigenvalue correction from textbook perturbation theory,
\begin{eqnarray}
	\left[(\D\spec_I)_{\psi}(\phi)\right]_\alpha&=&\bra{\alpha}\left(\tr_{I^c}\ketbra{\psi}{\phi}+\ketbra{\phi}{\psi}\right)\ket{\alpha}\nonumber\\&=&2 \Re \bra{\alpha}\left(\tr_{I^c}\ketbra{\psi}{\phi}\right)\ket{\alpha}\nonumber\\
	&=&2\Re\, \tr\left(\ketbra{\alpha}{\alpha}\otimes\mathds{1}_{I^c}\ketbra{\psi}{\phi}\right),\label{d1rdm}
\end{eqnarray}
where $\{\ket{\alpha}\}$ is the eigenbasis of the reduced density matrix, i.e.\ $\left(\tr_{I^c}\ketbra{\psi}{\psi}\right)\ket{\alpha}=p_\alpha\ket{\alpha}$. The differential of the Shannon entropy is easy to calculate,
\begin{equation}
	(\mathrm{\D}H)_{p}e_\alpha=-1-\log p_\alpha,
\end{equation}
where $p=(p_\alpha)_{\alpha=1}^d=\spec\left(\tr_{I^c}\ketbra{\psi}{\psi}\right)$ and $\{e_\alpha\}$ is the standard basis of the space $\R^{d_I}$ of spectra of Hermitian matrices. Note that, as will be shown below, the range of $\D\spec$ contains only differences $e_\alpha-e_\beta$, so the $-1$ in the above formula cancels and we get
\begin{equation}
	(\mathrm{\D}H)_{p}(e_\alpha-e_\beta)=\log\left(\frac{p_\beta}{p_\alpha}\right).
\end{equation}
The differential of the whole entropy is therefore, for non degenerate $\tr_{I^c}\ketbra{\psi}{\psi}$,
\begin{eqnarray}\label{dS}
	(\D s_I)_{\psi}\left(\phi\right)&=&(\D H)_{\spec\ \tr_{I^c}\ketbra{\psi}{\psi}}\circ (\D\spec)_{\tr_{I^c}\ketbra{\psi}{\psi}}\circ \tr_{I^c}\left(\ketbra{\psi}{\phi}+\ketbra{\phi}{\psi}\right)\nonumber\\
	&=&-2\sum_{\alpha}\log p_\alpha \Re \,\tr(\left(\ketbra{\alpha}{\alpha}\otimes\mathds{1}\ketbra{\psi}{\phi}\right)
\end{eqnarray}
Note that at this point we cannot employ a continuity argument to extend our result to the degenerate case, as the states $\ket\alpha$ are no longer well defined then. Let us look at Schmidt decomposition of $\ket{\psi}$ with respect to the bipartition $\mathcal{H}=\mathcal{H}_I\otimes\mathcal{H}_{I^c}$,
\begin{equation}\label{schmidt}
\ket{\psi}=\sum_{\alpha}\sqrt{p_\alpha}\ket{\alpha}\otimes\ket{\psi_\alpha}.
\end{equation}
The kernel of the differential of $\spec\circ\tr_{I^c}$ contains, as we can see from \eqref{d1rdm}, the orthogonal complement of the basis vectors appearing in the Schmidt decomposition of $\ket{\psi}$. The remaining subspace of the tangent space $T_{\ket{\psi}}\mathbb{P}\mathcal{H}\cong\left(\C\ket{\psi}\right)^\perp$ is $\R$-spanned by the vectors
\begin{eqnarray}
\ket{\phi_{\beta\gamma}}&=&\sqrt{p_{\gamma}}\ket{\beta}\otimes\ket{\psi_{\beta}}-\sqrt{p_{\beta}}\ket{\gamma}\otimes\ket{\psi_{\gamma}}\text{ and}\nonumber\\
\ket{\phi'_{\beta\gamma}}&=&i\left(\sqrt{p_{\gamma}}\ket{\beta}\otimes\ket{\psi_{\beta}}-\sqrt{p_{\beta}}\ket{\gamma}\otimes\ket{\psi_{\gamma}}\right).
\end{eqnarray}
On these basis vectors
\begin{eqnarray}
	(\D \spec_I)_{\psi}(\phi_{\beta\gamma})&=&2\sqrt{p_\beta p_\gamma}(e_\beta-e_\gamma)\text{ and}\nonumber\\
	(\D \spec_I)_{\psi}(\phi'_{\beta\gamma})&=&0.
\end{eqnarray}
This shows that the range of $(\D\spec)_{\psi}$ is $\spa\left\{e_\alpha-e_\beta\big|p_\alpha\not=0\not=p_\beta\right\}$.
 The differential of the $i$th particle's von Neumann entropy is hence
\begin{equation}\label{ds}
	(\D s_I)_{\psi}(\phi_{\beta\gamma})=2\sqrt{p_\beta p_\gamma}\log{\frac{p_\gamma}{p_\beta}},
\end{equation}
and, by continuity of the involved functions, this formula is also valid for $p_\beta=p_\gamma$, although we assumed non-degeneracy for the calculation of $\D\spec$. This shows that $s_I$ is \emph{critical}, i.e.\ $(\D s_I)_{\psi}\equiv 0$, if and only if the spectrum of $\tr_{I^c}\ketbra{\psi}{\psi}$ is \emph{flat}, meaning that there is a $k\in\N$ such that
\begin{equation}
	p_\alpha=\begin{cases}
	         	\frac{1}{k}&\alpha\le k\\
	         	0&\text{ else}
	         \end{cases}.
\end{equation}
That implies in particular that if the spectra of all reductions of a state are flat, it can populate isolated extremal rays. In the following we call the empty spectrum $\spec_{\emptyset}$ of a general state as well as the total spectrum $\spec_{[n]}$ of a pure $n$-partite state trivial.
It turns out that only states can populate isolated extremal rays where either all spectra are flat or none but the trivial ones are.
\begin{lem}\label{flatornot}
	Let $\ket{\psi}\in\mathcal{H}=\mathcal{H}_1\otimes ...\otimes\mathcal{H}_n$ with $s(\ket{\psi})_I\not=0$ for all $\emptyset \not=I\subsetneq [n]$ and $\im(\D s)_{\psi}=\R s(\ket{\psi})$. Then either all spectra are flat, or the only flat spectra are $\spec_\emptyset$ and $\spec_{[n]}$.
\end{lem}
\begin{proof}
	Let $\emptyset\not=I\subsetneq[n]$ such that $\spec_I(\ket{\psi})$ is flat. Then for all $\ket\phi\in\mathcal{H}$, $0=(\D s_I)_{\psi}(\phi)=c\, s_I(\ket{\psi})$. But $s_I(\ket{\psi})\not=0$, so $c=0$ and hence $(\D s_J)_{\ket{\psi}}(\ket{\phi})=c\, s_J(\ket{\psi})=0$ for all $\emptyset\not=J\subsetneq [n]$. This means all nonzero entropies are critical and therefore all spectra are flat.
\end{proof}
Let us look at single particle entropies now. All following derivations also work for an arbitrary partition, as the two are connected via a blocking morphism \eqref{iBlock}. As the von Neumann entropy is invariant under local basis change, we can assume without loss of generality that a states single particle reduced density matrices are diagonal in the standard basis. For completeness we state the conditions on the expansion coefficients resulting from that, 
\begin{eqnarray}\label{diagbasis}
&&\ket\psi=\sum_{\alpha\in[d]^n}c_{\alpha}\ket{\alpha}, \text{ with}\nonumber\\
&&\sum_{\alpha_1,...,\alpha_{j-1},\alpha_{j+1},...,\alpha_n}c^*_{\alpha_1...\alpha_{j-1}\beta\alpha_{j+1}...\alpha_n}c_{\alpha_1...\alpha_{j-1}\gamma\alpha_{j+1}...\alpha_n}=\delta_{\beta\gamma}p^{(j)}_{\beta_j}\,\, \forall i\in [n],
\end{eqnarray}
where $\ket{\alpha}:=\ket{\alpha_1...\alpha_n}:=\ket\alpha_1\otimes ... \otimes\ket\alpha_n$. In the Schmidt decomposition \eqref{schmidt} for $I=\{i\}$ a singleton the Schmidt basis for the complement is
\begin{equation}\label{diagschmidt}
	\ket{\psi_{\alpha_i}}=\frac{1}{\sqrt{p^{(i)}_{\alpha_i}}}\sum_{\substack{\gamma\in[d]^n\\ \gamma_i=\alpha_i}}c_\gamma\ket{\gamma_{i^c}}
\end{equation}
Let us have a look which noncritical states can populate isolated extremal rays.
\begin{prop}\label{nonflatissimple}
	Let $\ket{\psi}\in\mathcal{H}=\mathcal{H}_1\otimes ...\otimes\mathcal{H}_n$ with 
	\begin{enumerate}
		\item $s(\ket{\psi})_I\not=0$ for all $\emptyset \not=I\subsetneq [n]$,
		\item $\im(\D s)_{\ket{\psi}}=\R s(\ket{\psi})$ and let
		\item $s_\emptyset$ and $s_{[n]}$ be the only critical entropies, equivalently let $\spec_{\emptyset}$ and $\spec{[n]}$ be the only flat spectra.
	\end{enumerate}
Then $\ket\psi=\sum_{m=1}^kc_m\ket\psi_m$ for some $k>1,\ c_m\in \C$ with $\ket{\psi_m}=\bigotimes_j\ket{\psi_m^j}$ and $\bracket{\psi_m^j}{\psi_n^j}=\delta_{mn}$ for all $j\in[n]$, $m,n\in [k]$, i.e.\ $\ket\psi$ is a superposition of product states with each reduction of two of the latter being orthogonal.
\end{prop}
Th state $\ket{\psi}$ is hence a generalized GHZ-state \cite{greenberger1989going}, i.e.\ a state that  is entangled in a way that any reduced state is classically correlated.

\begin{proof}
	Take any state $\ket\psi$ with the properties 1-3. At least two coefficients in the expansion \eqref{diagbasis} are nonzero, as otherwise all entropies were zero, which contradicts 1. Let therefore two coefficients, $c_\alpha:=c_{\alpha_1...\alpha_n}$ and $c_\beta :=c_{\beta_1...\beta_n}$, be nonzero and let $\alpha_1\not=\beta_1$ and $p_{\alpha_1}\not=p_{\beta_1}$. We can assume the latter without loss of generality because of 3. Assume now that there is a $1\not=j\in[n]$ such that $\alpha_j=\beta_j$. Define 
	\begin{equation}
	\ket{a_{\alpha\beta}}=c_\beta^* \ket\alpha-c_\alpha^* \ket\beta=\ket{\alpha_j}\otimes\left(c_\beta^* \ket{\alpha_{j^c}}-c_\alpha^* \ket{\beta_{j^c}}\right).
	\end{equation}
	Then plugging this vector into \eqref{dS} yields 
	\begin{equation}\label{dsj0}
	(\D s_j)_{\ket\psi}(\eta\ket{a_{\alpha\beta}})=0\text{ for any }\eta\in\C^*.
	\end{equation}
	On the other hand, with Equation \eqref{diagschmidt} we get
	\begin{equation}
		\bracket{\phi_{\alpha_1\beta_1}}{a_{\alpha\beta}}=c_\alpha c_\beta\left(\sqrt{\frac{p_\alpha}{p_\beta}}-\sqrt{\frac{p_\beta}{p_\alpha}}\right)
	\end{equation}
	and with \eqref{ds} this implies 
	\begin{equation}
	(\D s_1)_{\ket\psi}\left(\eta\ket{a_{\alpha\beta}}\right)=2\Re\,\eta c_\alpha c_\beta\left(p_\alpha-p_\beta\right)\log{\frac{p_{\beta_1}}{p_{\alpha_1}}},
	\end{equation}
	which is nonzero for $\eta\not\in i\R c_\alpha c_\beta$. Together with Equation \eqref{dsj0}, $s_i(\ket\psi)\not=0$ (1.) and 2. this yields a contradiction. Hence the theorem is proven as therefore $\alpha_i\not=\beta_i$ for all $i\in[n]$ for any two nonzero coefficients $c_\alpha$ and $c_\beta$ of the expansion \eqref{diagbasis} of $\ket\psi$.
\end{proof}
As a corollary we can give a characterization of the states populating isolated extremal rays of $\overline{\Gamma}_n$. But first we still have to make precise how the case of $\ket\psi$ being a product state reduces to a lower dimensional case. Take any pure state $\ket\psi=\ket{\psi_I}\otimes\ket{\psi_{I^c}}\in\mathcal{H}=\mathcal{H}_1\otimes ...\otimes\mathcal{H}_n$ that is a product state with respect to the bipartition $\mathcal{H}=\mathcal{H}_I\otimes\mathcal{H}_{I^c}$. Without loss of generality let $I=[m]$ for some $m<n$. Define the injection
\begin{eqnarray}\label{dirsuminj}
	M: V_m\oplus V_{n-m}&\to& V_n\nonumber\\
	(v^1,v^2)&\mapsto& v \text{ with } v_I=v^1_{I\cap[m]}+v^2_{(I-m)\cap[n-m]}.
\end{eqnarray}
Then we have $s(\ket\psi)=M(s(\ket{\psi_1}),s(\ket{\psi_2}))$. Hence it is a necessary condition for $\ket\psi$ to generate an isolated extremal ray of $\overline\Gamma_n^p$ that $\ket{\psi_1}$ and $\ket{\psi_2}$ generate extremal rays of $\overline\Gamma_m^p$ and $\overline\Gamma_{n-m}^p$ respectively.
\begin{thm}\label{poprays}
	Let $s(\rho)=v\in\Gamma_n$ generate an isolated extremal ray of $\overline{\Gamma}_n$. Then one of the following is true:
	\begin{enumerate}
	\item $v$ results from lower-dimensional extremal rays according to \eqref{dirsuminj},
	\item $v=r\sum_{\emptyset\not= I\subset[n]}e^{(I)}$, or
	\item all non-trivial spectra of $\rho$ are flat.
	\end{enumerate}
\end{thm}
\begin{proof}
	If any non-trivial entropy of $\rho$ is zero, e.g.\ $s_I(\rho)=0$, $\rho=\ketbra{\phi}{\phi}\otimes \rho_{I^c}$ and 1. is true. Assume now that all non-trivial entropies are nonzero. As pointed out above and in Section \ref{sec:mor}, $\overline{\Gamma}_n\cong\overline{\Gamma}_{n+1}^p$. Let $\ket{\psi}$ be a purification of $\rho$. Then $s(\ket\psi)=\mathrm{pur}_n^{n+1}s(\rho)$ generates an isolated extremal ray in $\overline{\Gamma}_{n+1}^p$. According to Lemma \ref{flatornot} either all non-trivial spectra of $\ket\psi$ are flat, in which case we are done because 3. is true, or none of them is. If no non-trivial spectrum is flat, Theorem \ref{nonflatissimple} shows that all reductions of $\ket\psi$ have the same spectrum, hence 2. is true.
\end{proof}
The only isolated extremal ray that is populated by states with non-flat spectra is so simple that it is not hard to explicitly construct a flat representative.
\begin{cor}\label{existflat}
	For each populated isolated extremal ray $R\subset \overline{\Gamma}_n$ there is a state $\rho$ with flat spectra such that $s(\rho)\in R$.
\end{cor}
\begin{proof}
	Unless $R=R_e:=\R_{\ge 0}\sum_{\emptyset\not= I\subset[n]}e^{(I)}$ is the exceptional ray the statement follows from Corollary \ref{poprays}. For the single exceptional ray $R_e$ any $n$-partite reduction of the $(n\!+\!1)$-qubit pure state $\ket\psi=\frac{1}{\sqrt{2}}\left(\ket{0...0}+\ket{1...1}\right)$ does the trick.
\end{proof}
The above theorems can be slightly strengthened: Even if an extremal ray $R\subset\overline{\Gamma}_n$ is not isolated, all the above statements are true for it provided that the extremal ray is still ``edgy''. To make this notion precise, we have to find a way of defining the directional derivative of the \emph{boundary} of the cone. To achieve this, we write the boundary $\partial \overline\Gamma_n$ locally as a graph of a function in the spirit of the well known criterion that a subset of a real Vector space is a differentiable submanifold if and only if locally it can be written as a graph of a differentiable function. Let us show that this is possible.
\begin{lem}\label{boundaryisgraph}
	Let $V$ be a vector space, $K\subset V$ a closed convex cone with $\dim K=\dim V$. Then, for all $p\in\partial K$ there exist $\epsilon>0$ and an affine hyperplane $H\subset V_n$ with orthogonal projector $\pi_H$ such that $\pi_H\big|_{\partial K\cap B_\epsilon(p)}$ is invertible on its image.
\end{lem}
\begin{proof}
	Take an inner point $q\in \text{int}(K)$. Then there exists an $\epsilon>0$ such that $U=B_\epsilon(q)\subset K$. For an arbitrary $p\in\partial K$, define the affine hyperplane $H=\left\{v\in V\big|(p-q,v-q)=0\right\}$. Then, with $\pi:=\pi_H$ being the orthogonal projector onto $H$, $\pi\big|_{\partial K\cap B_\epsilon(p)}$ is invertible on $B_\epsilon(q)\cap H$. To see this, first observe that $\pi\left(B_\epsilon(p)\right)=B_\epsilon(\pi(p))\cap H=B_\epsilon(q)\cap H$. Now Assume it was not the case, i.e.\ we find $x,y\in\partial K$ with $a:=\pi(x)=\pi(y)\in U$. This means that $x,y$ and $a$ lie on a line, hence without loss of generality $y=a+\alpha (x-a)=\alpha x+(1-\alpha)a$, $\alpha\in ]0,1[$. But $a$ is an interior point of $U$, so $y$ is an interior point of $\alpha x+(1-\alpha)U\subset\conv(x,U)\subset K$ and hence $y\not\in\partial K$, which is a contradiction.
\end{proof}
Now let $R\subset\overline\Gamma_n$ be an extremal ray and $0\not =p\in R$, where, for now, we look at $\Gamma_n$ as a subset of $V_n'=\spa\left(e_I\big|\emptyset\not=I\subset [n]\right)$. As an element of an extremal ray, $p\in \partial \overline\Gamma_n$ lies on the boundary of $\overline\Gamma_n$, so according to Lemma \ref{boundaryisgraph} we can find a hyperplane $H$ and an $\epsilon>0$ such that $\pi_H \big|_{B_\epsilon(p)\cap\partial K}$ is invertible on its image $B_\epsilon(\pi_H(p))\cap H$. Let $g:B_\epsilon(\pi_H(p))\cap H\to B_\epsilon(p)\cap\partial K$ be the inverse just shown to exist. Define the function
\begin{eqnarray}
	f: B_\epsilon(\pi_H(p))\cap H&\to& \R\\
	v\mapsto \left\|v-g(v)\right\|.
\end{eqnarray}
Its graph is equal to $B_\epsilon(p)\cap\partial K$ under the isometry
\begin{eqnarray}
	\phi: H\times \R&\to& V_n'\nonumber\\
	(v,x)&\mapsto&v+x n_H,
\end{eqnarray}
where $n_H$ is the unit normal vector of $H$. Now we can define $R$ to be an \emph{edge}, if for all $p\in R\setminus\{0\}$ the directional derivative of the function $f$ constructed above exists only in one direction.
As an example consider a three-dimensional cone generated by a base of the following shape:

\begin{figure}[ht]
	\centering
	\includegraphics[width=3cm]{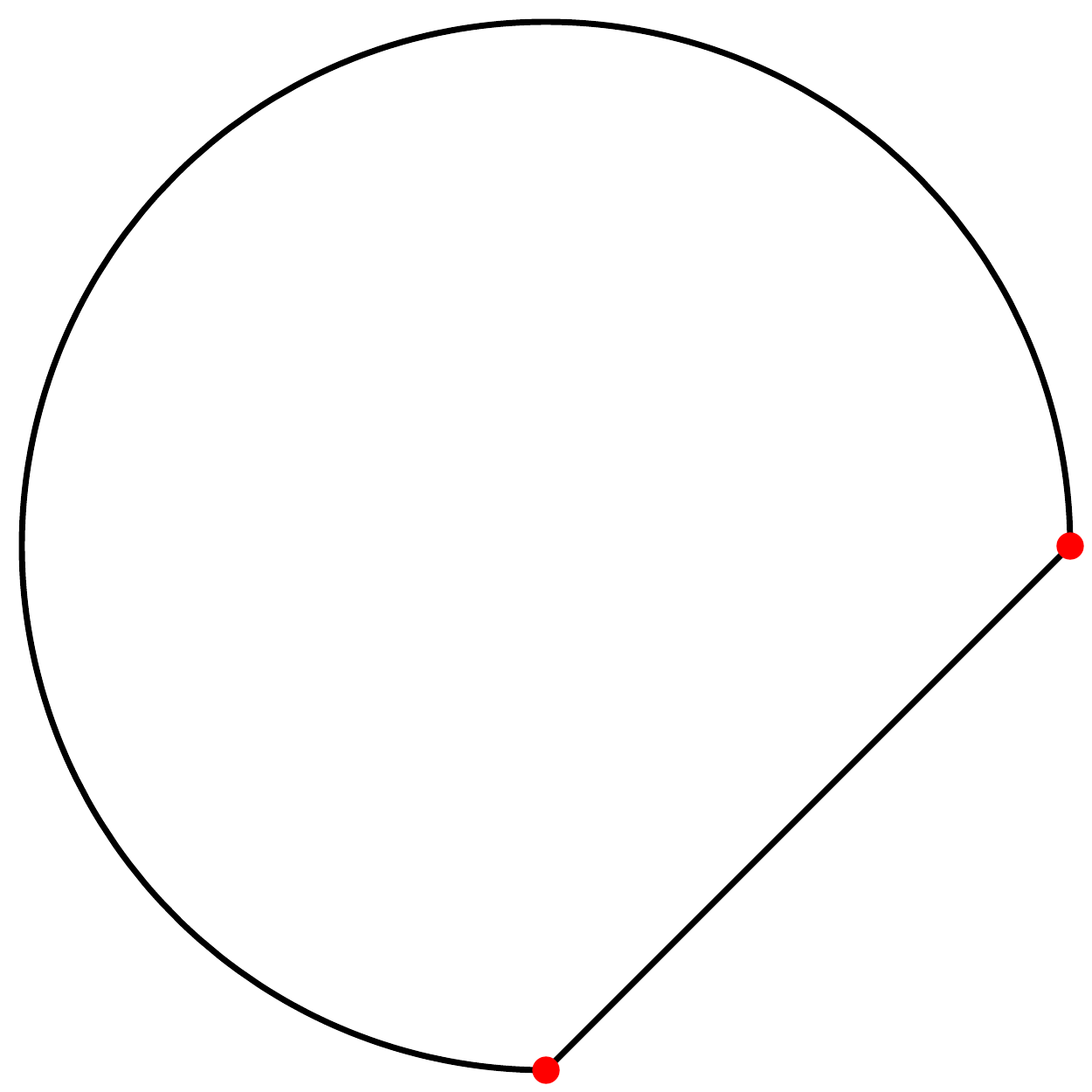}
\end{figure}
The extremal rays generated by the red points are not isolated, but the differential of a function whose image is contained in the cone still has to be of rank one at a preimage of a point on the ray.

We can also extend our characterization to any ray in the cone generated by the extremal rays with flat representatives. 
\begin{thm}
 For each ray $R$ in the subcone of $\overline\Gamma_n$ that is generated by the populated extremal rays there is a sequence of quantum states $\rho_k$ such that
 \begin{equation}
  \lim_{k\to \infty}\frac{1}{k}s(\rho_k)\in R
 \end{equation}

\end{thm}
\begin{proof}

Let $R_i\subset \overline\Gamma_n, i=I$ be the populated edges of $\overline\Gamma_n$ indexed by some set $I$, and $\rho^{(i)}$ some states with flat spectra populating them, i.e.\ $s\left(\rho^{(i)}\right)\in R_i$ for $i\in I$. Let further $h=\sum_{i\in I}a_i s\left(\rho^{(i)}\right),\ a_i\in\R_{\ge 0}$ be an arbitrary vector in $\cone\left(v^{(i)}\big| i\in I\right)$. Define the sequence of states
\begin{equation}
	\sigma_n=\bigotimes_{i\in I}{\rho^{(i)}}^{\otimes \lfloor n a_i\rfloor}.
\end{equation}
Then
\begin{equation}
	\frac{1}{n}s\left(\sigma_n\right)=\frac{1}{n}\sum_{i\in I}\left\lfloor na_i\right\rfloor s\left(\rho^{(i)}\right)\xrightarrow{\ n\to\infty\ }h,
\end{equation}
and $\sigma_n$ has flat spectra for all $n$, as the spectrum of a tensor product state is the product of the spectra of the factor states.
\end{proof}
Note that a countably infinite number of populated edges is not a problem here, as then (assume $I=\N$ in this case) $a_i\to 0(i\to\infty)$ and therefore for each $n\in\N$ there exists an index $i_n\in\N$ such that $n a_{i}<1$ for all $i\ge i_n$. Hence $\sigma_n$ lives on a finite tensor product space ($\rho^{\otimes 0}=1$). This implies that the set of rays with flat representatives is dense in the set of rays of $\cone\left(v^{(i)}\big| i\in I\right)$ with respect to the topology induced on the set of rays by any base.

\paragraph{}For the quantum R\`enyi entropies
\begin{equation}
	S_\alpha(\rho)=\frac{1}{1-\alpha}\log\tr \rho^\alpha,\ \
\end{equation}
the question about realizable extremal rays for $\alpha\in(0,1)$ was answered by Linden, Mosonyi and Winter \cite{linden2013structure}. Their result shows that the closed conic hull of the R\`enyi-entropic region is the positive orthant and that none of its extremal rays are realizable. For $\alpha\in(1,\infty)$ they have shown that the closure of the entropy region is not a convex cone. The only R\`enyi entropies that form a cone and are still to be characterized are the von Neumann Entropy $\alpha=1$ and the R\`enyi-0-entropy. Some results for the latter were achieved in \cite{cadney2013inequalities}.


\section{Classical}\label{classdiff}

In the following section, I show that the technique used for the quantum entropy cone can also be adapted to the classical entropy cone. Although the proof of Theorem 4.1 in \cite{chan2002relation} proves that for each point in $\Sigma_n$ there is a sequence of flat states whose entropies converge to the corresponding ray, the following result improves on that characterization by giving a condition on states that exactly realize a point on an extremal ray.

In the classical scenario, the geometry is slightly different and applying the technique from the last section requires some care. Let $\mathcal{X}$ be a finite alphabet of size $d$, $X=(X_i)_{i\in [n]}$ a random variable with values in $\mathcal{X}^n$ $p: \mathcal{X}^n\to [0,1]$ its distribution and $p_I, I\subset [n]$ its marginals. The important space is the set of probability distributions on $\mathcal{X}$, which is not a smooth manifold like the projective space $\mathbb{P}\mathcal{H}$ in the quantum case, but a convex polytope, more precisely the simplex
\begin{equation}
\mathcal{P}^{d^n}=\left\{p\in\R^{d^n}\Big|||p||_1=1, p_i\ge0\, \forall i\in [m]\right\}
\end{equation}
This is not a nice region to define a differentiable function on, as it is, first of all, not open.

This indicates that we have to be more careful here compared to the quantum case. One could think of defining ``one-sided derivatives'', but this would defeat the object, as in the reasoning in the above section we derived the restriction of the range of the differential, roughly speaking, from the fact that if we could infinitesimally walk out of some face into the cone, we could also walk out of the cone. Let us therefore make a definition that is similar to the tangent space of a manifold:
\begin{defn}
	Let $P$ be a polytope. For any point $x\in P$ define the \emph{supporting face} $F(p)$ as the unique face such that $p\in \mathrm{ri}(F(p))$. (We use the convention $\mathrm{ri}(\{x\})=\{x\}$.) Furthermore we define the \emph{supporting space} $T(p)$ to be the Vector space part of the affine hull of the supporting face.
\end{defn}

Now we have the means to unambiguously define a differential of the entropy function $h$
\begin{equation}
	(\D h)_p: T(p)\to \R^{2^n}
\end{equation}
The distributions $p\in\mathcal P^{d^n}$ with $\dim T(p)=0$ are precisely the deterministic points, and  for them we have $h(p)=0$. Let us therefore assume that $\dim T(p)\ge 1$.
\begin{prop}\label{callflat}
Let $p\in\mathcal{P}^{d^n}$ with $h(p)\in R\subset \overline\Sigma_n$ for some edge $R$ of $\overline{\Sigma}_n$. Then either all marginal distributions $p_I, \emptyset\not= I\subset[n]$ are flat, or none of them is.
\end{prop}
\begin{proof}
	Suppose some $p_I$ is flat and $H(p_I)\not=0$. Then we can find $x,y\in\mathcal{X}^n$ such that $x_I\not= y_I$, $p_I(x_I)=p_I(y_I)\not=0$ and $v=e_x-e_y\in T(p)$. We find
	\begin{equation}
		(\D h_I)_p(v)=\log\left(\frac{p_I(y_I)}{p_I(x_I)}\right)=0.
	\end{equation}
	But on the other hand, as $p$ populates an edge, $(\D h)_p(v)\in \R h(p)\, \forall v\in A(p)$ and therefore, as $H(p_I)\not=0$, $(\D h)_p=0$, which implies that all non-trivial marginals are flat.

\end{proof}

Now look at the case where no marginal distribution is flat.
\begin{cor}
	Suppose $p\in \mathcal{P}^{d^n}$ has no flat marginals except for the trivial $p_\emptyset$ and $h(p)$ lies in an edge of $\overline \Sigma_n$. Then $p_x\not=0\not=p_y$ implies $x_i\not=y_i\, \forall i\in[n]$, i.e.\ $p$ is the distribution of a collection of random variables $X=(X_i)_{i\in [n]}$ with $X_i=X_j,\ i,j\in[n]$.
\end{cor}
\begin{proof}
	If $p_x\not=0\not=p_y$, $v=e_x-e_y\in T(p)$. Suppose now there existed an $i\in[n]$ such that $x_i=y_i$. Then
	\begin{equation}
		(\D h_I)_p(v)=\log\left(\frac{p_I(y_I)}{p_I(x_I)}\right)=0.
	\end{equation}
	But as in Lemma \ref{callflat} this would imply $(\D h)_p= 0$ which contradicts the assumption that $p$ has no flat marginals
\end{proof}
For such a probability distribution we have $h(p)=H(p_1)\sum_{\emptyset\not=I\subset [n]}e^{(I)}$ as an entropy vector. Summarizing the results, we get
\begin{thm}\label{classpoprays}
 Let $h(p)=v\in\Sigma_n$ generate an isolated extremal ray of $\overline{\Sigma}_n$. Then one of the following is true:
	\begin{enumerate}
	\item $v$ results from lower-dimensional extremal rays according to \eqref{dirsuminj},
	\item $v=r\sum_{\emptyset\not= I\subset[n]}e^{(I)}$, or
	\item $p$ and all its marginals are flat.
	\end{enumerate}

\end{thm}
That is the same exceptional ray we had in the quantum case, and the flat representative described in Corollary \ref{existflat} is classical. Hence for the edges of the classical entropy cone we have the same
\begin{cor}
	Let $R\subset\overline\Sigma_n$ be a populated edge. Then there exists a probability distribution $p\in\mathcal{P}^{d^n}$ for some $d\in\N$ such that $h(p)\in R$ and $p_I$ is flat for all $I\in[n]$.
\end{cor}
As in the quantum case, this also implies that every ray in the cone spanned by the populated edges can be approximated by random variables with flat marginals. This provides no new insight, as it is known that, in fact, \emph{any} ray $R\subset \overline\Sigma_n$ can be approximated in that way \cite{chan2002relation}.


\chapter{Stabilizer states}\label{stabs}

\section{The Stabilizer Entropy Cone}

As the characterization of the whole quantum entropy cone for $n\ge4$ parties proved elusive so far, and even the classical entropy cone is far from characterized unless $n\le 3$, it seems a good idea to start by finding inner approximations like the weakly symmetric entropy cone described in the last section. This can also be done by looking at a subset of states that has additional structure such as to allow for a direct algebraic characterization of the possible entropy vectors.

One subset that allows for this to be done is the set of \emph{stabilizer states} \cite{linden2013quantum, gross2013stabilizer}. In the following chapter I introduce this set, review the results concerning stabilizer states that were obtained in \cite{gross2013stabilizer} and \cite{linden2013quantum} independently, and give a further improvement of the result under a certain condition on the local Hilbert space dimension that includes the important qubit case.

\paragraph{}Stabilizer states are quantum states in finite dimensional systems that are invariant under a certain group of operations, the \emph{stabilizer group}. This was also the original defining property, as they are the normalized projectors onto certain subspaces called \emph{stabilizer codes} that are invariant under the stabilizer group and thus robust against noise if the noise operators are from that group. However, they can also be constructed from a different point of view using \emph{finite phase spaces}. This approach was taken in \cite{gross2013stabilizer}, which I also follow in this introduction.

\paragraph{}Let $d\in \N$, and $V$ be the \emph{free module} of rank $2n$ over the ring of integers modulo $d$, that is $V=\left(\Z/d\Z\right)^{2n}=:\Z_d^{2n}$. $V$ is the finite analogue of the \emph{phase space} known from classical mechanics, in this case for an $n$-partite system with a local Hilbert space $\C^d$. The standard \emph{symplectic form} on $V$ is
\begin{equation}
	\omega(v,w)=\sum_{i=1}^n\left( v_{p_i}w_{q_i}-v_{q_i}w_{p_i}\right)
\end{equation}
where $v=\left(v_{p_1},...,v_{p_n},v_{q_1},...,v_{q_n}\right)$. If $d$ is prime, $V$ is a vector space and $\omega$ is the more broadly known standard symplectic form on vector spaces.

A submodule $M\subset V$ is called \emph{isotropic}, if the symplectic form vanishes on it, i.e.\ $\omega(M,M)=\{0\}$. We are interested in looking at subsets of the $n$ systems as well, so for each $I\subset [n]$ we define the projection $\pi_I$ onto the submodule $V_I=\left\{v\in V|v_i=0\, \forall i\not\in[n]\right\}$, where $v_i=\left((v_{p_i},v_{q_i}\right)$ and the restriction $A_I=A\cap V_I$ for subsets $A\subset V$. Note that it is not clear whether a projection onto an arbitrary submodule exists, but the special submodules $V_I$ are direct summands, $V=V_I\oplus V_{I^c}$, in in which case the existence is obvious.

\paragraph{}The additive group of the ring $\Z_d$ is, of course, Abelian and its characters are the powers of  $\chi_d(x)=e^{\frac{2\pi i}{d}x}$. Returning to the phase space analogy, half the direct summands (``dimensions'') of $V$ constitute the \emph{momentum space}, and the other half are the \emph{configuration space} $\Z_d^n$, which is also reflected in the notation above.  Wave functions are accordingly square summable complex functions on configuration space, i.e.\ $\mathcal{H}=L_2(\C,Z_d^n)\cong \left(\C^d\right)^{\otimes n}$. Let us define the \emph{Weyl operators} corresponding to the symplectic structure given by $\omega$, first for $n=1$ and $(P,Q)\in \Z^2$,
\begin{equation}
	(W(P,Q)\psi)(x)=\tau_{2d}(-PQ)\chi_d(Px)\psi(x-Q)
\end{equation}
where $\tau_{2d}(y)=\chi_{2d}((d^2+1)y)$. A short calculation shows, that
\begin{eqnarray}
	W(P,Q)W(P',Q')&=&\tau_{2d}(PQ'-QP')W(P+P',Q+Q')\nonumber\\
	W(P,Q)^{-1}&=&W(P,Q)^\dagger=W(-P,-Q)\\
	W(P,Q)W(P',Q')&=&\chi_{d}(PQ'-QP')W(P',Q')W(P,Q)\nonumber
\end{eqnarray}
Now we want to define the Weyl operators for $p,q\in\Z_d$. If $d$ is odd, we can define $w(p,q)=W(P,Q)$ with $p=P\mod d$ and $q=Q\mod d$ straightforward, as $W(P+d,Q)=W(P,Q+d)=W(P,Q)$ for all $(P,Q)\in \Z^2$. If $d$ is even, we define $w(p,q)=W(P,Q)$ as for odd $d$ but with the additional condition that $P,Q\in\{0,1,...,d-1\}$. This fixes the sign of the operators and is necessary, because $W(P+d,Q)=W(P,Q+d)=-W(P,Q)$ in that case. In both cases $(p,q)\mapsto w(p,q)$ defines at least a projective representation of the additive group of $\Z_d^2$. For $n\ge 2$ we have $V=\left(\Z_d^2\right)^{\otimes n}$ and define $w(v)=\bigotimes_{i=1}^nw(v_i)$ with $v_i=(v_{p_i},v_{q_i})$. 

\paragraph{}We are now ready to define stabilizer states and a few related notions.
\begin{defn}
	Given a phase space $V=\Z_d^{2n}$, a \emph{stabilizer group} is defined to be a group $G$ of multiples of Weyl operators such that the only multiple of $\mathds{1}=w(0)$ contained in $G$ is $\mathds{1}$ itself. The subspace 
	\begin{equation*}
	 \mathcal{H}^G=\left\{\psi\in \left(\C^d\right)^{\otimes n}\Big| g\psi=\psi \, \forall g\in G\right\}
	\end{equation*}
	is called the \emph{stabilizer code} associated with $G$. Finally, the normalized projector 
	\begin{equation*}
	 \rho_G=\frac{1}{d^n}\sum_{g\in G}g
	\end{equation*}
	is called the corresponding \emph{stabilizer state}. Two stabilizer states $\rho_1$ and $\rho_2$ are said to be \emph{equivalent}, if they differ only by conjugation by Weyl operators, i.e.\ there exists a $v\in V$ such that $\rho_1=w(v)\rho_2 w(-v)$.
\end{defn}
Some authors have a definition of stabilizer states narrower than the above one by additionally demanding purity.
To connect stabilizer states to isotropic submodules of the phase space we need a few auxiliary results.
\begin{lem}[\cite{gross2013stabilizer}, Lemma 6]\label{pontry}
	Let $V=\Z_d^{2n}$ be a phase space. Then the set of characters of the additive group of $V$ is $\hat{V}=\{\chi_d(\omega(v,\cdot))|v\in V\}$.
\end{lem}
\begin{proof}
	Of course each elements of $\hat{V}$ is a character, and as $\chi_d$ is injective on $\Z_d$ and $\omega$ is non degenerate, we have $|\hat{V}|=|V|$, and thus we have found all characters of $V$.
\end{proof}
If $V$ is a symplectic vector space, we have the well known formula $\dim U+\dim U^\omega=\dim V$ for all subspaces $U\subset V$, where $U^\omega$ denotes the \emph{symplectic complement} of $U$. For a Module, the dimension may not be well defined, nevertheless a similar statement is still true.
\begin{lem}[\cite{gross2013stabilizer}, Lemma 7]
	For a finite symplectic module $V$ and a submodule $M\subset V$, $|M||M^\omega|=|V|$
\end{lem}
\begin{proof}
	Consider the group homomorphism
	\begin{equation}
		\Phi: M^\omega\to \widehat{V/M}, x\mapsto \left([w]\mapsto \chi_d(\omega(x,w))\right),
	\end{equation}
	where $[w]$ is the coset class of $w$. It is injective because $\omega$ is non degenerate. Let now $\tau\in\widehat{V/M}$ be a character of $V/M$, then $(v\mapsto \tau([v])\in \hat{V}$, so according to Lemma \ref{pontry} there exists a $w\in V$ such that $\tau([v])=\chi_d(\omega(w,v))$. As $\tau$ vanishes on $M$, $w\in M^\omega$. So $\Phi$ is an isomorphism, and
	\begin{equation}
		|M^\omega|=|\widehat{V/M}|=|V/M|=\frac{|V|}{|M|}
	\end{equation}

\end{proof}
In particular, as $M\subset (M^\omega)^\omega$, $(M^\omega)^\omega=M$.

\paragraph{}Stabilizer states are characterized by isotropic submodules of the finite phase space defined above:
\begin{thm}[\cite{gross2013stabilizer}, Theorem 1]\label{stabchar}
	Let $d>1$ be an integer and $V=\Z_d^{2n}$ the phase space of $n$ $d$-dimensional quantum systems. Then there is a one-to-one correspondence between isotropic submodules $M\subset V$ and equivalence classes $[\rho(M)]$ of stabilizer states on $\mathcal{H}=(\C^d)^{\otimes n}$. The partial trace of the state corresponds to restriction to the phase space of the chosen systems, i.e.
	\begin{equation}
		\left[\rho(M)_I\right]=\left[\rho(M_I)\right],
	\end{equation}
	and the entropy vectors of stabilizer states are
	\begin{equation}
		S([\rho(M)_I])=|I|-\log|M_I|
	\end{equation}
	furthermore, for any representative $ \rho$ of the equivalence class there are phases $\mu_m$ such that
	\begin{equation}
		\rho=\sum_{m\in M}\mu_m w(m)
	\end{equation}

\end{thm}
In the following we look at the entropy cone generated by stabilizer states,
\begin{equation}
	\Gamma_n^{\text{stab}}=\cone\left\{s(\rho)|\rho \text{n-partite stabilizer state}\right\},
\end{equation}
which is an inner approximation of the full entropy cone $\overline\Gamma_n$.

Gross and Walter prove in \cite{gross2013stabilizer} that all balanced classical information inequalities hold for stabilizer states, that is $\left(\Gamma_n^{\text{stab}}\right)^*\cap B_n\supset \Sigma_n^*\cap B_n$. To this end, they explicitly construct a classical model that reproduces the entropy of a given stabilizer state up to a term proportional to the size of the subsystem:
\begin{thm}[\cite{gross2013stabilizer}, Theorem 2]\label{classmodel}
	Let $d>1$ be an integer and $V=\Z_d^{2n}$ the phase space of $n$ $d$-dimensional quantum systems. Let $\rho$ be a stabilizer state corresponding to a submodule $M\subset V$. Then the random variable $X=(X_1,...,X_n)\sim \mathbf{U}(M^\omega)$  that takes values uniformly on the symplectic complement of $M$ has the entropy vector
	\begin{equation}\label{stabchan}
		H(X_I)=S(\rho_I)+|I|=\log\frac{|M^\omega|}{|\ker(\pi_I)\cap M^\omega|},
	\end{equation}
	where the last expression gives the subgroup model corresponding to $X$ shown to exist in \cite{chan2002relation}.
\end{thm}

\section{Stabilizer Entropies and Linear Rank Inequalities}

It turns out that stabilizer entropies can be related to subspace ranks. The problem of characterizing subspace rank functions introduced in the following paragraph is well studied in a branch of mathematics called \emph{matroid theory}. Let us first define the notion of a rank function, which plays a role similar to the entropy function:
\begin{defn}
 Let $\mathbb F$ be a finite field, $V$ a vector space over $\mathbb{F}$ and $U=(U_i)_{i\in[n]}, U_i\subset V$ a collection of subspaces. Then we define the \emph{rank function} by
 \begin{eqnarray}
  r(U)&=&(\dim U_I)_{I\subset[n]},
 \end{eqnarray}
 where
 \begin{equation}
  U_I=\spa \cup_{i\in I}U_i
 \end{equation}
and we adopt the convention $\dim \emptyset =0$. We call $r(U)$ a \emph{rank vector}.
\end{defn}
 Note that the restriction to finite fields is not necessary, but the finite field case is the only one this thesis is concerned with.
The set of all rank vectors is, of course, not a convex cone, as it contains only integral points. Nevertheless it makes sense to define the conic hull of all possible rank functions,
\begin{equation}
 \Lambda_n=\cone\left\{r(U)\Big|U=(U_i)_{i\in [n]} \text{ a collection of subspaces}\right\}.
\end{equation}
\paragraph{} It is easy to see that for any collection of subspaces of a finite vector space there is a collection of random variables such that the rank function of the former coincides with the entropy vector of the latter.

\begin{prop}[\cite{hammer2000inequalities}, Theorem 2]
 Given a collection $U=(U_i)_{i\in[n]}$ of subspaces of a vector space $V$ over a finite field $\mathbb F$, there is a random variable $X=(X_i)_{i\in [n]}$ such that
 \begin{equation}
  r(U)=\alpha h(X).
 \end{equation}

\end{prop}
\begin{proof}
 Let $X$ be a random variable uniformly distributed on $V^*$, the dual space of $V$. Then we define
 \begin{equation}
  X_i=X|_{U_i},
 \end{equation}
 i.e.\ the restriction of the random functional $X$ to $U_i$. This yields
 \begin{equation}
  X_I=X|_{U_I}
 \end{equation}
and therefore
\begin{equation}
 h(X)=\log(|F|)r(U),
\end{equation}
as $X_I$ is uniformly distributed on $U_I$.
\end{proof}

This shows that the cone generated by all rank function is contained in the classical entropy cone, i.e.
\begin{equation}
 \Lambda_n\subset \overline \Sigma_n.
\end{equation}
However, there are inequalities respected by rank functions that are violated by entropies \cite{dougherty2009linear,chan2011truncation}, making the inclusion strict,

\begin{equation}
 \Lambda_n\subsetneq \overline \Sigma_n.
\end{equation}
These inequalities are called \emph{linear rank inequalities}.

\paragraph{} The cone $\Lambda_n$ plays an important role in network coding as it determines the capacities achievable by \emph{linear codes}. This class of codes has many advantages over general network codes, as encoding, decoding and construction can be done efficiently \cite{yeung2008information}. The existence of linear rank inequalities violated by Shannon entropies and explicit counterexamples \cite{dougherty2005insufficiency} show, however, that linear codes cannot achieve the maximal possible rates.

\paragraph{} Let us connect the dots by investigating the relationship between stabilizer entropies and subspace ranks. The following Lemma was stated without proof in \cite{chan2011truncation} with a forward reference to a publication that was not published afterwards, neither was I able to get information about it's status upon request \cite{Chandidntanswer}. In the following $V_1+V_2=\left\{v_1+v_2|v_1\in V_1, v_2\in V_2\right\}=\spa(V_1, V_2)$ denotes the Minkowski sum of two subspaces, and we also use the sum symbol $\sum$ for this concept.
\begin{lem}\label{vectolinear}
	Given a finite field $\mathbb{F}$, a vector space $V$ over $\mathbb{F}$ and a vector $h=\left(h_I\right)_{I\subset[n]}\in V_n$, the following two statements are equivalent:
	\begin{itemize}
		\item[(i)] There exists a collection $(V_i)_{i\in[n]}$ of subspaces of $V^*$ such that $h_I=\dim\sum_{i\in I}V_i$
		\item[(ii)] There exists a collection $(W_i)_{i\in[n]}$ of subspaces of $V$ such that $h_I=\log_{|\mathbb{F}|}\left(\frac{|V|}{|\bigcap_{i\in I}W_i|}\right)$.
	\end{itemize}

\end{lem}
\begin{proof}
	Let $U^o=\left\{f\in V^*|f(v)=0\,\, \forall v\in U\right\}\subset V^*$ denote the \emph{annihilator} of a subspace $U\subset V$. Then $\dim U+\dim U^o=\dim V$ and for any subspaces $U_1,U_2\subset V$ we have $\left(U_1\cap U_2\right)^o=U_1^o+U_2^o$. Furthermore, the cardinality of any subspace $U\subset V$ is $|U|=|\mathbb{F}|^{\dim U}$, therefore
	\begin{eqnarray}
		\log_{|\mathbb{F}|}\left(\frac{|V|}{|\bigcap_{i\in I}W_i|}\right)&=&\log_{|\mathbb{F}|}|V|-\log_{|\mathbb{F}|}|\bigcap_{i\in I}W_i|\nonumber\\
		&=&\dim V-\dim \bigcap_{i\in I}W_i=\dim\left(\bigcap_{i\in I}W_i\right)^o=\dim\sum_{i\in I}W_i^o.
	\end{eqnarray}
\end{proof}
With the help of the above result we can give a partial answer to a question posed by Linden, Matu\v{s}, Ruskai and Winter in \cite{linden2013quantum}: Do the entropy vectors of stabilizer states respect all linear rank inequalities? Below, we prove that the entropies of stabilizer states with square free local dimension respect all balanced linear rank inequalities. This result together with the results in \cite{gross2013stabilizer, linden2013quantum} is a bit disappointing, as it shows, that stabilizer states are, also from an entropic perspective, too simple to provide a model for general quantum states. On the other hand, this adds to the existing evidence that stabilizer codes should be thought of as quantum analogues of linear codes.
\begin{thm}\label{stablinprime}
	The entropy vectors of stabilizer states with a vector space as a phase space respect all balanced linear rank inequalities.
\end{thm}
\begin{proof}
	Let $\rho$ be such a stabilizer state. Then, according to Theorem \ref{classmodel} 
	\begin{equation}
		S(\rho_I)+|I|=\log\frac{|M^\omega|}{|\ker(\pi_I)\cap M^\omega|},
	\end{equation}
	where $M$ is the isotropic subspace corresponding to $\rho$. But $\ker(\pi_I)\cap M^\omega=\left(\bigcap_i\ker(\pi_i)\right)\cap M^\omega=\bigcap_i\left(\ker(\pi_i)\cap M^\omega\right)$, so according to Proposition \ref{vectolinear} that implies
	\begin{equation}
		S(\rho_I)+|I|=\dim\sum_{i\in I}\left(\ker(\pi_i)\cap M^\omega\right)^o,
	\end{equation}
	where ${}^o$ denotes the annihilator in $\left(M^\omega\right)^*$. As the correction terms proportional to $|I|$ cancel in a balanced functional, balanced inequalities that hold for subspace ranks remain valid for the entropy vectors of stabilizer states with vector phase space.
\end{proof}

Generalizing Proposition \ref{vectolinear} to modules seems difficult, as the dimension of a submodule might be ill-defined even if the supermodule is free. Also a similar statement only involving cardinalities is problematic, as $A\cap B\subset \left(A^o+B^o\right)^o$, but in general they are not equal.

It turns out though that submodules of free $R$-modules for $R=\Z_d$ with $d$ square free always have a direct sum decomposition into modules of prime order. With the help of this result we can reduce the case of square free phase space dimension stabilizers to the prime dimension case. The following results extend the Observation from \cite{appleby2012monomial} that the Heisenberg-Weyl-Group factorizes for square-free dimensions by the Fact that also stabilizer states factorize. The Lemma below is corollary of B\`{e}zout's Lemma:
\begin{lem}\label{bizlem}
	Let $z_1,...,z_k\in \Z$ pairwise coprime, $d=\prod_{i=1}^k z_i$ and $r_i=\frac{d}{z_i}$. Then there exist $a_i\in\Z$ such that
	\begin{equation}\label{biz}
	\sum_{i=1}^k a_i r_i=1.
	\end{equation}
\end{lem}
\begin{proof}
	We prove the lemma by induction over $k$. If $k=1$ then $r_1=\frac{z_1}{z_1}=1$ so $a_1=1$ does the trick. Suppose now the lemma were true up to $k-1$. As $z_{k-1}$ and $z_k$ are coprime, according to B\`{e}zout's lemma there exist $a,b\in \Z$ such that $a z_{k-1}+b z_k=1$. Then $b r_{k-1}+a r_k=\frac{d}{z_{k-1}z_k}$. As the numbers $z_j$ are pairwise coprime, so are $z_{k-1}'=z_{k-1}z_k$ and any $z_j$ with $j\not=k-1,k$, and thus, according to the induction hypothesis, we can find numbers $a_i'$ for $i=1,...,k-1$ such that $\sum_{i=1}^{k-2}a_i'z_i+a'_{k-1}z_{k-1}'=1$. So $a_i=a_i'$ for $i=1,...,k-2$, $a_{k-1}=a_{k-1}'b$ and $a_k=a_{k-1}'a$ fulfill \eqref{biz}
\end{proof}
Using this lemma we can prove the above mentioned submodule decomposition.
\begin{prop}\label{squarefreedec}
	Let $d=\prod_{i=1}^k p_i$ with $p_i$ distinct primes and $M\subset \Z_d^n\cong \Z^n_{p_1}\times ... \times \Z^n_{p_k}$ a submodule. Then $M=\bigoplus_{i=1}^k M_i$ such that $M_i\subset\Z_{p_i}^n$.
\end{prop}
\begin{proof}
	Define $r_i=\frac{d}{p_i}$ and the submodules $M_i=r_i M$. Take $x\in M_i\setminus  \{0\}$ then there is a $y\in M$ such that $x=r_i y$ and $p_i x=d y=0\mod d$. According to Lemma \ref{bizlem} there are $a_i\in\Z_d$ such that $\sum_{i=1}^k a_i r_i=1 \mod d$. Therefore $\sum_{i=1}^k a_i M_i=M$, and as obviously $\spa_{\Z_d}\bigcup_{i=1}^k M_k\subset M$ we proved $\spa_{\Z_d}\bigcup_{i=1}^k M_k=M$. But the $M_i$ are submodules and $M_i\cap M_j={0}$, so each element $m\in M$ of $M$ can be decomposed uniquely into a sum $m=\sum_{i=1}^k m_i$ such that $m_i\in M_i$.
\end{proof}
We can prove now that stabilizer states in square free dimensions are tensor products of stabilizer states of prime dimension.
\begin{thm}\label{stabsqfree}
	Let $d=\prod_{i=1}^k p_i$ with $p_i$ distinct primes and $\rho\in \mathcal{B}\left(\left(\C^d\right)^{\otimes n}\right)$ a stabilizer state. Then there is a tensor product structure $\C^d=\bigotimes_{i=1}^k \C^{p_i}$ and stabilizer states $\rho_i\in \mathcal{B}\left(\left(\C^{p_i}\right)^{\otimes n}\right)$ such that $\rho=\bigotimes_{i=1}^k \rho_i$.
\end{thm}
\begin{proof}

	Let $M\subset \Z_d^{2n}$ be the isotropic submodule corresponding to $\rho$ according to Theorem \ref{stabchar}. Then, according to Proposition \ref{squarefreedec} we can write 
	\begin{equation}
	\rho =\frac{1}{p_1^n...p_k^n}\sum_{m_1\in M_1,...,m_k\in M_k}\prod_{i=1}^k \mu_{m_i}w(m_i).
	\end{equation}
	As the stabilizer group is a representation of the additive group of $M=\bigoplus_{i=1}^k M_i$, which is a direct product of the additive groups of the $M_i$, it is well known that there exist representations of the $M_i$ such that the representation of the product group is their tensor product, i.e.\ there exists a tensor product structure $\C^d=\bigotimes_{i=1}^k \C^{p_i}$ where $w(m_i)$ acts on the $i$th tensor factor for $m_i\in M_i\subset M$. Hence we have
	\begin{equation}
		\rho=\bigotimes_{i=1}^k \frac{1}{p_k^n}\sum_{m_i\in M_i}\mu_{m_i}w(m_i)\in \mathcal{B}\left(\bigotimes_{i=1}^k \C^{p_i}\right),
	\end{equation}
	which is the desired tensor product decomposition.
\end{proof}
As the entropies of the factors of a product state are additive, Corollary \ref{stablinprime} generalizes to square free dimensions:
\begin{cor}
	The entropy vectors of stabilizer states with square free local dimension respect all balanced linear rank inequalities.
\end{cor}
\begin{proof}
	Using the fact that $S(\rho_1\otimes\rho_2)=S(\rho_1)+S(\rho_2)$ and Theorem \ref{stabsqfree}, we find that the entropy vector of any stabilizer state $\rho\in \mathcal{B}\left(\left(\C^d\right)^{\otimes n}\right)$ with $d=\prod_{i=1}^k p_i$ is
	\begin{equation}\label{entsum}
		S(\rho)=\sum_{i=1}^kS(\rho_i)
	\end{equation}
	where $\rho_i\in\mathcal{B}\left(\left(\C^{p_i}\right)^{\otimes n}\right)$ is a stabilizer state in prime local dimension. But all the summands on the right hand side of \eqref{entsum} respect all balanced linear rank inequalities according to Corollary \ref{stablinprime}, and hence so does the left hand side.
\end{proof}



\chapter{Entropy Vectors and Type Classes}\label{CYR}

The main goal of this section is to better understand the essence of the correspondence between entropy vectors and group sizes proved in \cite{chan2002relation} that was already briefly mentioned in Section \ref{classical}. It turns out that the result can be reformulated using only type classes, without reference to groups. This viewpoint also makes it possible to connect this result to representation theory and find a ``classical analogue'' of the Schur-Weyl decomposition that was briefly discussed in \cite{christandl2006structure} in a non-representation-theoretic way. This provides a relation between representation theory and the classical marginal problem in the spirit of the result from \cite{christandl2006spectra} for the quantum case. To this end we develop a clear understanding how strings and permutation modules are connected. The representation-theoretic formulation yields a formula for the restrictions of irreducible representations of the Unitary group to the symmetric group as a byproduct. The correspondence between the spectrum estimation theorem and the asymptotic equipartition property is also easily made clearer in this framework, and we argue why a simple quantum analogue of \cite{chan2002relation} cannot be expected.

\paragraph{} Throughout this chapter \emph{frequency vectors} of strings play an important role, let us therefore recall their definition.
\begin{defn}[Frequency]
	Let $\mathcal{X}$ be a finite alphabet and $x\in\mathcal{X}^q$ a string. Then we define the \emph{frequency vector} $\mathfrak{f}(x)\in\N^{\mathcal{X}}$ of $x$ by
	\begin{equation}
		\mathfrak{f}(x)_a=\left|\left\{\alpha\in[q]\,\big|\,x_\alpha=a\right\}\right|
	\end{equation}
\end{defn}

Another important notion is that of a

\begin{defn}[Type Class]
	Let $\mathcal{X}$ be a finite alphabet, $x\in\mathcal{X}^n$ a string and $f=\mathfrak{f}(x)$ its frequency vector. Then the $S_n$-orbit generated by $x$ is called the \emph{type class} of $x$, i.e.
	\begin{equation}
		T_f=S_n.x\subset \mathcal{X}^n.
	\end{equation}
\end{defn}
The type class has an index $f$ instead of $x$ because strings with the same frequency vector generate the same type class. The size of the type class is
\begin{equation}\label{typeclasssize}
	|T_f|=\frac{\ell(\lambda)!}{\lambda!}=\frac{\ell(\lambda)!}{\prod_i\lambda_i!},\ \lambda=\sh(f),
\end{equation}

\section{Type Class Characterization of Entropy Vectors}

In this section we recast the main theorem from \cite{chan2002relation} purely in terms of type classes.
Let us first recall the original theorem that connects group sizes and Entropy vectors, which was already stated in Section \ref{classinfo}:
\begin{thm}[\cite{chan2002relation}]\label{chanyeung}
	Let $X=(X_i)_{i\in [n]}$ be an $n$-partite random variable. Then there exits a sequence of tuples of finite groups $(G,G_1, G_2,...,G_n)_k, k\in\N$ with $G_i\subset G$ subgroups, such that
	\begin{equation}
		H(X_I)=\lim_{k\to\infty}\frac{1}{k}\log\frac{|G|}{\left|G_I\right|}\, \forall I\subset[n]
	\end{equation}
	where $G_I=\bigcap_{i\in I}G_i$.
	Conversely, for any group tuple $(G,G_1, G_2,...,G_n)$ there is a random variable $Y=(Y_I)_{I\subset [n]}$ such that
	\begin{equation}
		H(Y_I)=\log\frac{|G|}{\left|G_I\right|}
	\end{equation}
\end{thm}
The proof uses a construction of Young subgroups that are the symmetry groups of the (joint and marginal) type classes of $X$. It turns out that this can be formulated without reference to group theory, in fact, looking closely at their proof we see that they actually connect type class sizes and group sizes, the entropies appear as a mere corollary.

\paragraph{}Suppose now the alphabet $\mathcal{X}=\mathcal{A}^n$ is a product alphabet. Then for each subset $I\subset[n]$ we define the \emph{marginal string} $x_I$ of a string $x\in\mathcal{X}^q$ by setting
\begin{equation}
	(x_I)_\alpha=\left((x_\alpha)_i\right)_{i\in I}.
\end{equation}
This definition makes sense because it is compatible with the definition of $X_I$ for a composite random variable $X=(X_i)_{i\in[n]}$ where the $X_i$ all have values on the same finite alphabet $\mathcal{A}$. More precisely let $Z=X^n$, then $Z_I=\left(X_I\right)^n$, where $Z_I$ is the string marginal of $Z$ and $X_I$ as usual.

\paragraph{}Now observe that each type class defines a rational probability distribution, the so called \emph{empirical distribution} $p=\frac{f}{n}$. Conversely to each rational probability distribution $p\in\mathcal P^{d}\cap\Q^d$ with denominator $q$, i.e.\ such that $p q\in \N^d$, and each integer $k>0$ we get a corresponding type class $T_{kqp}$.

We are now ready to recast Chan's and Yeung's result purely in terms of type classes:
\begin{thm}\label{grouplesscy}
	Let $X=(X_i)_{i\in [n]}$ be an $n$-partite random variable with rational probability distribution $p$ of denominator $q$, each $X_i$ has values on the same finite alphabet $\mathcal{A}$ of size $d$, so $X$ has values on $\mathcal{X}=\mathcal{A}^n$. Then for each integer $k$ there exists a random variable $Y^{(k)}=(Y^{(k)}_{i})_{i\in[n]}$ such that $Y^{(k)}_{I}$ is uniformly distributed on the type class $T_{qkp_I}\subset\left(\mathcal{A}^{|I|}\right)^{qk}$ for all $I\subset [n]$.
\end{thm}
\begin{proof}
	First write the type class $T_{kqp}$ as an $S_{kq}$-orbit, i.e.
	\begin{equation}
		T_{kqp}=S_{kq}.x
	\end{equation}
	with a fixed string $x\in T_{kqp}$. Now observe that $x_I\in T_{kqp_I}$ and marginalization of strings commutes with the action of the symmetric group, hence
	\begin{eqnarray}\label{typemarginal}
		\left(T_{kqp}\right)_I&=&\left(S_{kq}.x\right)_I\nonumber\\&=&S_{kq}.x_I\nonumber\\&=&T_{kqp_I}.
	\end{eqnarray}
	Let $Y$ be a random variable uniformly distributed on $T_{kqp}$, $p'$ its probability distribution and $p'_I$ the distribution of $Y_I$ which is defined as the marginal string of $Y$. Take any two elements $z,t\in T_{kqp_I}$ and let $\pi\in S_{kq}$ such that $\pi.z=t$. Now calculate 
	\begin{eqnarray}
		p'_I(t)&=&p'_I(\pi.z)\nonumber\\&=&\sum_{\substack{x\in T_{kqp}\\x_i=(\pi.z)_i\, \forall i\in I}}p'(x)\nonumber\\&=&\sum_{\substack{x\in T_{kqp}\\x_i=z_i\, \forall i\in I}}p'(\pi^{-1}x)\nonumber\\&=&\sum_{\substack{x\in T_{kqp}\\x_i=z_i\, \forall i\in I}}p'(x)=p'_I(z),
	\end{eqnarray}
	where in the second to last equality we used uniformity of $p'$. This proves that $p'_I$ is uniformly distributed on its support, which is equal to the type class $T_{kqp_I}$ according to \eqref{typemarginal}, hence $Y$ is the desired random variable.
\end{proof}
In the above theorem the random variable $Y$ has strings of product letters as values. The random variable obtained by taking the first letter of the string is distributed identically to $X$, as by definition the frequency of the letter $s$ is $kqp(s)$ and $Y$ is uniformly distributed on its support $T_{kqp}$ which is invariant under permuting the $kq$ letters in the string.

Now observe that type class sizes are asymptotically related to entropies. In particular,
\begin{eqnarray}
\lim_{k\to\infty}\frac{1}{qk}\log|T_{qkp_I}|&=&\lim_{k\to\infty}\frac{1}{qk}\log\frac{(qk)!}{\prod_{a\in\mathcal{A}} (qkp(a))!}\nonumber\\
&=&\lim_{k\to\infty}\frac{1}{qk}\log\frac{(qk)^{qk}}{\prod_{a\in\mathcal{A}} (qkp(a))^{qkp(a)}}\nonumber\\
&=&-\lim_{k\to\infty}\frac{1}{qk}\log\prod_{a\in\mathcal{A}} (p(a))^{qkp(a)}\nonumber\\
&=&H(p_I),
\end{eqnarray}
where for the second equality we used sterlings approximation on the factorials and discarded sub-exponential factors. Furthermore we can relate the size of a type class to the number of cosets of a Young subgroup, i.e.\ for any $x_I\in T_{qkp_I}$
\begin{eqnarray}\label{typelagrange}
|T_{qkp_I}|=|S_{qk}x_I|=\frac{\left|S_{qk}\right|}{\left|S_{qk,x_I}\right|},
\end{eqnarray}
where $S_{qk,x}$ is the \emph{stabilizer subgroup} or \emph{stabilizer} of $x$ in $S_{qk}$, which is isomorphic to $S_{qkp}=S_{kqp_1}\times S_{kq p_2}\times ...\times S_{kqp_d}$. In the second equation we used Lagrange's theorem for the number of cosets. But $S_{qk,x_I}$ can be expressed in terms of the stabilizers $S_{qk,x_i}$, $i\in I$, as an element $\pi\in S_{qk}$ stabilizes $x_I$ exactly if it stabilizes $x_i$ for all $i\in I$. This implies that
\begin{equation}
 S_{qk,x_I}=\bigcap_{i\in I}S_{qk,x_i}
\end{equation}
and Theorem \ref{chanyeung} follows together with \eqref{typelagrange}. The construction is also interesting in its own right because it proves that every ray in $\overline\Sigma_n$ can be approached by random variables with flat marginals \cite{chan2002relation}, which is a stronger statement for the classical entropy cone than we could derive in Chapter \ref{differential} for the classical and quantum cones with the help of the differential of the entropy function.

\section{Permutation Modules and the Classical Marginal Problem}

\subsection{Strings and Permutation Modules: The Bipartite Case}\label{strings}

Let us now connect Theorem \ref{grouplesscy} to th representation theory of the symmetric group. To this end we look at the permutation modules $M^\lambda$ introduced in Section \ref{symirreps}, as their natural basis is the type class of the probability distribution $p=\lambda/n$ if $\lambda\vdash n$. In the following section I describe this correspondence and investigate its implications.

\paragraph{}Given an alphabet $\mathcal{X}$ of $d$ letters and a string $x\in \mathcal{X}^n$, relabel the letters such that $1$ is the most frequent symbol, $2$ is the second most frequent etc. Then we can encode the information in $x$ in the tableau $T$ with $\sh(T)=\lambda$ such that it has the positions of the ones in the first row, the positions of the twos in the second and so on, so $\lambda\vdash(n,d)$ is the type of $x$. As an example, look at the string
\begin{equation}
	x=312313231.
\end{equation}
Now relabel the alphabet,
\begin{equation}\label{string}
	x=123121312,
\end{equation}
And put the positions in a Young tableau,
\begin{equation}
	T_x=\Yvcentermath1\young({1468},{259},{37}).
\end{equation}
The action $S_n\looparrowright M^{\lambda}$ restricted to the basis given by row-standard tableaux is exactly the natural action of $S_n$ on strings of type $\lambda$. An equivalent way of looking at this representation is to describe it as an action of $S_n$ on the dissections of the set $[n]$ of shape $\lambda$. Let us first define a
\begin{defn}[Dissection]
 Let $S$ be a set. A dissection of $S$ into $k$ parts is a set partition with labeled parts, i.e.\ a tuple $d\in\left(2^S\right)^k$ such that $d_i\cap d_j=\emptyset$ for $i\not=j$ and $\bigcup_{i=1}^{k}d_i=S$. More generally, we also allow a dissection of $S$ to be indexed by an arbitrary set $I$.
\end{defn}
The latter part of the definition has some advantages, as it captures the structure of the object the best. Labels are needed (see the remark below), but the order can be chosen arbitrarily, or  in other words, forgotten about.

There is a one-to-one correspondence between strings and dissections. Let $s\in[k]^n$ be a string, then we define the corresponding dissection $d^{(x)}$ by $i\in d^{(x)}_{x_i}$, i.e.\ $d^{(x)}_{j}$ is the set of indices carrying the letter $j\in[k]$ in $x$. Using this correspondence we define the frequency of a dissection by
\begin{equation}
 f^{(d)}_i=\left|d_i\right|
\end{equation}
and, accordingly, the shape of a dissection $\sh(d)=\sh\left(f^{(d)}\right)$.
The natural $S_k$- and $S_n$-actions are in a way dual to the string picture, as in this case $S_k$ acts by permuting the indices and $S_n$ acts element-wise. Also the correspondence proves immediately that $\C S_n d\cong M^{\sh(d)}$.

\begin{rem}
 Why are we using dissections instead of set partitions? Dissections of the same shape generate isomorphic $S_n$-modules, i.e.\ $\C S_n d\cong M^{\sh(d)}$, but in the case of degeneracies in the frequency vector, or, equivalently, in the corresponding partition, the order of the parts plays a role. The corresponding fact in the string case is that we do not care about the order of the alphabet because different orders generate isomorphic $S_n$-modules, but it is important to \emph{fix} an order. This is understood best by giving an extreme example: For $n\le k$, an atomic dissection of $[n]$ into $n$ singletons generates $M^{(1,...,1)}=\C S_n$, but the corresponding set partition  is invariant under the $S_n$-action.
\end{rem}

This viewpoint makes it quite simple to find the decomposition of the tensor product $M^\lambda\otimes M^\mu$ in terms of permutation modules. The action of $S_n$ on this tensor product is given by permuting the elements of $[n]$ and keeping track of two dissections of it. But that is equivalent to keeping track of their \emph{coarsest common refinement}. If we define that having two indices, we avoid the problem of fixing an order on the Cartesian product of two ordered sets:
\begin{defn}[Coarsest Common Refinement]
  Let $d_1$ and $d_2$ be dissections of a set $S$ indexed by sets $I_1$ and $I_2$. Then we define the coarsest common refinement to be the dissection $d$ of $S$ indexed by $I_1\times I_2$ satisfying
  \begin{equation}
   d_{ij}=(d_1)_i\cap (d_2)_j
  \end{equation}
\end{defn}
It has now already become obvious that the tensor product of two permutation modules is isomorphic to a direct sum of permutation modules:
\begin{equation}
	M^\lambda\otimes M^\mu\cong\bigoplus_{\nu\vdash n} \left(M^\nu\right)^{\oplus h_{\lambda\mu}^\nu}.
\end{equation}
The multiplicities $h_{\lambda\mu}^\nu$ are exactly the ones defined in \cite{christandl2006structure}, Section 2.3.3, in the string analogue of the Clebsch-Gordan isomorphism. Note that $h_{\lambda\mu}^\nu=0$ if $\nu\succ\lambda$ or $\nu\succ \mu$, a partition always dominates its refinement. We call the coefficients $h_{\lambda\mu}^\nu$ \emph{classical Kronecker coefficients}.

\paragraph{}The interesting thing to understand is now how multiplicities larger than one emerge. For this purpose, consider first the $S_n\times S_k$-action on the set of dissections of $S$ labeled by $I$, where $|S|=n$ and $|I|=k$. There, the set of orbits is indeed labeled by the shapes of the dissections, in other words, the shape is the only invariant. The task of finding the multiplicities of $M^\lambda$ in $M^\mu\otimes M^\nu$ is equivalent to understanding the $S_n\times S_k$-invariants of $S$-dissections labeled by $I\times I$. The $S_n$ part is still easily treated in the same way: It permutes the content of the sets and hence reduces the problem to treating set sizes and labels. Set sizes remain invariant. Nontrivial multiplicities arise now because the $S_n\times S_k$ action does not generate all of $S_{k^2}$, the permutation group of $I\times I$.

\paragraph{}For an illustrative example we switch back to the string picture for a moment, take a string $x=1122$ and pair it with $y=1112$ and with $(14).y=2111$. One can do this 
by writing the two strings below each other and view pairs of symbols as new symbols, i.e.
\begin{equation}
	\begin{array}{c||c|c|c|c||c}
		\text{String}&&&&&\text{Type}\\\hline
		x&1&1&2&2&(2^2)\\
		y&1&1&1&2&(31)\\
		(14).y&2&1&1&1&(31)\\
		x\otimes y&(1,1)&(1,1)&(2,1)&(2,2)&(21^2)\\
		x\otimes(14).y&(1,2)&(1,1)&(2,1)&(2,1)&(21^2)
	\end{array}
\end{equation}

Let us look at the same product in the set dissection picture. The dissections are
\begin{eqnarray}
	x&=&\left(\{12\},\{34\}\right)\nonumber\\
	y&=&\left(\{123\},\{4\}\right)\nonumber\\
	(14).y&=&\left(\{234\},\{1\}\right).
\end{eqnarray}
Writing down the refinements and ordering the indices lexicographically we get
\begin{eqnarray}
	x\otimes y&=&\left(\{12\},\emptyset,\{3\},\{4\}\right)\nonumber\\
	x\otimes(57).y&=&\left(\{2\},\{1\},\{34 \},\emptyset\right).
\end{eqnarray}
Both dissections have shape $(21^2)$ but the different order renders the copies of $M^{(21^2)}$ they generate distinct, as the labels of the parts with two elements, $(1,1)$ and $(2,1)$, cannot be transformed into each other by permuting $1$ and $2$.

\paragraph{}Theorem \ref{grouplesscy} and its proof using string marginalization shows now, that the decomposition of tensor products of permutation modules into permutation modules is connected to the bipartite quantum marginal problem: Two random variables with joint distribution $\lambda/q$ and marginals $\mu/q$ and $\nu/q$ exist if and only if $h^{\lambda}_{\mu\, \nu}\not=0$. This was already described in \cite{christandl2006structure}, however, this was done purely combinatorially, not connecting the results to the representation theory of the symmetric group. The above representation-theoretic formulation makes it easier to understand the similarities and differences between classical and quantum marginal problem, which will be discussed below. The relation between permutation modules and the classical marginal problem was also indicated in a remark in \cite{klyachko2004quantum} after Theorem 6.5.1, but was not elaborated.

\subsection{Shannon Type Inequalities from Permutation Modules}

As the type classes of a probability distribution are defined to have a cardinality rate asymptotically equal to the Shannon entropy of the latter, this result immediately yields simple Shannon type inequalities. For the following discussion we introduce the notation $A(k)\dot=B(k)$ for ``$A(k)$ is equal to $B(k)$ in leading order in $k$'' from \cite{cover2012elements}. The size of the type class $T_\lambda$ i.e.\ the dimension of the permutation module $M^\lambda$ is
\begin{equation}
	\dim M^\lambda=\frac{n!}{\lambda!},
\end{equation}
where $\lambda\vdash n$ and $\lambda!=\lambda_1!\lambda_2!...\lambda_d!$. Now suppose we have random variables $X_1$ and $X_2$ with distributions $p_{12}=\lambda/q$, $p_1=\mu/q$, $p_2=\nu/q$. Then $h^\lambda_{\mu\nu}\not=0$ But that means that there exist set partitions $p^\lambda,p^\mu$ and $p^\nu$ such that $p^\lambda$ is a common refinement of $p^\nu$ and $p^\mu$ (As we are only interested in whether the ``classical Kronecker coefficient'' $h^\lambda_{\mu\nu}$ is zero or not, considering partitions instead of dissections suffices here). This implies $\lambda\prec\mu,\nu$. Now we calculate
\begin{eqnarray}
	H(X_1|X_2)&=&H(X_{12})-H(X_2)\nonumber\\&\dot=&\frac{1}{kq}\left(\log|T_{k \lambda}|-\log|T_{k \mu}|\right)\nonumber\\ &=&\frac{1}{kq}\log\frac{(k\mu)!}{(k\lambda)!}\ge 0.
\end{eqnarray}
The inequality follows because for each $\mu_i$ there exists a set $\lambda|_{\mu_i}=\{\lambda_j|j\in I\}$ such that $\sum_{j\in I}\lambda_j=\mu_i$ and these sets are distinct for different $\mu_i$ (this is nothing else but saying that $ \lambda$ is a refinement of $\mu$), and therefore
\begin{equation}
	\frac{\mu!}{\lambda!}=\prod_{i}\frac{\mu_i!}{\lambda|_{\mu_i}!}\ge 1.
\end{equation}
The same is true for $k\lambda, k\mu$ for any $k$. Thereby we proved monotonicity of the Shannon entropy solely by exploiting the Asymptotic equipartition property and the string analogue of the Clebsch-Gordon isomorphism from \cite{christandl2006structure}.
Subadditivity is even simpler to prove: Following the quantum version in \cite{christandl2006structure}, Chapter 2, observe that $M^{k\lambda}\tilde{\subset\, } M^{k\mu}\otimes M^{k\nu}$ for all $k\in\N_+$ and therefore
\begin{eqnarray}
	H(X_{12})&\dot=&\frac{1}{kq}\log\dim M^{k\lambda}\nonumber\\&\le&\frac{1}{kq}\log\dim M^{k\mu}\otimes M^{k\nu}\nonumber\\&=&\frac{1}{kq}\log\left(\dim M^{k\mu}\dim M^{k\nu}\right)\nonumber\\&=&\frac{1}{kq}\left(\log\dim M^{k\mu}+\log\dim M^{k\nu}\right)\nonumber\\&\dot=&H(X_1)+H(X_2).
\end{eqnarray}
Strong subadditivity follows easily as well:
\begin{prop}[Strong subadditivity of the Shannon entropy]\label{repssa}
	The Shannon entropy is strongly subadditive, i.e.\ for three random variables $X,Y$ and $Z$ we have
	\begin{equation}
		H(XY)+H(YZ)-H(XYZ)-H(Y)\ge0
	\end{equation}
\end{prop}
\begin{proof}
	Let us first assume that the probability distributions involved are rational, $qp_y=\beta$, $qp_{xy}=\mu$, $qp_{yz}=\nu$ and $qp_{xyz}=\lambda$ for some Young diagrams $\beta,\mu,\nu,\lambda\vdash q$.
	We first express the inequality to be proven in terms of type class sizes:
	\begin{eqnarray}
		\exp qk\left(H(XY)+H(YZ)-H(XYZ)-H(Y)\right)&\dot=&\frac{\dim M^{k\mu}\dim M^{k\nu}}{\dim M^{k\lambda} \dim M^{k\beta}}\nonumber\\
		&=&\frac{(k\lambda)!(k\beta)!}{(k\mu)!(k\nu)!}.
	\end{eqnarray}
	As the diagrams correspond to a random variable, $\mu,\nu$ and $\lambda$ are refinements of $\beta$. If $\gamma$ is any refinement of $\beta$, we can relabel the parts of $\gamma$ by $\gamma_{ij}$ such that $\sum_{j}\gamma_{ij}=\beta_i$ thus defining $\gamma^{(i)}=(\gamma_{ij})_{j\in [l_i]}\vdash\beta_i$ for suitable $l_i\in \N$.
	With this definition we can write the expression above as
	\begin{eqnarray}
		\frac{(k\lambda)!(k\beta)!}{(k\mu)!(k\nu)!}&=&\prod_i \frac{(k\beta_i)!(k\lambda^{(i)})!}{(k\mu^{(i)})!(k\nu^{(i)})!}\nonumber\\
		&=&\prod_i \frac{\dim M^{\mu^{(i)}}\dim M^{\nu^{(i)}}}{\dim M^{\lambda^{(i)}}}\ge 1,
	\end{eqnarray}
	because $\lambda^{(i)}$ is a common refinement of $\mu^{(i)}$ and $\nu^{(i)}$ and therefore $M^{\lambda^{(i)}}\tilde{\subset\,}M^{\mu^{(i)}}\!\!\otimes M^{\nu^{(i)}}$. If the probability distribution is irrational, write it as the limit of a sequence of rational distributions and use the fact that the Shannon entropy is continuous.
\end{proof}
Note that the way in which we used the refinement property in the last equation of the proof reflects the fact that the conditional mutual informations can be written as convex combination of mutual informations using conditional probabilities, hence one should expect difficulties using a similar technique in the quantum case, which was treated with significantly more effort in \cite{christandl2012recoupling}.

\subsection{$S_n\times S_d$-Duality and the Classical and Quantum Marginal Problems}\label{repcy}

\paragraph{}Both classically and in quantum theory there is a correspondence between Young diagrams (or limits of those) and states. In quantum theory this is facilitated by the spectrum estimation Theorem \ref{specest} \cite{christandl2006spectra, christandl2012recoupling}, in classical theory by the asymptotic equipartition property \ref{AEP}. In the following I find the classical analogue of the Schur-Weyl decomposition of $(\C^d)^{\otimes n}$. This yields a representation-theoretic formulation of the asymptotic equipartition property that is a direct analogue of Theorem \ref{specest}. This clarifies the meaning of Theorem \ref{grouplesscy} in the classical Schur-Weyl picture and makes the impossibility of a direct generalization to the quantum case apparent.

\paragraph{}Let us look at the the representation given by the action
\begin{equation}
	S_n\stackrel{\phi}{\looparrowright}\left(\C^d\right)^{\otimes n}
\end{equation}
defined in Equation \eqref{Sntens}. But this time we do not decompose the tensor product according to Schur Weyl duality, but look at a decomposition into subrepresentations that respects a fixed product basis. This is the best we can do to find as much structure as possible while still retaining the product basis with respect to which the classical states are embedded. We cannot expect this decomposition to contain only irreducible representations.

\paragraph{}Taking an arbitrary product basis vector $\ket{v}$, $v\in[d]^n$ we see that $S_n\ket{v}$ contains all basis vectors $\ket{w}$ with $\mathfrak{f}(v)=\mathfrak{f}(w)$, or, if put it in another way, $S_n\ket{v}=T_{\mathfrak{f}(v)}$, the orbit is the subset of product basis vectors representing the type class of $f$. This decomposes $\left(\C^d\right)^{\otimes n}$ into subrepresentations that are spanned by type classes of product basis vectors,
\begin{equation}
	\left(\C^d\right)^{\otimes n}\cong\bigoplus_{f}\spa S_n\ket{f},
\end{equation}
where $\ket{f}:=\ket{1^{f_1}...d^{f_d}}$.

The Young subgroup $S_{f}=S_{f_1}\times S_{f_2}\times ...\times S_{f_d}$ is the stabilizer of $\ket{v}$, so the representation spanned by it is isomorphic to the induced representation $1 \uparrow^{S_n}_{S_{f}}\cong M^{\sh(f)}$, where $\sh(f)$ is the partition corresponding to the frequency $f$. So we found
\begin{equation}
	\left(\C^d\right)^{\otimes n}\cong\bigoplus_{\lambda\vdash(n,d)} \left(M^\lambda\right)^{\oplus m_\lambda}
\end{equation}
with multiplicities $m_\lambda$ that are still to be determined.

\paragraph{}There is of course also a natural action of $S_d$ that preserves this decomposition, that is the one permuting our fixed basis $B$ of $\C^d$ that generates the product basis. It preserves the above decomposition because it can only permute the frequencies of the basis vectors and thus does not change the corresponding partition. This means we can write
\begin{equation}
	\left(\C^d\right)^{\otimes n}\cong\bigoplus_{\lambda\vdash(n,d)} M^\lambda \otimes W(\lambda),
\end{equation}
where $W(\lambda)$ is some representation of $S_d$. Let us identify this representation. For this purpose we label the standard basis of $\left(\C^d\right)^{\otimes n}$ in a different way. Given a string $x$ we can identify it by giving the shape of it's frequency vector, $\lambda$, a permutation $\pi\in S_d$ such that $\pi(i)$ is the $i$th most frequent symbol and a permutation $\sigma\in S_n$ such that $\sigma.\left(\pi(1)^{\lambda_1}\pi(2)^{\lambda_2}...\pi(d)^{\lambda_d}\right)=x$. 

\paragraph{}This description is, however, not unique. Any $\sigma'=\sigma\tau$ for some $\tau\in S_\lambda$ can replace $\sigma$, this shows again that the corresponding representation of $S_n$ is equal to $1\!\uparrow_{S_\lambda}^{S_n}=M^\lambda$. A similar argument holds for $\pi$. Let $\mu$ be the partition corresponding to the shape of the multiplicities of entries occurring in $\lambda$. This is by far easier explained by an example than it is done with words. Take a Young diagram
\begin{equation*}
	\lambda=\Yvcentermath1\yng(5,3,2,2,1,1,1),
\end{equation*}
i.e.\ $\lambda=(5,3,2^2,1^3)$. Then the multiplicity shape is $\mu=(3,2,1^2)$, because there is one triple of equal parts, one pair and two parts that are different from all others.

Now take a string described by the triple $(\lambda,\pi,\tilde\sigma)$, where $\tilde\sigma=\sigma S_\lambda$ is a left coset of $S_\lambda$ and suppose $\lambda_i=\lambda_j$. Then the same string is also described by $(\lambda,\pi(ij),\widetilde{\sigma\tau})$, where $\tau=(a+1\, b+1)(a+2\, b+2)...(a+\lambda_i,b+\lambda_i)$ with $a=\lambda_1+...+\lambda_{i-1}$ and $b=\lambda_1+...+\lambda_{j-1}$, and $(ij)$ denotes the transposition of $i$ and $j$. This reflects the fact that two letters that occur with the same frequency can also be interchanged by a permutation of the positions within the string.

To remove this ambiguity we replace $\pi$ by $\pi S_\mu$ where $\mu$ is the shape of the multiplicity pattern of $\lambda$, i.e.\ $\mu=\sh(\mathfrak{f}(\lambda))$. This observation enables us to find the representations $W(\lambda)$. Let us find the stabilizer of $S_n\ket{\lambda}$ under $S_d$, $\lambda=\left(\lambda_1^{f_1}\lambda_{f_1+1}^{f_2}...\lambda_{d-f_r+1}^{f_r}\right)$. This is equal to the Young subgroup $S_f\subset S_d$, as permuting basis vectors that have the same frequency can also be achieved by the action of $S_n$ like in the above example. These observations show that the representation $W(\lambda)$ of $S_d$ is equal to $M^{\sh(\mathsf f(\lambda))}=:M^{\lambda^+}$, because it is the one that corresponds to the action on strings with shape $\sh(f)$. Putting things together we get
\begin{equation}\label{clschurweyl}
	\left(\C^d\right)^{\otimes n}\cong\bigoplus_{\lambda\vdash(n,d)} M^\lambda \otimes M^{\lambda^+},
\end{equation}
which is the representation theoretic formulation of the string analogue of Schur Weyl duality that has been proposed in \cite{christandl2006structure}.

\paragraph{}The following observations are not of any direct use in our strive to understand the relations between quantum and classical marginal problem, but yield a result that might be interesting in itself for representation theory and is therefore stated here, forming a short digression.

Having done two different decompositions of the tensor product space $\left(\C^d\right)^{\otimes n}$, we can find the restriction of any irreducible representation of $\mathrm{U}(d)$ to $S_d$ by comparing them. Using the decomposition \eqref{kostka} on the $S_n$-modules $M^\lambda$ in \eqref{clschurweyl}, it transforms into
\begin{equation}
	\left(\C^d\right)^{\otimes n}\cong\bigoplus_{\lambda\vdash(n,d)}\ \ \bigoplus_{\mu\succ\lambda}\C^{K_{\mu\lambda}} \otimes [\mu]\otimes M^{\lambda^+}.
\end{equation}
Upon swapping the direct sums we get
\begin{equation}
	\left(\C^d\right)^{\otimes n}\cong\bigoplus_{\mu\vdash(n,d)}[\mu]\otimes\left[\bigoplus_{\substack{\lambda\vdash(n,d)\\\lambda\prec\mu}} \C^{K_{\mu\lambda}}\otimes M^{\lambda^+} \right],
\end{equation}
and comparing this expression to the Schur Weyl decomposition \eqref{schurweyl} implies
\begin{equation}
	V_\mu\!\downarrow^{\mathrm{U}(d)}_{S_d}\cong \bigoplus_{\substack{\lambda\vdash(n,d)\\\lambda\prec\mu}} \C^{K_{\mu\lambda}}\otimes M^{\lambda^+},
\end{equation}
expressing the restriction of the $\mathrm{U}(d)$-irreducible representation $V_\mu$ to $S_d$ as a direct sum of permutation modules. Using \eqref{kostka} once more we get the decomposition into irreducible representations,
\begin{thm}\label{UtoS}
\begin{equation}
	V_\mu\!\downarrow^{\mathrm{U}(d)}_{S_d}\cong\bigoplus_{\nu\vdash d}[\nu]^{\oplus \eta_{\nu\mu}}
\end{equation}
with
\begin{equation}
	\eta_{\nu\mu}=\sum_{\substack{\lambda\vdash(|\mu|,d)\\\lambda\prec\mu\\\lambda^+\prec\nu}} K_{\mu\lambda}K_{\nu\lambda^+}.
\end{equation}

\end{thm}

\paragraph{}Let us return to the decomposition \eqref{clschurweyl} and formulate the asymptotic equipartition property (AEP) in this picture. This can be done such that it looks very similar to the spectrum estimation theorem \ref{specest}. A good reference for a classical treatment of this kind of information theoretic basics is \cite{cover2012elements}. 
\begin{prop}
	Let $\rho\in\mathcal{B}\left(\C^d\right)$ be diagonal in a fixed basis $B$, $\rho=\sum_{i}p_i\ketbra{i}{i}$ and let $Q_\lambda$ be the projector onto $M^\lambda\otimes M^{\lambda^+}$ in \eqref{clschurweyl} defined with respect to $B$. Further define
	\begin{equation}
		B_{n,d,\epsilon}(r):=\left\{\lambda\vdash(n,d)\Big|\left\|\overline\lambda-r\right\|\le \epsilon\right\},
	\end{equation}
	with $r=\spec\rho=p\downarrow$. Then
	\begin{equation}
		\lim_{n\to\infty}\sum_{\lambda\in B_{n,d,\epsilon}(r)}\tr Q_\lambda \rho^{\otimes n}=1
	\end{equation}
	for all $\epsilon>0$.
\end{prop}
\begin{proof}
	The AEP as it is found in \cite{cover2012elements} makes, among other, the following statement: Let $A^n_{\epsilon'}=\left\{x\in\mathcal{X}^n|p(x)=2^{-n(H(x)\pm\epsilon')}\right\}$, then $\mathbf P\left(A^n_{\epsilon'}\right)\stackrel{n\to\infty}{\longrightarrow} 1$. The projectors $Q_\lambda$ are $S_d$-invariant, let therefore without loss of generality $p_1\ge p_2\ge...\ge p_d$. Assume for now $p_d\not=0$. Take $\lambda\vdash(n,d)$ such that $x_\lambda\in A^n_{\epsilon'}$ and calculate
	\begin{equation}
		2^{-n(H(p)+\epsilon')}\le p^n(x_\lambda)=\prod_i p_i^{\lambda_i}=\prod_i p_i^{n\left(p_i+(\overline\lambda_i-p_i)\right)}\le2^{-n\left(H(p)+\|\overline\lambda-p\|\log(p_d)\right)},
	\end{equation}
	so $2^{-n\epsilon'}\le 2^{n \|\overline\lambda-p\|\log(p_d)}$ and therefore $\|\overline\lambda-p\|\le -\frac{\epsilon'}{\log p_d}$. So we have proven that all strings in $A^n_{\epsilon'}$ have a frequency of shape $\lambda\in B_{n,d,-\frac{\epsilon'}{\log p_d}}(r)$ and the result follows by setting $\epsilon'=- \epsilon\log p_d$. The case where $p_d=0$ reduces to a lower dimension $d'<d$ because the assumption $2^{-n(H(p)+\epsilon')}\le p^n(x_\lambda)$ already implies $\lambda_d=0$ then.
\end{proof}

Now we see what the construction in \cite{chan2002relation} means in this formulation: The random variable that is uniformly distributed on the type class $T(p)$ of some rational distribution $p=f/q$ of denominator $q$ is nothing else but the projector onto one copy of $M^\lambda$ in \eqref{clschurweyl}, it makes the limit $n\to\infty$ obsolete by setting the probabilities on non-typical events to zero. Note that the discarded events include the few most probable events as well as a huge number of events of negligible probability, and that the probability of the exactly typical sequences approaches zero for large sequence lengths.

\paragraph{}It has been tried for some time to find a construction for the quantum case analogous to the result of Chan and Yeung (\cite{christandl2006spectra}, \cite{christandl2012recoupling}).

The straightforward analogue would be taking a state of the form
\begin{equation}\label{qucystate}
	\rho=\frac{1}{\dim [\lambda]}P_{[\lambda]}\otimes\ketbra{x}{x}
\end{equation}
In the Schur-Weyl picture, where $\ket{x}\in V_\lambda$.

Given a state $\rho\in\mathcal{B}\left(\C^d\right)$ with rational spectrum $p=\spec\rho\in\Q^d$ with denominator $q$, we want to find a state
\begin{equation}
\sigma\cong \frac{1}{\dim[\lambda]}P_{[\lambda]}\otimes\ketbra{x}{x}\in\mathcal{B}\left( [\lambda]\otimes V_\lambda\right)\subset\mathcal{B}\left( \bigoplus_{\mu\vdash(q,d)}[\mu]\otimes V_\mu\right)\cong\mathcal{B}\left( \left(\C^d\right)^{\otimes q}\right),
\end{equation}
where $\lambda=\sh(qp)$, such that $\rho=\tr_{1^c}\sigma=\rho$.

The task is now to identify the vector $\ket{x}\in V_\lambda$. Let us derive an expression for the partial trace $\rho=\tr_{1^c}\sigma$ for the candidate states 
\begin{equation}
\sigma=\frac{1}{\dim[\lambda]}P_{[\lambda]}\otimes\ketbra{v_{\lambda, B}}{v_{\lambda, B}},
\end{equation}
where $\ket{v_{\lambda, B}}$ is the highest weight vector in $V_\lambda$ for the maximal torus $T\subset \mathrm{U}(d)$ defined by the eigenbasis $B$ of $\rho$ and fix as ordering on the weights the  lexicographical ordering with respect to the ordering of the basis such that the first vector has the highest corresponding eigenvalue etc. We denote by $\ket{i}$ the elements of that basis and calculate
\begin{eqnarray}
	\rho_{ij}&=&\bra{i}\left(\tr_{1^c}\sigma\right)\ket{j}\nonumber\\ &=&\tr\left(\ketbra{j}{i}\tr_{1^c}\sigma\right)\nonumber\\&=&\tr\left(\ketbra{j}{i}\otimes\mathds{1}^{\otimes q-1}\sigma\right)\nonumber\\
	&=&\frac{1}{\dim[\lambda]}\tr\left(\ketbra{j}{i}\otimes\mathds{1}^{\otimes q-1}P_{[\lambda]}\hat{\otimes}\ketbra{v_{\lambda, B}}{v_{\lambda, B}}\right)\nonumber\\
	&=&\frac{1}{q!\bra{\lambda}e_\lambda^\dagger e_\lambda\ket{\lambda}}\tr\left(\ketbra{j}{i}\otimes\mathds{1}^{\otimes n-1}\sum_{\pi\in S_q}\pi^\dagger e_\lambda\ketbra{\lambda}{\lambda}e_\lambda^\dagger\pi\right)\nonumber\\
	&=&\frac{1}{q!\bra{\lambda}e_\lambda^\dagger e_\lambda\ket{\lambda}}\tr\left(\sum_{\pi\in S_q}\pi\ketbra{j}{i}\otimes\mathds{1}^{\otimes n-1}\pi^\dagger e_\lambda\ketbra{\lambda}{\lambda}e_\lambda^\dagger\right)\nonumber\\
	&=&\frac{\dim[\lambda]^2}{q!^3\bra{\lambda}e_\lambda^\dagger e_\lambda\ket{\lambda}}\tr\left(e_\lambda^\dagger\sum_{\pi\in S_q}\pi\ketbra{j}{i}\otimes\mathds{1}^{\otimes n-1}\pi^\dagger e_\lambda e_\lambda\ketbra{\lambda}{\lambda}e_\lambda^\dagger\right)\nonumber\\
	&=&\tr\left(V_\lambda(\ketbra{j}{i})\ketbra{v_{\lambda, B}}{v_{\lambda, B}}\right)=\begin{cases}
	                                                                                 	\lambda_i&i=j\\
	                                                                                 	0&\text{else}
	                                                                                 \end{cases}
\end{eqnarray}
Here we have put a hat on the tensor product in the second line to distinguish between the two different tensor structures of the same space present in that expression, corresponding to the notation
\begin{equation*}
	\left(\C^d\right)^{\otimes n}\cong \bigoplus_{\lambda\vdash(n,d)}[\lambda]\hat\otimes V_\lambda
\end{equation*}
for the Schur-Weyl decomposition \eqref{schurweyl}. $e_\lambda$ denotes the Young symmetrizer corresponding to the standard tableau $T_\lambda$ of shape $\lambda$ with numbers 1 to $n$ inserted from left to right and from top to bottom, e.g.
\begin{eqnarray*}
	\lambda&=&\Yvcentermath1\yng(2,1)\\
	T_\lambda&=&\Yvcentermath1\young({12},3).
\end{eqnarray*}
In line two we used that in the Schur Weyl picture 
\begin{equation}\label{hwv}
\ket{v_{\lambda, B}}=\frac{e_T\ket{T}}{\sqrt{\bra{T}e_T^\dagger e_T\ket{T}}}
\end{equation}
is the normalized highest weight vector in the Weyl module $e_T\left(\C^d\right)^{\otimes n}$ (see \cite{christandl2006structure}, Lemma 1.22). From the second to the third line we used the fact that the Young symmetrizer is proportional to a projector \eqref{youngproj}. In the last line the identity \eqref{hwv} is used again and $\frac{1}{q!}\sum_{\pi\in S_q}\pi\ketbra{j}{i}\otimes\mathds{1}^{\otimes n-1}\pi^\dagger$ is just the diagonal representation of the Lie algebra element $E_{ji}=\ketbra{j}{i}$, hence 
\begin{equation}
V_\lambda(\ketbra{j}{i})=e_\lambda^\dagger\sum_{\pi\in S_q}\pi\ketbra{j}{i}\otimes\mathds{1}^{\otimes n-1}\pi^\dagger e_\lambda
\end{equation}
is the $\mathfrak{gl}(d)$-representation with highest weight vector $e_\lambda\ket{\lambda}$.

To return to a basis free expression we get
\begin{equation}
	\rho=V_\lambda^\dagger\left(\ketbra{v_{\lambda, B}}{v_{\lambda, B}}\right)
\end{equation}
where here $V_\lambda$ denotes the Lie algebra representation map and $V_\lambda^\dagger$ denotes its adjoint as an element of $\hom\left(\mathfrak{gl}(d),\mathfrak{gl}\left(d^n\right)\right)$.

\paragraph{}We have therefore found a quantum analogue of the Chan Yeung construction for $n=1$: Given a state $\rho=\frac{1}{q}\sum_{i=1}^d\lambda_i\ketbra{i}{i}$ we choose the maximal torus corresponding to the eigenbasis of $\rho$ and the state $\sigma=S_q.\ketbra{v_{\lambda, B}}{v_{\lambda, B}}$ is supported only on the copy of $[\lambda]$ corresponding to the coherent state $e_T\ket{T}$ where $\ket{T}$ is the product basis vector that has $\ket{i}$ at the positions in row $i$ of $T$ and $\sh(T)=\lambda$. Also it has the original state as one body marginal. 

\paragraph{}Unfortunately it is obvious that this construction does not commute with the partial trace in case of a multipartite system: The highest weight vector $\ket{v_{\lambda, B}}$ itself is, in general, entangled with respect to the tensor product structure of the physical Hilbert space and therefore contributes to the entropy of reduced states, and in particular, will not yield reduced states with flat spectra.


\chapter{Summary and Open Questions}

In this thesis I have approached the problem of characterizing multipartite quantum entropies from different perspectives. This led to a variety of insights:

\begin{itemize}
 \item The quantum entropy cone is more symmetric than its classical analogue.
 \item On the other hand it is less structured in the sense that finding quantum information inequalities cannot be reduced to finding balanced information inequalities, as it is the case for the classical entropy cone.
 \item There are weak monotonicity inequalities that define facets of the quantum entropy cone. They have a structure similar to the monotonicity inequalities that define facets of the classical entropy cone.
 \item Quantum states whose entropy vectors lie on an extremal ray of the quantum entropy cone have a very simple structure: Their marginals have only one nonzero eigenvalue each. The same is true for the classical analogue.
 \item Entropies from stabilizer states have more structure than previously known, rendering them uninteresting to characterize the full entropy cone as well as to achieve the capacity of a general quantum network communication scenario. However, the structural insight gained in this thesis shows from the entropic perspective that stabilizer codes are quantum analogues of linear codes. Linear codes are useful in classical network coding, indicating that the same should be true for stabilizer states and future quantum network coding.
 \item The group characterization theorem \cite{chan2002relation} for Shannon entropy vectors can be reformulated in a purely combinatoric language.
 \item There is a connection between the classical marginal problem and the representation theory of the symmetric group. The corresponding formalism simplifies the comparison between quantum and classical marginal problem
 
\end{itemize}

The main questions that remained open are the following:

\begin{itemize}
 \item Are all inequalities that were shown to be convex independent for the von Neumann cone \cite{pippenger2003inequalities} essential for the real quantum entropy cone as well?
 \item Can the reasoning behind Theorem \ref{poprays} be generalized to states whose entropy vectors are \emph{close} to an extremal ray?
 \item Do stabilizer states in square dimension respect balanced linear rank inequalities? A solution would also answer the question whether Abelian codes are more powerful than linear codes.
 \item Is it feasible to calculate the marginals of a state like \eqref{qucystate}? Are they supported on the typical subspaces corresponding to spectra close to the ones of the marginals of the parent state?
\end{itemize}


\appendix

\chapter{Tripartite Quantum Marginal Problem}

During the time I did the research for this Thesis, I also tried to generalize a result on the quantum marginal problem for tripartite mixed states by Christandl, \c Sahino\u glu and Walter \cite{christandl2012recoupling}. This did not lead to any significant results, but some parts of the proof of the main theorem of \cite{christandl2012recoupling} were quite difficult to understand, so I record my slightly more detailed reformulation of the proof in this appendix for future benefit.

Let us first set the scene for understanding the result of Christandl, \c{S}ahino\u{g}lu and Walter \cite{christandl2012recoupling}, we use the notation from this paper in the following.

We want to use the Clebsch-Gordan-isomorphism
\begin{equation}\label{cgiso}
	[\alpha]\otimes[\beta]\cong\bigoplus_\lambda[\lambda]\otimes H^{\alpha\beta}_\lambda,
\end{equation}
where $\alpha,\beta$ are Young diagrams, the direct sum is taken over all Young diagrams, and $H^{\alpha\beta}_\lambda$ is a multiplicity space with dimension equal to the \emph{Kronecker coefficient} $g_{\alpha\beta\lambda}$. Note that this dimension can be zero and in fact is nonzero only for finitely many Young diagrams $\lambda$ for given $\alpha$ and $\beta$.

Consider the following three alternative ways of decomposing $[\alpha]\otimes[\beta]\otimes[\gamma]$ into irreducible representations of $S_n$, all three of the using the Clebsch-Gordan isomorphism \eqref{cgiso} twice:
\begin{eqnarray}
	[\alpha]\otimes[\beta]\otimes[\gamma]&\cong&\bigoplus_\eta [\eta]\otimes [\gamma]\otimes H^{\alpha\beta}_\eta\cong\bigoplus_{\eta,\lambda}[\lambda]\otimes H^{\eta\gamma}_\lambda\otimes H^{\alpha\beta}_\eta\label{recoup}\nonumber\\
	{}[\alpha]\otimes[\beta]\otimes[\gamma]&\cong&\bigoplus_\eta [\eta]\otimes [\alpha]\otimes H^{\beta\gamma}_\eta\cong\bigoplus_{\eta,\lambda}[\lambda]\otimes H^{\alpha\eta}_\lambda\otimes H^{\beta\gamma}_\eta\nonumber\\
	{}[\alpha]\otimes[\beta]\otimes[\gamma]&\cong&\bigoplus_\eta [\eta]\otimes [\beta]\otimes H^{\alpha\gamma}_\eta\cong\bigoplus_{\eta,\lambda}[\lambda]\otimes H^{\eta\beta}_\lambda\otimes H^{\alpha\gamma}_\eta
\end{eqnarray}
It is natural to ask now how the different decompositions are related. One tool is given by the so-called \emph{recoupling coefficients}. The terminology is quite misleading here, as the recoupling coefficients are not, in general, coefficients, but maps relating the different decompositions \eqref{recoup}. To understand their definition, look at the chain of maps
\begin{equation}
	H^{\mu\gamma}_\lambda\otimes H^{\alpha\beta}_\mu\hookrightarrow \bigoplus_{\eta}H^{\eta\gamma}_\lambda\otimes H^{\alpha\beta}_\eta\stackrel{\sim}{\longrightarrow}\bigoplus_{\eta} H^{\alpha\eta}_\lambda\otimes H^{\beta\gamma}_\eta\twoheadrightarrow  H^{\alpha\nu}_\lambda\otimes H^{\beta\gamma}_\nu
\end{equation}
for fixed Young diagrams $\alpha,\beta,\gamma, \lambda, \mu$ and $\nu$. The first map is just the natural embedding, the second is the isomorphism resulting from the fact that the two direct sums are just different decompositions of the multiplicity space of $[\lambda]$ in $[\alpha]\otimes [\beta]\otimes[\gamma]$ and the last map is the projection onto the specified direct summand. The composition of the three maps is the recoupling coefficient
\begin{equation}
	\left[\begin{array}{ccc}
	      	\alpha&\beta&\mu\\\gamma&\lambda&\nu
	      \end{array}\right]:H^{\mu\gamma}_\lambda\otimes H^{\alpha\beta}_\mu\to H^{\alpha\nu}_\lambda\otimes H^{\beta\gamma}_\nu.
\end{equation}
The decomposition \eqref{recoup} enables us to write down three decompositions of the Hilbert space of a tripartite system to the $n$-th tensor power via Schur-Weyl duality, i.e.\ a Hilbert space 
\begin{equation}
\mathcal{H}=\left(\mathcal{H}_A\otimes\mathcal{H}_B\otimes\mathcal{H}_C\right)^{\otimes n}\cong\mathcal{H}_A^{\otimes n}\otimes\mathcal{H}_B^{\otimes n}\otimes\mathcal{H}_C^{\otimes n}
\end{equation}
 can be decomposed by first applying Schur-Weyl duality separately for $\mathcal{H}_A^{\otimes n}$, $\mathcal{H}_B^{\otimes n}$, and $\mathcal{H}_C^{\otimes n}$ each and then using one of the above decompositions of the triple products of irreducibles of $S_n$, e.g.
\begin{equation}\label{schurweyl+recoup}
	\mathcal{H}\cong\bigoplus_{\alpha\beta\gamma}[\alpha]\otimes[\beta]\otimes[\gamma]\otimes V_\alpha\otimes V_\beta\otimes V_\gamma\cong\bigoplus_{\alpha\beta\gamma\eta\lambda}[\lambda]\otimes H^{\eta\gamma}_\lambda\otimes H^{\alpha\beta}_\eta\otimes V_\alpha\otimes V_\beta\otimes V_\gamma
\end{equation}

For each quintuple $(\alpha,\beta,\gamma,\mu,\lambda)$ of Young diagrams define the operator
\begin{equation}\label{pathproj}
	Q_{\gamma\mu,\lambda}^{\alpha\beta\gamma}=(P_\alpha\otimes P_\beta\otimes P_\gamma)(P_\mu\otimes P_\gamma)P_\lambda
\end{equation}
and similar operators for the other decompositions according to \eqref{recoup}. These operators are, in fact, orthogonal projectors, as the three projectors in their definition are all block diagonal and either zero or the identity on each block in the direct sum decomposition on the right hand side of \eqref{schurweyl+recoup}. One can now easily relate the product of two such projectors to a recoupling coefficient:
\begin{equation}\label{projrecoup}
	Q_{\gamma\mu,\lambda}^{\alpha\beta\gamma}Q_{\alpha'\mu,\lambda'}^{\alpha'\beta'\gamma'}=\delta_{\alpha\alpha'}\delta_{\beta\beta'}\delta_{\gamma\gamma'}\delta_{\lambda\lambda'}P_\lambda\otimes\left[\begin{array}{ccc}
	      	\alpha&\beta&\mu\\\gamma&\lambda&\nu
	      \end{array}\right]\otimes \mathds{1}_{V_\alpha\otimes V_\beta\otimes V_\gamma}
\end{equation}
The product vanishes for $\alpha\not=\alpha'$ etc., because $P_\alpha P_\alpha'=\delta_{\alpha\alpha'}P_\alpha$ and, as already mentioned, the three projectors in the right hand side of \eqref{pathproj} commute.

The result of Christandl, \c{S}ahino\u{g}lu and Walter reads as follows:
\begin{thm}[\cite{christandl2012recoupling}]\label{recoupspec}
	There exists a finite dimensional Hilbert space $\mathcal{H}$ and a state $\rho\in\mathcal{B}\left(\mathcal{H}^{\otimes 3}\right)$ with spectra $s=(r_{ABC},r_{AB},r_{BC}, r_A,r_B,r_C)$ if and only if there is a sequence of Young diagrams $(\lambda,\mu, \nu,\alpha,\beta,\gamma)_k,\ k\in\N$ such that
	\begin{equation}
		\lim_{k\to\infty}\frac{1}{k}(\lambda,\mu, \nu,\alpha,\beta,\gamma)_k=s
	\end{equation}
	and
	\begin{equation}
		\left\|\left[\begin{array}{ccc}
	      	\alpha&\beta&\mu\\\gamma&\lambda&\nu
	      \end{array}\right]\right\|\ge\frac{1}{f(k)}
	\end{equation}
	for some polynomial $f$.
\end{thm}

The argument for the proof of the ``only if'' direction of Theorem \ref{recoupspec} roughly goes like this: Define for each local dimension $d$, each $k\in \N$ and each spectrum $r$ and $\delta>0$ the set 
\begin{equation}
	\Delta(k,\delta,r)=\left\{\lambda\vdash(k,\ell(r))\big|||\overline{\lambda}-r||_1\le\delta\right\},
\end{equation}
with $\ell(r)$ being the length of $r$, i.e.\ the dimension of the underlying Hilbert space. Furthermore, define the projectors
\begin{eqnarray}\label{smearedprojs}
	\tilde{Q}^{k,\delta}_1&=&\bigoplus Q^{\alpha\beta\gamma}_{\mu\gamma, \lambda}\text{ and}\nonumber\\
	\tilde{Q}^{k,\delta}_2&=&\bigoplus Q^{\alpha\beta\gamma}_{\nu\alpha,\lambda},
\end{eqnarray}
where the direct sums are taken over all diagrams $\alpha\in\Delta(k,\delta,r_A)$, $\nu\in\Delta(k,\delta,r_{BC})$ etc.
Now, observe that for any contraction $P$, any projector $Q$ and an arbitrary density matrix $\rho$ we have the elementary bound
\begin{equation}\label{elbound}
	|\tr PQ\rho|=|\tr P\rho-\tr P\overline{Q}\rho|\ge \tr P\rho-\tr\overline{Q}\rho,
\end{equation}
which implies, together with Theorem \ref{specest} that for all $\epsilon> 0$ there exists a $\delta>0$ such that
\begin{equation}
	\left\|\tilde{Q}^{k,\delta}_1\tilde{Q}^{k,\delta}_2\right\|_\infty\ge\left|\tr\tilde{Q}^{k,\delta}_1\tilde{Q}^{k,\delta}_2\rho^{\otimes k}\right|\ge 1-2\epsilon.
\end{equation}
According to \eqref{projrecoup} all terms in $\tilde{Q}^{k,\delta}_1\tilde{Q}^{k,\delta}_2$ where $\alpha\not=\alpha'$ or $\beta\not=\beta'$ etc. vanish such that
\begin{equation}\label{bigprojprod}
	\tilde{Q}^{k,\delta}_1\tilde{Q}^{k,\delta}_2=\sum_{\alpha\beta\gamma\mu\nu\lambda} Q^{\alpha\beta\gamma}_{\mu\gamma, \lambda}Q^{\alpha\beta\gamma}_{\nu\alpha,\lambda},
\end{equation}
the sum again taken over all diagrams with normalization $\delta$-close to the spectra of $\rho$. Therefore we can conclude that
\begin{equation}
	N \max_{\alpha\beta\gamma\mu\nu\lambda}\left\|Q^{\alpha\beta\gamma}_{\mu\gamma, \lambda}Q^{\alpha\beta\gamma}_{\nu\alpha,\lambda}\right\|_\infty\ge\sum_{\alpha\beta\gamma\mu\nu\lambda} \left\|Q^{\alpha\beta\gamma}_{\mu\gamma, \lambda}Q^{\alpha\beta\gamma}_{\nu\alpha,\lambda}\right\|_\infty\ge\left\|\tilde{Q}^{k,\delta}_1\tilde{Q}^{k,\delta}_2\right\|_\infty\ge\left|\tr\tilde{Q}^{k,\delta}_1\tilde{Q}^{k,\delta}_2\rho^{\otimes k}\right|\ge 1-2\epsilon
\end{equation}
With $N$ being the number of terms in the sum in \eqref{bigprojprod}, which is at most polynomial in $k$, and $\epsilon_k=q(k)e^{-k \delta^2}$ for some fixed polynomial $q$. That is , for every $\delta>0$ there is a polynomial $p_\delta$  and a series of tuples of Young diagrams $(\lambda,\mu, \nu,\alpha,\beta,\gamma)^{(\delta)}_k,\ k\in\N$ such that
\begin{equation}
	\left\|Q^{\alpha\beta\gamma}_{\mu\gamma, \lambda}Q^{\alpha\beta\gamma}_{\nu\alpha,\lambda}\right\|_\infty\ge\frac{1-2\epsilon}{p_\delta(k)}
\end{equation}
and the normalization of all involved diagrams are $\delta$-close to their partner spectrum. What is more, we can choose the set $\{p_\delta|\delta\in\R_+\}$ such that $p_\delta(k)\le p_{\delta'}(k)$ for all $\delta\le\delta'$ and all $k\in\N$. With this we can prove that the series of diagrams $(\lambda,\mu, \nu,\alpha,\beta,\gamma)^{(1/\sqrt[4]{k})}_k$ converges to $s$ and the product of the corresponding projectors, $Q^{\alpha\beta\gamma}_{\mu\gamma, \lambda}Q^{\alpha\beta\gamma}_{\nu\alpha,\lambda}$ decreases at most polynomially in $k$ with respect to the norm $||\cdot||_\infty$. On the other hand, this norm is equal to any norm of the corresponding recoupling coefficient up to a factor polynomial in $k$.

We can now strengthen this direction of the theorem.
\begin{thm}
	For any quantum state $\rho\in\mathcal{B}\left(\C^{\otimes 3}\right)$ with spectra $s=(r_{ABC},r_{AB},r_{BC},r_{AC},r_A,r_B,r_C)$ there exists a series of Young diagrams $x_k=(\lambda,\mu,\nu,\sigma,\alpha,\beta,\gamma)_k, k\in\N$ such that
	\begin{equation}
		\lim_{k\to\infty}\frac{1}{k} x_k=s
	\end{equation}
	and the operator norm of one of  $Q^k_1Q^k_2Q^k_3$ and cyclic permutations decay at most polynomially, where the three projectors $Q^k_i$, $i=1,2,3$ correspond to the three ways of decomposing $([\alpha]\otimes[\beta]\otimes[\gamma])_k$. More precisely there exist polynomials $p_i, i=0,1,2$ such that
	\begin{equation}
		\left\|Q^k_{(123)^i(1)}Q^k_{(123)^i(2)}Q^k_{(123)^i(3)}\right\|_\infty\ge\frac{1}{\left|p_i(k)\right|}, i=0,1,2.
	\end{equation}
	In particular, the recoupling coefficients
	\begin{equation}
		\left[\begin{array}{ccc}
	      	\alpha&\beta&\mu\\\gamma&\lambda&\nu
	      \end{array}\right],\left[\begin{array}{ccc}
	      	\beta&\gamma&\nu\\\alpha&\lambda&\sigma
	      \end{array}\right]\text{ and }\left[\begin{array}{ccc}
	      	\gamma&\alpha&\sigma\\\beta&\lambda&\mu
	      \end{array}\right]
	\end{equation} 
	do not decay exponentially.
\end{thm}
Note that there are actually there are six recoupling coefficients, but the ones missing in the theorem are adjoints of the above.
\begin{proof}
	We use the projectors \eqref{smearedprojs} and define the analogous one corresponding to $Q^{\alpha\beta\gamma}_{\sigma\beta,\lambda}$, $\tilde{Q}^k_3$. Using the bound \eqref{elbound} twice we get
	\begin{equation}
		\left|\tr\tilde{Q}^k_1\tilde{Q}^k_2\tilde{Q}^k_3\rho^{\otimes k}\right|\ge\left|\tr\tilde{Q}^k_1\tilde{Q}^k_2\rho^{\otimes k}\right|-\tr\overline{\tilde{Q}_3^k}\rho^{\otimes k}\ge\tr\tilde{Q}^k_1\rho^{\otimes k}-\tr\overline{\tilde{Q}_2^k}\rho^{\otimes k}-\tr\overline{\tilde{Q}_3^k}\rho^{\otimes k}\ge 1-3\epsilon.
	\end{equation}
	Now we employ the argument that, already at this point, we can find $(\lambda,\mu,\nu,\sigma,\alpha,\beta,\gamma)_k$ such that the projectors $Q_i^k, \ i=1,2,3$ corresponding to the different decompositions of $([\alpha]\otimes[\beta]\otimes[\gamma])_k$ fulfill the inequality
	\begin{equation}
		\left|\tr Q^k_1Q^k_2Q^k_3\rho^{\otimes k}\right|\ge\frac{1}{p(k)}
	\end{equation}
	for some polynomial $p$. We can use this inequality to bound the operator norm of $\tr Q^k_1Q^k_2Q^k_3$ and cyclic permutations thereof, and, in particular, the three different recoupling coefficients:
	\begin{eqnarray}
		\left|\tr Q^k_1Q^k_2Q^k_3\rho^{\otimes k}\right|&\le&\left\|Q^k_1Q^k_2Q_3^k\right\|_\infty\le\min\left\{\left\|Q^k_1Q^k_2\right\|_\infty,\left\|Q^k_2Q^k_3\right\|_\infty\right\},\nonumber\\
		\left|\tr Q^k_1Q^k_2Q^k_3\rho^{\otimes k}\right|&\le&\left|\tr Q^k_1Q^k_2Q^k_3\right|=\left|\tr Q^k_2Q^k_3Q^k_1\right|\le\left(\tr Q^k_1\right)\left\|Q^k_2Q^k_3Q^k_1\right\|_\infty\nonumber\\ &\le&\left(\tr Q^k_1\right)\left\|Q^k_3Q^k_1\right\|_\infty.
	\end{eqnarray}
	Analogously,
	\begin{equation}
		\left|\tr Q^k_1Q^k_2Q^k_3\rho^{\otimes k}\right|\le \left(\tr Q^k_2\right)\left\|Q^k_3Q^k_1Q^k_2\right\|_\infty.
	\end{equation}
	Together with \eqref{projrecoup} and the fact that all involved irreducible representations have a dimension at most polynomial in $k$, we get the assertion using an argument completely analogous to the one in the proof sketch above.
	
\end{proof}


\bibliographystyle{plain}

\newpage

 \includegraphics[width=0.9\textwidth]{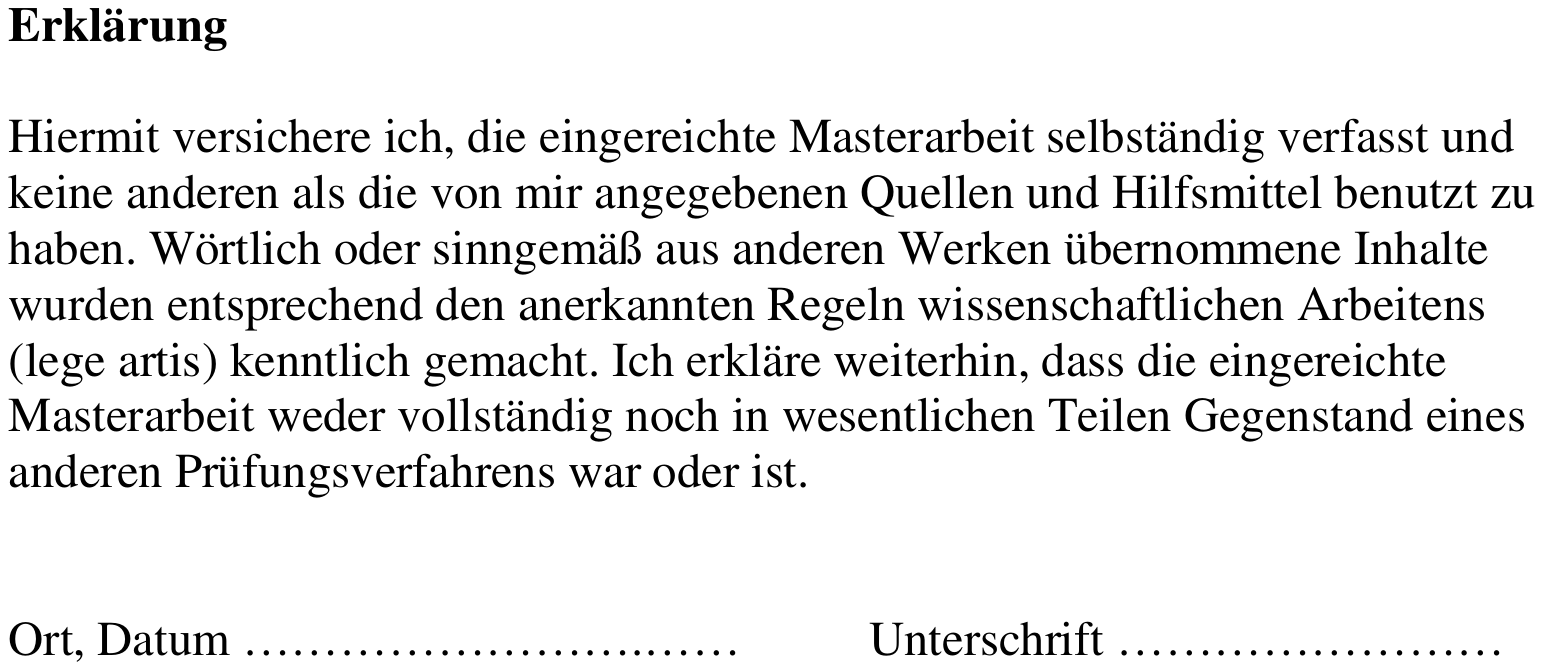}

 \newpage
\end{document}